\documentclass[peerreview]{IEEEtran}

\usepackage[utf8]{inputenc} 
\usepackage[T1]{fontenc}
\usepackage{url}
\usepackage{ifthen}
\usepackage{cite}
\usepackage[cmex10]{amsmath} % Use the [cmex10] option to ensure complicance
\usepackage{amsthm,amssymb,array,psfrag,stmaryrd,wasysym,cite,bbm}
\usepackage{epic}
\usepackage{pgf}
\usepackage{tikz}
\usetikzlibrary{patterns,arrows,automata,shapes.symbols,positioning,calc,fit,shapes.multipart,chains}
\usepackage{cleveref}
\usepackage[english]{babel}
\usepackage[autostyle]{csquotes}

\ifCLASSOPTIONcompsoc
  \usepackage[caption=false,font=normalsize,labelfont=sf,textfont=sf]{subfig}
\else
  \usepackage[caption=false,font=footnotesize]{subfig}
\fi
\usepackage{graphicx}
\usepackage{ulem}
\usepackage{caption}

\newtheorem{thm}{Theorem}
\newtheorem{corollary}{Corollary}
\newtheorem{proposition}{Proposition}
\newtheorem{definition}{Definition}

\newtheorem{lemma}{Lemma}
\newtheorem{lemma*}{Lemma}
\newtheorem{remark}{Remark}

\newcommand{\openone}{\leavevmode\hbox{\small1\normalsize\kern-.33em1}}
\newcommand{\Prob}{\mathbb{P}}
\newcommand {\E} {\mathbb{E}}

\newcommand{\lp}{\left(}
\newcommand{\rp}{\right)}
\newcommand{\nn}{\nonumber\\}

\begin{document}
\title{Optimal Age over Erasure Channels} 

% %%% Single author, or several authors with same affiliation:
% \author{%
%   \IEEEauthorblockN{Stefan M.~Moser}
%   \IEEEauthorblockA{ETH Zürich\\
%                     ISI (D-ITET)\\
%                     CH-8092 Zürich, Switzerland\\
%                     Email: moser@isi.ee.ethz.ch}
% }

%%% Several authors with up to three affiliations:
\author{%
%  \IEEEauthorblockN{Elie Najm, Emre Telatar, and Rajai Nasser}
Elie Najm, Emre Telatar, and Rajai Nasser
\thanks{This paper was presented in part at the IEEE International Symposium on Information Theory, Paris, July 2019.}
}

\maketitle

%%%%%%
%% Abstract: 
%% If your paper is eligible for the student paper award, please add
%% the comment "THIS PAPER IS ELIGIBLE FOR THE STUDENT PAPER
%% AWARD." as a first line in the abstract. 
%% For the final version of the accepted paper, please do not forget
%% to remove this comment!
%%
\normalem
\begin{abstract}
Previous works on age of information and erasure channels have dealt with specific models and computed the average age or
average peak age for certain settings. In this paper, %we ask a more fundamental question: 
given a source that produces a letter
every $T_s$ seconds and an erasure channel that can be used every $T_c$ seconds, we ask what is the coding
strategy that minimizes the time-average \enquote{age of information} that an observer of the channel output incurs. We first analyze the case where the source alphabet and the channel-input alphabet have the same size. We show that a trivial coding
strategy is optimal and a closed form expression for the age can be derived. We then analyze the case where the alphabets have different sizes. We use a random coding argument to
bound the average age and show that the average age achieved using random codes converges to the optimal average age of linear block codes as the 
source alphabet becomes large.

%Previous works on age of information and erasure channels have dealt with specific models and computed the average age or
%average peak age for certain settings. In this paper, we ask a more fundamental question: Given a single source and an erasure
%channel, what is the optimal coding scheme from an age point of view and what is the optimal achievable age? We answer these
%questions in the following two scenarios: $(i)$ the source alphabet and the erasure-channel input-alphabet are  the same, and
%$(ii)$ the source alphabet and the channel input-alphabet are different. We show that, in the first case, no coding is
%required and we derive a closed form formula for the average age. Whereas, in the second case, we use a random coding argument
%to bound the average age and show that the average age achieved using random codes converges to the optimal average age as the
%source alphabet becomes large.
\end{abstract}

%% The paper must be self-contained. However, if you are referring to
%% a full version for checking certain proofs, please provide the
%% publically accessible location below.  If the paper is completely
%% self-contained, you can remove the following line from your
%% submission.
%\textit{A full version of this paper is accessible at:}
%\url{http://isit2019.fr/} 

\section{Introduction}
\label{sec:sec_ch8_intro}
The concept of \emph{age} as a performance metric in communication systems was first used in 2011 by Kaul et. al in
\cite{KaulYatesGruteser-Globcom2011,kaul_minimizing_2011}, in order to assess
the performance of a given vehicular network. Vehicular networks are part of the growing group of real-time status-monitoring systems
that are used also in healthcare, finance, transportation, smart homes, warehouse and natural environment surveillance, to name but
a few. In such systems, a remote monitor is interested in the status of one or multiple processes. A sender takes samples of
the observed processes and sends them to the monitor. However, the aim of the communication system in this case is not to
transmit as fast as possible but to keep the information that the destination has about the observed processes as \emph{fresh} as possible.
%For instance, take the case of a monitor interested in the status of
%a distant process. A sender observes this process and, to keep the monitor up-to-date, sends updates to it. However, if, 
Indeed, if, at any
time $t$, the last received update at the monitor was generated at time $u(t)$, then the information at the receiver reflects the
status of the observed process at time $u(t)$, not at time $t$. Hence, the monitor has a distorted version of reality. In fact, it has an
obsolete version with an age of $\Delta(t)=t-u(t)$.

Kaul et al. in \cite{KaulYatesGruteser-2012Infocom} use a graphical method to compute and minimize an age-related metric: the
\emph{average age}. This metric is defined as
\begin{equation}
	\label{eq:eq_avg_age_def}
	\Delta = \lim_{\tau\to\infty}\frac{1}{\tau}\int_0^\tau\Delta(t)\mathrm{d}t.
\end{equation}
A growing body of works has used this metric to evaluate the
performance of multiple communication systems represented using queuing models
\cite{YatesKaul-2012ISIT,KamKompellaEphremides2013ISIT,KamKompellaEphremides2014ISIT,CostaCodreanuEphremides2014ISIT,UpdateorWait-2016Infocom,
NajmNasser-ISIT2016,BedewySunShroff17,NajmTelatar2018,NajmNasserTelatar-ISIT2018,InoueEtAl}; some of them
being subject to resource allocation constraints (such as energy
\cite{Elif-2015ITA,Yates-2015ISIT,ArafaUlukus17,BacinogluUysal-ISIT17,BacinogluSunUysalMutlu-ISIT18}). For an excellent recent survey about age of information, see \cite{YatesEtAlJSACSurvey}.
The works previously cited have mostly focused on computing the average age (AoI) given a
certain status updating policy while assuming no errors. At a more physical level, the effect of noise and channel coding on the average age was also
investigated, especially when the erasure channel is used \cite{RefR1Reviewer1,RefR2Reviewer1,RefR3Reviewer1,RefR4Reviewer1,RefR5Reviewer1,RefR6Reviewer1,RefR7Reviewer1,
RefR8Reviewer1,ChenHuang-ISIT2016,ParagTaghaviChamberland17,NajmYatesSoljanin-ISIT2017,RoyNajmSoljaninZhong-ISIT2017}. Chen et al. in \cite{ChenHuang-ISIT2016} assume a
random service time but the transmitted packet has a certain probability of being lost at the end of the transmission. Parag et
al. in \cite{ParagTaghaviChamberland17} consider the binary erasure channel (BEC) and compute the average age for two
transmission schemes: single transmission and hybrid automatic repeat request (HARQ). While the authors in
\cite{ParagTaghaviChamberland17} assume a just-in-time generation process, Najm et al. in \cite{NajmYatesSoljanin-ISIT2017}
consider a Poisson process generation and two HARQ protocols to combat erasures: infinite incremental redundancy (IIR) and fixed
redundancy (FR). Yates et al. in \cite{RoyNajmSoljaninZhong-ISIT2017}  consider the previous two schemes, IIR and FR, but
assume a just-in-time generation policy. While both papers, \cite{NajmYatesSoljanin-ISIT2017} and
\cite{RoyNajmSoljaninZhong-ISIT2017}, agree on the definition of IIR, they use the term FR to describe two different schemes.
In  \cite{NajmYatesSoljanin-ISIT2017}, the update is divided into $k_p$ packets encoded ratelessly and each packet is encoded
using  an $(n_s,k_s)$-\emph{maximum distance separable} (MDS) code. In \cite{RoyNajmSoljaninZhong-ISIT2017}, FR means that each
$k$-bit update is encoded into an $n$-bit codeword, and the update is successfully received if and only if at least $k$ bits are not
erased. We are interested here in transmission schemes similar to the FR that is considered in \cite{RoyNajmSoljaninZhong-ISIT2017}.

Most of the works that addressed the presence of noise and erasure have assumed some form of feedback from the receiver to the transmitter. In this paper, we take an information-theoretic approach to the age problem and provide a
characterization of the optimal achievable age when the channel used is the erasure channel and no feedback is assumed. This means that we consider the
following question: Given a $q$-ary erasure channel without feedback and with input alphabet $\mathcal{V}$ and a source with alphabet $\mathcal{U}$, what
is the lowest average age that can be achieved in this system?

Since the channel can introduce errors and since no feedback is available, we will be forced (at least in some cases) to use some form of coding. However, unlike classical communication systems where the primary role of coding is to guarantee reliable communication of \emph{all} packets, in our system, we do not care too much if some packets are not delivered. The primary goal of our coding is that the monitor reliably receives enough \emph{timely} packets so that it remains up-to-date as much as possible. 

In order to study the age of information over erasure channels, we distinguish two cases: 
\begin{itemize}
	\item \textbf{Case 1}: The source alphabet and the channel-input alphabet are of the same size.
	\item \textbf{Case 2}: The source alphabet and the channel-input alphabet have different sizes.
\end{itemize}
For the first case we derive an exact closed-form expression for the average age and show that the optimal average age is
achieved without any encoding done on the source symbols. Whereas for the second case, encoding is mandatory and we use random
coding to give an upper and lower bounds on the achievable average age of the system, as well as an approximation of the lower
bound inspired by \cite{RoyNajmSoljaninZhong-ISIT2017,YatesNSZ17}.

The rest of this paper is organized as follows: In Section~\ref{sec:sec_ch8_preliminaries}, we present the system model and
some definitions which are common to all later sections. In Section~\ref{sec:sec_ch8_same_alphabet}, we derive the optimal average
age for Case 1 and in Section~\ref{sec:sec_ch8_diff_alphabets} we study the optimal average age
for Case 2.

\section{Preliminaries}
\label{sec:sec_ch8_preliminaries}
\begin{figure*}[t]
\centering
\begin{tikzpicture}[scale=0.4,font=\scriptsize]
\node[circle,draw,align=center] (S) at (0,0) {Information\\source};
\node[rectangle,draw] (T) [right of=S, node distance=3.2cm] {Encoder};
\node[rectangle,draw,node distance=2.9cm,align=center] (C) [right of=T] {Erasure\\Channel};
\node[rectangle,draw,node distance=2.9cm] (R) [right of=C] {Decoder};
\node[circle,draw,node distance=4.2cm] (D) [right of=R] {Destination};
\draw[->,thick] (S) -- node [anchor=north] {$U_1U_2\cdots U_{m}$} (T);
\draw[->,thick] (T) -- node [anchor=north,align=center] {$V_1V_2\cdots V_n$} (C);
\draw[->,thick] (C) -- node [anchor=north,align=center] {$Z_1Z_2\cdots Z_n$} (R);
\draw[->,thick] (R) -- node [anchor=north,align=center] {${(\hat{U}_{i_1},i_1),\ldots,(\hat{U}_{i_n},i_n)}$} (D);
\draw[dotted,thick] (-2.1,-2.1) rectangle (9.5,2.1);
\node at (3.7,2.4) {Sender};
\draw[dotted,thick] (21,-2.1) rectangle (35,2.1);
\node at (26.5,2.4) {Receiver};
\end{tikzpicture}
\caption{The communication system.}
\label{fig:fig_ch8_communication_system}
\end{figure*}
We start by defining the communication system that we study.  Fig.~\ref{fig:fig_ch8_communication_system} illustrates such a
system. %The following discussion is based on the definitions presented in Section~\ref{sec:sec_ch1_classic_problem}:
\begin{itemize}
	\item \emph{The channel}: We {consider} a discrete memoryless {$q$-ary} erasure channel with erasure probability $\epsilon$. We refer to such channel by {$q$}EC($\epsilon$). The
		channel-input alphabet is given by $\mathcal{V}=\{0,1,\ldots,q-1\}$, and the channel-output alphabet by
		$\mathcal{V}\cup\{?\} = \{0,1,\ldots,q-1,?\}$. We also assume {that} there is no feedback from the
		receiver. This means that the output of the encoder depends only on the source symbols
		and the sender does not know whether a sent symbol was successfully received or not. In addition to that, we
		assume that transmitted channel-symbols are received instantaneously\footnote{If the transmitted channel-symbols are not received instantaneously but are received after a delay that is constant, then this constant delay can be added to all the age expressions that are derived in this paper.}. Furthermore, there exists a period $T_c$ between two consecutive
		channel uses. {More precisely, the $i^{th}$ channel-use takes place at time $t_i^c=iT_c$. Note that we assumed without loss of generality that $t_0^c=0$. We} define the \emph{channel-use rate} $\mu= \frac{1}{T_c}$ to be the allowed number of channel uses
		per {second}. %Notice that in this case $\mu$ is a fixed number.
	\item \emph{The source}: We assume a single discrete memoryless source generating messages that belong to the set $\mathcal{U}
		= \{1,2,\ldots,L\}$. So each symbol in this set is a message and we will use interchangeably the terms
		\emph{source symbol} and \emph{message} in this paper. We define $k=\left\lceil \log_q(L) \right\rceil = \left\lceil
		\frac{\ln(L)}{\ln(q)}\right\rceil$ where $\log_q$ is the base-$q$ logarithm. Hence, in order
		to represent one source symbol we need $k$ channel-input symbols. This means that there exists an injective
		function $h(.)$ that maps every message $m\in\mathcal{U}$ to a length-$k$ sequence $u^k=(u_1,\ldots,u_k)\in\mathcal{V}^k$, with
		$u_j\in\mathcal{V}$ for $1\leq j\leq k$. Thus, $h(\mathcal{U})\subseteq\mathcal{V}^k$. Similar to the channel-use case, the
		source symbol generation is assumed to be periodic with period $T_s$. More precisely, the $m^{th}$ source-symbol is generated at time $t_m^s=mT_s+t_0^s$. Note that since the source and channel clocks might not be synchronized, we need to take into account the possibility that \emph{the starting time of the source} $t_0^s$ is nonzero. We define the \emph{message generation
		rate} $\lambda=\frac{1}{T_s}$ as the fixed number of source symbols generated per second. 
		
		Notice that if the source alphabet and the channel-input alphabet have the same size, then
		$h(\mathcal{U})=\mathcal{V}$ and $k=1$. In this case, we can assume without loss of generality that $\mathcal{U}=\mathcal{V}$. This is the system that we study in Section \ref{sec:sec_ch8_same_alphabet}.
		
		In the case where the
		source alphabet and channel-input alphabet have different sizes, we focus on strategies
induced by linear codes, {so we will assume that $q$ is a power of a prime number, $\mathcal{V}=\mathbb{F}_q$, and $\mathcal{U}=\mathcal{V}^k=\mathbb{F}_q^k$. This is the system that we study in Section \ref{sec:sec_ch8_diff_alphabets}.}
	\item \emph{The encoder and decoder}: At the $i^{th}$ channel use, the encoder uses all the generated source symbols
		and encodes them into a single channel-input letter, {i.e., the encoder is a function}
\begin{equation}
{f_i:\mathcal{U}^{\big\lfloor\frac{iT_c - t_0^s}{T_s}\big\rfloor}\to\mathcal{V}}.
\end{equation}		
		
The decoder, at the ${i^{th}}$ channel use,
		uses all {the} received channel-output symbols to compute an estimate of {a transmitted message}, along with its index. Thus, {the decoder is a function} 
\begin{equation}
{g_i:(\mathcal{V}\cup\{?\})^i\to
		\Big(\mathcal{U}\times\Big\{1,2,\ldots,\Big\lfloor\frac{iT_c - t_0^s}{T_s}\Big\rfloor\Big\}\Big)\cup\{\textsc{erasure}\}.}
\end{equation}		

		We assume that the decoder never makes mistakes. In other words, if the generated source symbols are $U_1,\ldots,U_{\big\lfloor\frac{iT_c - t_0^s}{T_s}\bigcap\rfloor}$, the channel-output symbols are $Z_1,\ldots,Z_i$, and $g_i(Z_1,\ldots,Z_i)=(\hat{U}_m,m)$, then we have $\hat{U}_m=U_m$ with probability 1. In this case, the age of information at time $t\in\big[iT_c,(i+1)T_c\big)$ is equal to
		\begin{equation}
		\label{eq:eq_Instantaneous_age_oi_def}
\Delta_{\mathcal{C}}(t)=t-mT_s-t_0^s.
		\end{equation}
		
		It is easy to see that for a given sequence of encoders $(f_i)_{i\geq 1}$, the optimal decoder (from age perspective) is the one defined as $g_i(Z_1,\ldots,Z_i)=(\hat{U}_{m_i},m_i)$, where $m_i$ is the maximum index in $\left\{1,\ldots,\left\lfloor\frac{iT_c - t_0^s}{T_s}\right\rfloor\right\}$ such that $U_{m_i}$ can be deterministically decoded from $(Z_1,\ldots,Z_i)$, and $\hat{U}_{m_i}=U_{m_i}$. If no such $m_i\in \left\{1,\ldots,\left\lfloor\frac{iT_c - t_0^s}{T_s}\right\rfloor\right\}$ exists, we have $g_i(Z_1,\ldots,Z_i)=\textsc{erasure}$ and we adopt the convention that $m_i=0$ for such cases. By noticing that for every $t\geq 0$ we have $t\in\big[iT_c,(i+1)T_c\big)$ where $i=\lfloor\frac{t}{T_c}\rfloor$, we can see from \eqref{eq:eq_Instantaneous_age_oi_def} that the instantaneous age of information of the coding scheme $\mathcal{C}=(f_i,g_i)_{i\geq1}$ is given by
		\begin{equation}
		\label{eq:eq_Instantaneous_age_oi}
\Delta_{\mathcal{C}}(t)=t-m_{\lfloor\frac{t}{T_c}\rfloor}T_s-t_0^s.
		\end{equation}
		It is worth noting that the function $t\mapsto m_{\lfloor\frac{t}{T_c}\rfloor}$ is nondecreasing and piecewise constant. Furthermore, the discontinuities in the function $t\mapsto m_{\lfloor\frac{t}{T_c}\rfloor}$ correspond to instants at which the receiver successfully decodes new packets. From this and from \eqref{eq:eq_Instantaneous_age_oi}, we can see that the instantaneous age $\Delta_{\mathcal{C}}(t)$ is a piecewise linearly-increasing function that has a sawtooth shape.
		
		In Fig.~\ref{fig:fig_instantaneous_age_general}, we show an example illustrating how the the instantaneous age varies with time. In this figure, $\tilde{t}_{\ell}$ represents the instant at which the $\ell^{th}$ successfully received message was decoded at the receiver, and $t_\ell$ represents the generation time of this message at the source. More precisely, $\tilde{t}_{\ell}=i_{\ell}T_c$ where $i_{\ell}$ is the index at which $(m_i)_{i\geq 1}$ changes its value for the $\ell^{th}$ time, and $t_{\ell}=m_{i_\ell}T_s+t_0^s$.
		
			%\item \emph{The encoder and decoder}: The encoder and decoder are two functions, $f:\mathcal{V}^k\to\mathcal{V}^n$ and
	%	$g:(\mathcal{V}\cup\{\text{?}\})^n\to\mathcal{V}^k\cup\{\text{erasure}\}$, respectively.
\end{itemize}

\begin{figure*}[t]
	\centering
	\begin{tikzpicture}[scale=0.7,font=\scriptsize]
		\draw[thick,->] (-4,0) -- (13.5,0) node[anchor=north] {$t$};

		\draw[color=blue] 	(1,0) node [anchor=north] {$t_1$};
		\draw[color=red] 	(2.6,0) node [anchor=north] {$\tilde{t}_1$};
		\draw[color=blue] 	(5,0) node [anchor=north] {$t_2$};
		\draw[color=red] 	(6.2,0) node [anchor=north] {$\tilde{t}_2$};
		\draw[color=blue] 	(7,0) node [anchor=north] {$t_3$};
		\draw[color=red] 	(7.8,0) node [anchor=north] {$\tilde{t}_3$};
		\draw[color=blue] 	(11,0) node [anchor=north] {$t_4$};
		\draw[color=red] 	(12.6,0) node [anchor=north] {$\tilde{t}_4$};
		\draw (-3,0) node [anchor=north] {$t_0^s$};
		
		% vertical axis
		\draw[thick,->] (-1.4,0) node[anchor=north] {$0$}-- (-1.4,6) node[anchor=south] {$\Delta_{\mathcal{C}}(t)$};
		\draw[dashed] (-3,0) -- (-1.4,1.6);
		
		\draw[thick] (-1.4,1.6) -- (2.6,5.6) -- (2.6,1.6) -- (6.2,5.2) -- (6.2,1.2) -- (7.8,2.8) -- (7.8,0.8) --
		(12.6,5.6) -- (12.6,1.6) -- (13,2);
		
		\draw[dashed,color=blue] (1,0) -- (2.6,1.6);
		\draw[dashed,color=blue] (5,0) -- (6.2,1.2);
		\draw[dashed,color=blue] (7,0) -- (7.8,0.8);
		\draw[dashed,color=blue] (11,0) -- (12.6,1.6);

		\draw[dotted,color=red] (2.6,1.6) -- (2.6,0);
		\draw[dotted,color=red] (6.2,1.2) -- (6.2,0);
		\draw[dotted,color=red] (7.8,0.8) -- (7.8,0);
		\draw[dotted,color=red] (12.6,1.6) --(12.6,0);
	\end{tikzpicture}
	\caption{The instantaneous age $\Delta_{\mathcal{C}}(t)$.}
	\label{fig:fig_instantaneous_age_general}
\end{figure*}

In the previous section, we indicated that we are interested in bounding the optimal achievable average age. Here, we formally define
the concepts of achievable age and optimal achievable age.
\begin{definition}
	\label{def:def_ch8_coding_scheme}
	We call $\mathcal{C}=(f_i,g_i)_{i\geq 1}$ to be a coding scheme where $(f_i)_{i\geq 1}$ is the sequence of encoders and $(g_i)_{i\geq 1}$ is the sequence of decoders. The average age corresponding to such scheme is
	denoted by 
	\begin{equation}
		%\label{eq:eq_ch8_avg_age_coding_scheme}
		\Delta_{\mathcal{C}} = \lim_{\tau\to\infty}\frac{1}{\tau}\int_0^\tau \Delta_{\mathcal{C}}(t) \mathrm{d}t,
	\end{equation}
	where $\Delta_{\mathcal{C}}(t)$ is the instantaneous age that is obtained by using the coding scheme $\mathcal{C}$, and which is given by \eqref{eq:eq_Instantaneous_age_oi_def}. If the decoders $(g_i)_{i\geq 1}$ are optimal for the encoders $(f_i)_{i\geq 1}$ then $\Delta_{\mathcal{C}}(t)$ is given by \eqref{eq:eq_Instantaneous_age_oi}.\\
	Such a definition {can be generalized to channels other than erasure channels}. However, for the special case of the erasure channel
	with erasure probability $\epsilon$, the average age relative to the coding scheme $\mathcal{C}$ will be denoted by 
	\begin{equation}
		\label{eq:eq_ch8_avg_age_coding_scheme}
		\Delta_{\epsilon,\mathcal{C}} = \lim_{\tau\to\infty}\frac{1}{\tau}\int_0^\tau \Delta_{\epsilon,\mathcal{C}}(t) \mathrm{d}t,
	\end{equation}
	where $\Delta_{\epsilon,\mathcal{C}}(t)$ is the instantaneous age {that is obtained by using the coding scheme $\mathcal{C}$ when the channel is $q$EC($\epsilon$)}.
\end{definition}
\begin{definition}
	\label{def:def_ch8_achievable_AoI}
	We say that an age $D$ is achievable for {$q$EC($\epsilon$)}, if for every $\delta>0$ there exists a coding scheme $\mathcal{C}=(f_i,g_i)_{i\geq 1}$ %\footnote{Given that the source alphabet is
		%usually fixed, the only variables left to tune in any coding scheme are the blocklength, the encoder and the
	%decoder.} 
		such that 
	\begin{equation}
		\label{eq:eq_ch8_achievable_avg_age}
		\Delta_{\epsilon,\mathcal{C}} \leq D+\delta,
	\end{equation}
	and the probability of error on the decoded messages is zero.
\end{definition}
\begin{definition}
	\label{def:def_ch8_optimal_achievable_AoI}
	Given a channel {$q$EC($\epsilon$)}, we define the optimal average age $\Delta_\epsilon$ to be the minimum achievable
	average age. Formally,
	\begin{equation}
		\label{eq:eq_ch8_optimal_avg_age}
		\Delta_\epsilon = \inf_{\mathcal{C}\in\Gamma} \Delta_{\epsilon,\mathcal{C}},
	\end{equation}
	where $\Gamma$ is the set of all possible coding schemes. \\
	The set $\mathcal{R}=\left\{(\epsilon,D); D\geq\Delta_\epsilon\text{ and }\epsilon\in[0,1]\right\}$ forms the set of achievable
	average ages over all erasure channels.
\end{definition}

%%%%%%%%%%%%%%%%%%%%%%%%%%%%%%%%%%%%%%%%%%%%%%%%%%%%%%%%%%%%%%%%%%%%%%%%%%%%%%%%%%%
\section{Optimal Age with the Same Source \& Channel Alphabets}
\label{sec:sec_ch8_same_alphabet}
In this first case, we take $k=1$ which means that the source and channel-input alphabets are the same. We first show that to
achieve the optimal age, no encoding is required and we provide the optimal transmission policy. We then compute the optimal
average age.

\subsection{The Optimal Transmission Policy}
\label{subsec:subsec_ch8_opt_transmission_policy}

%\begin{thm}
%	\label{thm:thm_ch8_enc_not_optimal}
%{\color{red}\sout{	For a channel EC($\epsilon$), if the source alphabet and the channel-input alphabet are the same, then in order to minimize the average age no
%	encoding is required.}}
%\end{thm}
%\begin{proof}
%{\color{red}\sout{	If an oracle were to give the erasure pattern to the transmitter then the optimal thing to do from an age perspective
%	is to send the newest source symbol at the non-erased channel-uses since each message needs only one channel use to be
%	transmitted. This means that at every channel use the sender is sending the freshest information. If instead we encode
%	the messages using a coding scheme $\mathcal{C}$ into codewords of length $n>1$ then each message will need strictly more
%	than one channel use to be transmitted, hence, at any instant $t$ and for an arbitrary erasure pattern, the instantaneous age would be larger than the one that corresponds to
%	the uncoded messages, i.e., $\Delta_{\epsilon,\mathcal{C}}(t)\geq\Delta_{\epsilon,uncoded}(t)$. Therefore,
%	$\Delta_{\epsilon,\mathcal{C}}\geq\Delta_{\epsilon,uncoded}$. This is so since with encoded messages, the transmitter
%	is not necessarily sending the freshest information at every channel use.}}
%\end{proof}

\begin{thm}
	\label{thm:thm_ch8_optimal_policy}
	For a channel {$q$EC($\epsilon$)}, if the source alphabet and the channel-input alphabet are the same, then the optimal transmission policy from an age perspective is to keep transmitting the last-generated source-symbol until a new
	one is generated{, at which point we start transmitting the newly generated source-symbol} and discard all previous messages. This is {an} LCFS with no buffer policy.
\end{thm}
\begin{proof}
	{Let us} assume that an oracle provides us with the erasure pattern. It is clear that at each non-erased channel use we should send the latest
update so that the drop in the instantaneous age is the most important. Indeed, if there is a non-erased channel use at time
$t'$ and the latest update is generated at $t_{last}$ then the instantaneous age, $\Delta_{opt}(t)$, that corresponds to the LCFS
with no buffer policy drops to $\Delta_{opt}(t)=t'-t_{last}$. {We cannot do better than this because there is no source-symbol that is generated after $t_{last}$.} This argument shows that the optimal transmission policy would send the latest generated
source symbol at every non-erased channel use while it can transmit anything at the erased channel uses. However, since in
practice the transmitter {does} not have access to the erasure pattern beforehand, the policy that consists of keeping on
transmitting the last generated update until a new one is created satisfies the optimality criterion that is to send the latest
generated message at each non-erased channel use.

For the case where $T_c\leq T_s$ or $\mu\geq\lambda$, the LCFS with no buffer policy leads to the transmission of all source
symbols at least once. Whereas for the case of $T_c> T_s$ or $\mu<\lambda$, some messages will be dropped and will never
be sent.
\end{proof}

\subsection{The Optimal Average Age}
\label{subsec:subsec_ch8_same_alph_opt_avg_age}
\begin{thm}
	\label{thm:thm_ch8_opt_avg_age}
	{Given a source with message-generation rate $\lambda$ and starting time $t_0^s$, an erasure channel $q$EC($\epsilon$) with channel-use rate $\mu$, 
	and utilization $\rho=\frac{\lambda}{\mu}$,} the optimal average age achieved over {$q$EC($\epsilon$)} is:
	\begin{itemize}
		\item For irrational utilization $\rho\in\mathbb{R}\setminus\mathbb{Q}$,
			\begin{equation}
				\label{eq:eq_ch8_opt_avg_age_irrational_rho}
				\Delta_\epsilon = \frac{1}{2\lambda}+\frac{1+\epsilon}{2\mu(1-\epsilon)}.
			\end{equation}
		\item For rational utilization $\rho\in\mathbb{Q}$, 
			\begin{equation}
				\label{eq:eq_ch8_opt_avg_age_rational_rho}
				\Delta_\epsilon =\frac{1}{2\lambda} + \frac{2\left[-d\lambda t_0^s\right]-1}{2d\lambda}+\frac{1+\epsilon}{2\mu(1-\epsilon)},
			\end{equation}
			where $[x]=x-\lfloor x\rfloor$ is the fractional part of $x$, and $d$ is the denominator of $\rho=\frac{\lambda}{\mu}$ when it is written as a rational fraction of integers in irreducible form, i.e., $\rho=\frac{c}{d}$ with
			$c,d\in\mathbb{N}$ and $\gcd(c,d)=1$.
	\end{itemize}
\end{thm}
Before giving the proof of \Cref{thm:thm_ch8_opt_avg_age}, we need the following lemmas.
\begin{lemma}
	\label{lemma:lemma_ch8_weyl_equidistribution}
	{
	Let $(X_l)_{l\geq 1}$ be a sequence of independent and identically distributed nondeterministic\footnote{A random variable $X\in\mathbb{N}^*$ is nondeterministic if there are at least two different integers with nonzero probability. It is worth noting that \Cref{lemma:lemma_ch8_weyl_equidistribution} remains true if $(X_l)_{l\geq 1}$ are deterministic, but in this paper we are only interested in the case where $(X_l)_{l\geq 1}$ are nondeterministic.} random variables which take values in the set of strictly positive natural numbers $\mathbb{N}^\ast$ and which satisfy $\E(X_l^2)<\infty$. Let $S_0=0$ and $\displaystyle S_l=\sum_{r=1}^lX_r$ for $l\geq 1$. For every $i\geq 0$, let 
\begin{equation}
Y_i=\max\left\{S_l: l\geq 0\text{ and }S_l\leq i\right\}.
\end{equation}

Let $\rho\in\mathbb{R}\setminus\mathbb{Q}$ be an irrational number, and let $\alpha\in\mathbb{R}$ be an arbitrary real number. Then, almost surely, we have
		\begin{equation}
%		\label{eq:eq_ch8_weyl_equidistribution}
		\lim_{N\to\infty} \frac{1}{N}\sum_{i=1}^{N} \left[\rho Y_i+\alpha\right] = \frac{1}{2},
	\end{equation}
	where $[x]=x-\lfloor x\rfloor$ is the fractional part of $x$.
	}
\end{lemma}
\begin{proof}
	This lemma is a consequence of Weyl's equidistribution theorem \cite{Weyl1916} (see \Cref{subsec:subsec_ch8_weyl_thm}). A full proof of this lemma can be found in \Cref{subsec:subsec_ch8_app_weyl_equi_proof}.
\end{proof}

	{
\begin{lemma}
	\label{lemma:lemma_ch8_equidistribution_rational}

Let $c,d\in\mathbb{Z}$ be such that $d>0$ and $\gcd(c,d)=1$. Then, for every $\alpha\in\mathbb{R}$, we have
		\begin{equation}
		\label{eq:eq_ch8_weyl_equidistribution_rational}
		\sum_{a=0}^{d-1} \left[\frac{c}{d} a+\alpha\right] = \frac{d-1}{2}+[d\alpha],
	\end{equation}
	where $[x]=x-\lfloor x\rfloor$ is the fractional part of $x$.
	
\end{lemma}
\begin{proof}
Since $\gcd(c,d)=1$, then for every $r\in\mathbb{Z}$, the mapping $a\mapsto (ca+r\bmod d)$ is a bijection from $\{0,\ldots,d-1\}$ to itself. Therefore,
\begin{align}
\sum_{a=0}^{d-1} \left[\frac{c}{d} a+\alpha\right]
&= \sum_{a=0}^{d-1} \left[\frac{ca+d\alpha}{d} \right] = \sum_{a=0}^{d-1} \left[\frac{ca+\lfloor d\alpha\rfloor + [d\alpha]}{d} \right]\nonumber \\
&= \sum_{a=0}^{d-1} \left[\frac{\left(ca+\lfloor d\alpha\rfloor \bmod d\right) + [d\alpha]}{d} \right] \stackrel{(\ast)}{=} \sum_{b=0}^{d-1} \left[\frac{b+ [d\alpha]}{d} \right] \stackrel{(\dagger)}{=} \sum_{b=0}^{d-1} \frac{b+ [d\alpha]}{d}  = \frac{d-1}{2}+[d\alpha],
\end{align}
where $(\ast)$ follows from the fact that the mapping $a\mapsto (ca+\lfloor d \alpha\rfloor\bmod d)$ is a bijection from $\{0,\ldots,d-1\}$ to itself, and $(\dagger)$ follows from the fact that $b+[d\alpha]<d$ for every $b\leq d-1$.
\end{proof}
}

%\begin{lemma}
%	\label{lemma:lemma_ch8_equidistribution_rational}
%	{
%	Let $(X_i)_{i\geq 1}$ be a sequence of independent and identically distributed random variables which take values in the set of strictly positive natural numbers $\mathbb{N}^\ast$ and which satisfy $\E(X_i^2)<\infty$ and $\Prob(X_i=1)>0$. Let $S_0=0$ and $\displaystyle S_i=\sum_{l=1}^iX_l$ for $i\geq 1$. For every $n\geq 1$, let $$Y_n=\max\left\{S_i: i\geq 0, S_i\leq n\right\}.$$
%
%	Let $\rho\in\mathbb{Q}$ be an rational number, and let $\alpha\in\mathbb{R}$ be an arbitrary real number. Moreover, assume that $\rho$ can be
%	written in an irreducible form as $\rho =\frac{m}{\ell}$, where $m,\ell\in\mathbb{N}$, $\ell\neq0$ and $\gcd(m,\ell)=1$. Then, almost surely, we have
%		\begin{equation}
%		\label{eq:eq_ch8_weyl_equidistribution_rational}
%		\lim_{N\to\infty} \frac{1}{N}\sum_{n=1}^{N} \left[\rho Y_n+\alpha\right] = \frac{\ell-1}{2\ell}+\frac{[\ell\alpha]}\ell,
%	\end{equation}
%	where $[x]=x-\lfloor x\rfloor$ is the fractional part of $x$.
%	}
%\end{lemma}
%\begin{proof}
%	A full proof of this lemma is presented in \Cref{subsec:subsec_ch8_app_equi_rational_proof}.
%\end{proof}

\begin{proof}[Proof of \Cref{thm:thm_ch8_opt_avg_age}]
%	In this proof we will use a different approach than the ATA and DTA presented in \Cref{ch2}. 
	We know that
	$\mu=\frac{1}{T_c}$. In \eqref{eq:eq_avg_age_def} we saw that the average age is given by
	\begin{equation}
		\Delta_\epsilon = \lim_{\tau\to\infty} \frac{1}{\tau}\int_0^\tau \Delta_\epsilon(t)\mathrm{d}t.
	\end{equation}
	%Noticing that $\tau=\frac{\tau}{T_c}T_c=\lp \lfloor\frac{\tau}{T_c}\rfloor + \left[\frac{\tau}{T_c}\right]\rp T_c$, with $\left[\frac{\tau}{T_c}\right]$ being the fractional part of $\frac{\tau}{T_c}$, w
	We can rewrite the average age as
{
	
	\begin{align}
		\Delta_\epsilon &= \lim_{\tau\to\infty}
		\frac{1}{\tau}\lp\sum_{i=1}^{\lfloor\frac{\tau}{T_c}\rfloor}\int_{(i-1)T_c}^{iT_c}
		\Delta_\epsilon(t)\mathrm{d}t+ \int_{\lfloor\frac{\tau}{T_c}\rfloor T_c}^\tau\Delta_\epsilon(t)\mathrm{d}t\rp.
	\end{align}
	
	Therefore,	
	
		\begin{align}
		\label{eq:eq_ch8_app_same_alph_avg_age_1}
		\lim_{\tau\to\infty}
		\frac{1}{\tau}\sum_{i=1}^{\lfloor\frac{\tau}{T_c}\rfloor}\int_{(i-1)T_c}^{iT_c}
		\Delta_\epsilon(t)\mathrm{d}t
		\leq \Delta_\epsilon 
		\leq \lim_{\tau\to\infty}
		\frac{1}{\tau}\sum_{i=1}^{\lfloor\frac{\tau}{T_c}\rfloor+1}\int_{(i-1)T_c}^{iT_c}
		\Delta_\epsilon(t)\mathrm{d}t.
	\end{align}
	
	Let 
\begin{equation}
M_\tau=\left\lfloor\frac{\tau}{T_c}\right\rfloor\quad\text{and}\quad \Delta_{\epsilon,i} =
	\frac{1}{T_c}\int_{(i-1)T_c}^{iT_c}\Delta_\epsilon(t)\mathrm{d}t.
\end{equation}

	By noticing that 
\begin{equation}
\lim_{\tau\to\infty}\frac{M_\tau}{\frac{\tau}{T_c}}=\lim_{\tau\to\infty}\frac{M_\tau + 1}{\frac{\tau}{T_c}} = 1,
\end{equation}	
we can deduce from \eqref{eq:eq_ch8_app_same_alph_avg_age_1} that
\begin{equation}
\lim_{\tau\to\infty} \frac{1}{M_\tau}\sum_{i=1}^{M_\tau} \Delta_{\epsilon,i} \leq	\Delta_\epsilon \leq \lim_{\tau\to\infty} \frac{1}{M_\tau + 1}\sum_{i=1}^{M_\tau+1} \Delta_{\epsilon,i},
\end{equation}
	which implies that
		\begin{equation}
		\label{eq:eq_ch8_app_same_alph_avg_age_2}
\Delta_\epsilon = \lim_{N\to\infty} \frac{1}{N}\sum_{i=1}^{N} \Delta_{\epsilon,i}.
	\end{equation}
	}

{	At any instant $t\in\big[(i-1)T_c,iT_c\big)$, the last channel use took place at time $t_{i-1}^c=(i-1)T_c$. Assume that the last successful channel use before time $iT_c$ was the $(i-1-K_i)^{th}$ channel use, i.e., it took place at time $(i-1-K_i)T_c$. The source symbol that was transmitted at this time was generated at time $\left\lfloor\frac{(i-1-K_i)T_c-t_0^s}{T_s}\right\rfloor T_s+t_0^s$. Therefore, at any instant $t\in[(i-1)T_c,iT_c)$, the timestamp of the last successfully received source symbol is 
\begin{equation}
u(t)=\frac{1}{\lambda}\left\lfloor\frac{\lambda}{\mu}(i-1-K_i)-\lambda t_0^s\right\rfloor +t_0^s,
\end{equation}
which means that the age of information at time $t$ is equal to
\begin{equation}
\Delta_\epsilon(t) = t-u(t) = t- \frac{1}{\lambda}\left\lfloor\frac{\lambda}{\mu}(i-1-K_i)-\lambda t_0^s\right\rfloor -t_0^s.
\end{equation}
	
	Hence,
	\begin{align}
		\label{eq:eq_ch8_age_same_alph_proof_1}
		\Delta_{\epsilon,i}	&= \frac{1}{T_c}\int_{(i-1)T_c}^{iT_c} \Delta_\epsilon(t)\mathrm{d}t\nn
					&= \frac{1}{T_c}\int_{(i-1)T_c}^{iT_c} \lp t- \frac{1}{\lambda}\left\lfloor\frac{\lambda}{\mu}(i-1-K_i)-\lambda t_0^s\right\rfloor -t_0^s\rp\mathrm{d}t\nn
					&= \frac{1}{T_c} \lp
					\frac{i^2T_c^2}{2}-\frac{(i-1)^2T_c^2}{2}-\frac{T_c}{\lambda}\left\lfloor\frac{\lambda}{\mu}(i-1-K_i)-\lambda t_0^s\right\rfloor -t_0^sT_c\rp\nn
					&=  iT_c-\frac{T_c}{2}-\frac{1}{\lambda}\left\lfloor\frac{\lambda}{\mu}(i-1-K_i)-\lambda t_0^s\right\rfloor -t_0^s\nn
					&= \frac{i}{\mu} -\frac{1}{2\mu}-\frac{1}{\lambda}\left\lfloor\frac{\lambda}{\mu}(i-1-K_i)-\lambda t_0^s\right\rfloor -t_0^s\nn
					&= \frac{i}{\mu} -\frac{1}{2\mu}-\frac{1}{\lambda}\left(\frac{\lambda}{\mu}(i-1-K_i)-\lambda t_0^s\right)+\frac{1}{\lambda}\left(\frac{\lambda}{\mu}(i-1-K_i)-\lambda t_0^s-\left\lfloor\frac{\lambda}{\mu}(i-1-K_i)-\lambda t_0^s\right\rfloor\right) -t_0^s\nn	
					&= \frac{i}{\mu} -\frac{1}{2\mu}- \frac{i}{\mu}+ \frac{1}{\mu}+ \frac{K_i}{\mu}+ t_0^s+\frac{1}{\lambda}\left[\frac{\lambda}{\mu}(i-1-K_i)-\lambda t_0^s\right] -t_0^s\nn
					&= \frac{1}{2\mu}+ \frac{K_i}{\mu}+\frac{1}{\lambda}\left[\frac{\lambda}{\mu}(i-1-K_i)-\lambda t_0^s\right],
	\end{align}
	}
	where $[x]=x-\lfloor x\rfloor$ is the fractional part of $x$.

	Setting $K_1=0$, then for $i\geq 2$ we can write $K_i$ as
	\begin{equation}
		K_i= \begin{cases}
			K_{i-1} +1 &\text{with probability }\epsilon\\
			0	   &\text{with probability }1-\epsilon.
		\end{cases}
	\end{equation}
	So {$(K_i)_{i\geq 1}$ forms} a Markov process represented by the Markov chain in Fig.~\ref{fig:fig_ch8_Kn_sys_mc}.
\begin{figure*}[!t]
	\centering
\begin{tikzpicture}[>=stealth',shorten >=1pt,auto,node distance=2.4cm, semithick]
	\tikzstyle{every state}=[fill=red,draw=none,text=white]
	\node[state] (A)                     {$0$};
	\node[state] (B) [right of=A]  	     {$1$};
	\node[state] (C) [right of=B]  	     {$2$};
	\node[state] (D) [right of=C] 	     {$3$};
	\node[state,fill=white]	(E) at (8.5,0) {};
	
	\path [->] (A) edge [bend left]      node[ fill=white, anchor=center, pos=0.5] {$\epsilon$} (B);
	\path [->] (A) edge [loop left]      node {$1-\epsilon$} (A);
	
	\draw [->] (B) edge [bend left]      node[ fill=white, anchor=center, pos=0.5] {$1-\epsilon$} (A);
	\path [->] (B) edge [bend left]      node[ fill=white, anchor=center, pos=0.5] {$\epsilon$} (C);
	
	\draw [->] (C) edge [bend left]      node[ fill=white, anchor=center, pos=0.5] {$1-\epsilon$} (A.290);
	\path [->] (C) edge [bend left]      node[ fill=white, anchor=center, pos=0.5] {$\epsilon$} (D);
	
	\draw [->] (D) edge [bend left]      node[ fill=white, anchor=center, pos=0.5] {$1-\epsilon$} (A.270);
	\path [-]  (D) edge [bend left]      node {} (E);
	\path [->] (E.south) edge [bend left]      node[ fill=white, anchor=center, pos=0.5] {$1-\epsilon$} (A.250);

\end{tikzpicture}
\caption{Markov chain governing $K_n$, the number of transmissions since the reception of the last successful source symbol.}
\label{fig:fig_ch8_Kn_sys_mc}
\end{figure*}
	This Markov process is ergodic and has a stationary distribution which is identical to a geometric random variable $K$.
	This means 	{that
	\begin{align}
		& \lim_{i\to\infty}\Prob(K_i=a)=\Prob(K=a)=\epsilon^a(1-\epsilon),\quad\forall a\geq0,
	\end{align}
	}
	and {almost surely, we have} \begin{equation}
		\label{eq:eq_ch8_K_expectation}
		\lim_{N\to\infty} \frac{1}{N}\sum_{i=1}^N K_i = \E(K) = \frac{\epsilon}{1-\epsilon}.
	\end{equation}
	Replacing \eqref{eq:eq_ch8_age_same_alph_proof_1} in \eqref{eq:eq_ch8_app_same_alph_avg_age_2}, we get
	
	{
	\begin{align}
		\label{eq:eq_ch8_same_alph_avg_age_3}
		\Delta_\epsilon
				&= \lim_{N\to\infty} \frac{1}{N}\sum_{i=1}^{N}\lp  \frac{1}{2\mu}+ \frac{K_i}{\mu}+\frac{1}{\lambda}\left[\frac{\lambda}{\mu}(i-1-K_i)-\lambda t_0^s\right]\rp\nn
				&=\frac{1}{2\mu}+ \frac{1}{\mu}\lim_{N\to\infty} \frac{1}{N}\sum_{i=1}^{N}K_i+ \frac{1}{\lambda}\lim_{N\to\infty} \frac{1}{N}\sum_{i=1}^{N}\left[\frac{\lambda}{\mu}(i-1-K_i)-\lambda t_0^s\right]\nn
				&=\frac{1}{2\mu}+ \frac{\E(K)}{\mu}+ \frac{1}{\lambda}\lim_{N\to\infty}  \frac{1}{N}\sum_{i=1}^{N}\left[\frac{\lambda}{\mu}(i-1-K_i)-\lambda t_0^s\right]\nn
				&=\frac{1}{2\mu}+ \frac{\epsilon}{\mu(1-\epsilon)}+ \frac{1}{\lambda}\lim_{N\to\infty}  \frac{1}{N}\sum_{i=1}^{N}\left[\frac{\lambda}{\mu}(i-1-K_i)-\lambda t_0^s\right],
	\end{align}

	where the third and fourth equalities follow from \eqref{eq:eq_ch8_K_expectation}.
	
	At this point, we need to distinguish between two cases:
	\begin{itemize}
		\item $\rho=\frac{\lambda}{\mu}$ is irrational. In this case, we need to rewrite 
		$i-1-K_i$. Let $S_0=0$, and for $r\geq 1$ let $S_r$ be the index of the $r^{th}$ channel-use which was not erased. Define $X_l=S_l-S_{l-1}$. Clearly, $(X_l)_{l\geq1}$ are independent and identically distributed as geometric random variables, i.e., $\mathbb{P}(X_l=x)=\epsilon^{x-1}(1-\epsilon)$ for $x\in\mathbb{N}^*$. It is easy to see that $i-1-K_i=Y_i$, where
		\begin{equation}
Y_i=\max\{S_l:l\geq 0\text{ and }S_l\leq i\}.		
		\end{equation}

		 Then, by \Cref{lemma:lemma_ch8_weyl_equidistribution},
		 \begin{equation}
				\lim_{N\to\infty}\sum_{i=1}^{N}\left[ \frac{\lambda}{\mu}(i-1-K_i)-\lambda t_0^s\right]=\lim_{N\to\infty}\sum_{i=1}^{N}\left[ \frac{\lambda}{\mu}Y_i-\lambda t_0^s\right]=
				\frac{1}{2}.		 
		 \end{equation}
			Using this result in \eqref{eq:eq_ch8_same_alph_avg_age_3} we get
			\eqref{eq:eq_ch8_opt_avg_age_irrational_rho}.
		\item $\rho=\frac{\lambda}{\mu}=\frac{c}{d}$ is rational with $c,d\in\mathbb{N}$, $d>0$ and $\gcd(c,d)=1$. Since
		\begin{align}
\frac{1}{N}\sum_{i=1}^{d\lfloor \frac{N}{d}\rfloor}\left[ \frac{\lambda}{\mu}(i-1-K_i)-\lambda t_0^s\right]
&\leq \frac{1}{N}\sum_{i=1}^{N}\left[ \frac{\lambda}{\mu}(i-1-K_i)-\lambda t_0^s\right]\nonumber\\
&\leq \left(\frac{1}{N}\sum_{i=1}^{d\lfloor \frac{N}{d}\rfloor}\left[ \frac{\lambda}{\mu}(i-1-K_i)-\lambda t_0^s\right]\right)+\left(\frac{1}{N}\sum_{i=d\lfloor \frac{N}{d}\rfloor}^{d(\lfloor \frac{N}{d}\rfloor+1)}\left[ \frac{\lambda}{\mu}(i-1-K_i)-\lambda t_0^s\right]\right)\nonumber\\
&\leq \left(\frac{1}{N}\sum_{i=1}^{d\lfloor \frac{N}{d}\rfloor}\left[ \frac{\lambda}{\mu}(i-1-K_i)-\lambda t_0^s\right]\right)+\frac{d}{N},
		\end{align}
		it follows that
		\begin{align}
\lim_{N\to\infty}	\frac{1}{N}\sum_{i=1}^{N}\left[ \frac{\lambda}{\mu}(i-1-K_i)-\lambda t_0^s\right]
&=\lim_{N\to\infty}	\frac{1}{N}\sum_{i=1}^{d\lfloor \frac{N}{d}\rfloor}\left[ \frac{\lambda}{\mu}(i-1-K_i)-\lambda t_0^s\right]\nonumber\\
&=\lim_{N\to\infty}	\frac{d\lfloor \frac{N}{d}\rfloor}{N}\frac{1}{d\lfloor \frac{N}{d}\rfloor}\sum_{i=1}^{d\lfloor \frac{N}{d}\rfloor}\left[ \frac{\lambda}{\mu}(i-1-K_i)-\lambda t_0^s\right]\nonumber\\
&=\lim_{N\to\infty}	\frac{1}{d\lfloor \frac{N}{d}\rfloor}\sum_{i=1}^{d\lfloor \frac{N}{d}\rfloor}\left[ \frac{\lambda}{\mu}(i-1-K_i)-\lambda t_0^s\right]\nonumber\\
&=\lim_{N\to\infty}	\frac{1}{d N}\sum_{i=1}^{d N}\left[ \frac{\lambda}{\mu}(i-1-K_i)-\lambda t_0^s\right]\nonumber\\
&=\lim_{N\to\infty}	\frac{1}{d N}\sum_{i=0}^{d N-1}\left[ \frac{\lambda}{\mu}(i-K_{i+1})-\lambda t_0^s\right]\nonumber\\
&=\lim_{N\to\infty}	\frac{1}{d N}\sum_{a=0}^{d-1}\sum_{i=0}^{N-1}\left[ \frac{c}{d}(d i+a-K_{d i+a+1})-\lambda t_0^s\right]\nonumber\\&=	\frac{1}{d}\sum_{a=0}^{d-1}\lim_{N\to\infty}	\frac{1}{N}\sum_{i=0}^{N-1}\left[ ci+\frac{c}{d}(a-K_{d i+a+1})-\lambda t_0^s\right]\nonumber\\
&=	\frac{1}{d}\sum_{a=0}^{d-1}\lim_{N\to\infty}	\frac{1}{N}\sum_{i=0}^{N-1}\left[ \frac{c}{d}a-\left(\frac{c}{d}K_{d i+a+1}+\lambda t_0^s\right)\right]\nonumber\\
&\stackrel{(\ast)}{=}	\frac{1}{d}\sum_{a=0}^{d-1}\E\left(\left[ \frac{c}{d}a-\left(\frac{c}{d}K+\lambda t_0^s\right)\right]\right)\nonumber,
\end{align}
where $(\ast)$ follows from the fact that $(K_i)_{i\geq 1}$ is ergodic. Continuing,
\begin{align}
\lim_{N\to\infty}	\frac{1}{N}\sum_{i=1}^{N}\left[ \frac{\lambda}{\mu}(i-1-K_i)-\lambda t_0^s\right]
&=	\frac{1}{d}\E\left(\sum_{a=0}^{d-1}\left[ \frac{c}{d}a-\left(\frac{c}{d}K+\lambda t_0^s\right)\right]\right)\nonumber\\
&\stackrel{(\dagger)}{=}\frac{1}{d}\E\left(\frac{d-1}{2}+\left[-d\left(\frac{c}{d}K+\lambda t_0^s\right)\right]\right)\nonumber\\
&=\frac{1}{d}\E\left(\frac{d-1}{2}+\left[-K-d\lambda t_0^s\right]\right)\nonumber\\
&=\frac{1}{d}\E\left(\frac{d-1}{2}+\left[-d\lambda t_0^s\right]\right)\nonumber\\
&=\frac{d-1}{2d}+\frac{\left[-d\lambda t_0^s\right]}{d},
		\end{align}
where $(\dagger)$ follows from \Cref{lemma:lemma_ch8_equidistribution_rational}. Using this result in \eqref{eq:eq_ch8_same_alph_avg_age_3} we get
			\eqref{eq:eq_ch8_opt_avg_age_rational_rho}.
	\end{itemize}
	}
\end{proof}

\begin{remark}
One interesting application of \Cref{thm:thm_ch8_opt_avg_age} is the computation of the average age of information for the D/M/1 system with preemption: If we have a D/M/1 queue with deterministic interarrival time of rate $\lambda$ and exponential service time of rate $\mu$, then we can model the random exponentially distributed service time as being the result of having an erasure channel that can be used every $T_c=dt\ll 1$ and where the erasure probability is $1-\mu dt$. In this case, we get from \eqref{eq:eq_ch8_opt_avg_age_irrational_rho} that the average age of information of a D/M/1 system with preemption is equal to
$$\frac{1}{2\lambda}+\frac{1+1-\mu dt}{2\frac{1}{dt}(1-(1-\mu dt))}=\frac{1}{2\lambda}+\frac{1}{\mu}-\frac{dt}{2} \;\stackrel{dt\to 0}{\longrightarrow}\; \frac{1}{2\lambda}+\frac{1}{\mu},$$
which is consistent with the formula that was derived in \cite{InoueEtAl} for D/M/1 systems with preemption.
\end{remark}

%%%%%%%%%%%%%%%%%%%%%%%%%%%%%%%%%%%%%%%%%%%%%%%%%%%%%%%%%%%%%%%%%%%%
\section{Optimal Age with Different Source \& Channel alphabets}
\label{sec:sec_ch8_diff_alphabets}
In this setup, we consider the model described {in} \Cref{sec:sec_ch8_preliminaries}: The channel is a {$q$-ary} erasure
channel {$q$EC($\epsilon$)} {without} feedback and {with} input alphabet $\mathcal{V}=\{0,1,\ldots,q-1\}$. The source alphabet is $\mathcal{U}=\mathcal{V}^k$ where $k>1$. We only consider the special case where $\lambda=\mu$, so that at every channel use, a new source symbol is generated. By combining the techniques of \Cref{sec:sec_ch8_same_alphabet} and this section, one might be able to obtain reasonably good lower and upper bounds on the optimal age for the more general case where $\lambda$ can be different from $\mu$, but we expect the calculation to be more complicated.

Since we only consider the case where $\lambda=\mu$, we can assume without loss of
generality that $\lambda=\mu=1$. {The $i^{th}$ channel use takes place at time $t^c_{i}=iT_c=i$, and the $m^{th}$ source-symbol is generated at time $t^s_m=mT_s+t_0^s=m+t_0^s$.} The difference between the source alphabet and the channel-input alphabet as well as
the presence of erasures impose the use of channel coding on the generated source symbols before their transmission. We will focus on coding schemes that are induced by linear {block} codes, so we will assume that
$\mathcal{V}=\mathbb{F}_q$, where $q$ is a power of a prime number. {We fix a blocklength $n\geq k$, and each transmitted message will be encoded into a block of $n$ channel-input symbols in $\mathcal{V}=\mathbb{F}_q$, and then transmitted through $n$ consecutive channel uses. More precisely, the $l^{th}$ transmitted message is encoded using an $(n,k)$ linear block encoder $F_l:\mathbb{F}_q^k\to\mathbb{F}_q^n$ to produce $n$ channel-input symbols in $\mathcal{V}=\mathbb{F}_q$, which will be transmitted through the $((l-1)n+1)^{th}$,  the $((l-1)n+2)^{th}$, \ldots, and the $(ln)^{th}$ channel uses. In order to transmit messages that are as fresh as possible, the $l^{th}$ transmitted message is the last message that was generated before time $t^c_{(l-1)n+1}=(l-1)n+1$, i.e., the $l^{th}$ transmitted message is the $m(l)^{th}$ generated source symbol, where $m(l)=\lfloor (l-1)n+1-t_0^s\rfloor$. All the source-symbols that are generated between $t_{m(l)+1}^s$ and $t_{m(l+1)}^s$ are discarded. Fig.~\ref{fig:fig_ch8_instantaneous_age_MDS_opt} illustrates this
concept. We emphasize the fact that the $(n,k)$-linear codes $(F_l)_{l\geq1}$ used to encode different messages can be different. We denote a coding scheme that is induced by a given sequence of $(n,k)$-linear codes  $(F_l)_{l\geq1}$ as $\mathcal{C}(n,k)$.
}

%For a $q-ary$ EC($\epsilon$), a \emph{coding scheme} $(\mathcal{V}^k,n,f,g)$
%is a scheme where every source symbol $U^k$ is represented using $k$ channel-input
%symbols and is encoded using $f: \mathcal{V}^k\to\mathcal{V}^n$ into a codeword $V^n=f(U^k)$ of $n$ channel-input symbols. The decoder $g:
%(\mathcal{V}\cup\{\text{?}\})^n\to\mathcal{V}^k$ observes $n$ consecutive channel-output symbols $Z^n$ and decodes them into an
%estimation of the transmitted source symbol, $\hat{U}^k=g(Z^n)$. Fig.~\ref{fig:fig_ch8_communication_system} illustrates this
%concept. We denote a code corresponding to a $(\mathcal{V}^k,n,f,g)$ coding scheme by $\mathcal{C}(n,k)$. In this paper we
%are interested only in linear codes.
\begin{figure}[t]
	\centering
	\begin{tikzpicture}[scale=0.6,font=\scriptsize]
		% horizontal axis
		\draw[->] (0,0) -- (7,0) node[anchor=north] {channel uses};
		% labels
		\draw[color=red,thick]	(1,3pt) -- (1,-3pt);
		\draw[->] 		(1,-0.7) node[anchor=north,align=center]{Time of generation and\\beginning of
		transmission\\of a new successful update} -- (1,-0.3);
		\draw[color=red,thick]	(1.4,3pt) -- (1.4,-3pt);
		\draw[color=blue]	(1.4,0) node {x};
		\draw[color=red,thick]	(1.8,3pt) -- (1.8,-3pt);
		\draw[color=blue]	(1.8,0) node {x};
		\draw[color=red,thick]	(2.2,3pt) -- (2.2,-3pt);
		\draw[color=blue]	(2.2,0) node {x};
		\draw[color=red,thick]	(2.6,3pt) -- (2.6,-3pt) ;
		\draw[color=red,thick]	(3,3pt) -- (3,-3pt);
		\draw[color=blue]	(3,0) node {x};
		
		\draw[color=red,thick]	(3.4,3pt) -- (3.4,-3pt);
		\draw[color=red,thick]	(3.8,3pt) -- (3.8,-3pt);
					
		\draw[color=red,thick]	(4.2,3pt) -- (4.2,-3pt);
		\draw[color=red,thick]	(4.6,3pt) -- (4.6,-3pt);
		\draw[->] 		(4.6,-0.7) node[anchor=north,align=center]{Time of end\\of transmission} -- (4.6,-0.3);
		\draw[color=red,thick]	(5,3pt) -- (5,-3pt);
		
		% vertical axis
		\draw[->] (0,0) -- (0,8) node[anchor=south] {$\Delta(t)$};
		
		\draw[thick,color=red] (1,4) -- (3.4,6.4) -- (3.4,2.4) -- (5,4);
		\draw[dashed,thick] (1,4) -- (4.6,7.6) -- (4.6,3.6)--(5,4);
		%\draw[thick,color=red] (1,2) -- (3.4,4.4) -- (3.4,0.4) -- (5,2);
		%\draw[dotted,thick]  (1,2) -- (4.6,5.6) -- (4.6,1.6)--(5,2);
		\draw[dotted,color=blue] (0,3.6) node [color=black,anchor=east] {$n-1$} -- (4.6,3.6) -- (4.6,0);
		\draw[dotted,color=blue] (0,2.4) node [color=black,anchor=east] {$i_{\text{dec}}^{\text{MDS}}-1$} -- (3.4,2.4) -- (3.4,0);
		%\draw[dotted,color=blue] (1,0) -- (3.4,2.4);
			
		%Legend
		\draw (5.2,4.8) rectangle (8.6,6);
		\draw[thick,color=red] (5.3,5.8)--(5.7,5.8) node [anchor=west] {$\Delta_{\epsilon,\mathcal{C}^{\text{MDS}}}(t)$};
		\draw[dashed,thick] (5.3,5.4) -- (5.7,5.4) node [anchor=west] {$\Delta_{\epsilon,\mathcal{C}}(t)$};
		\draw[color=blue] (5.5,5) node {x};
		\node[anchor=west] at (5.7,5) {channel erasure};
	\end{tikzpicture}
	\caption{Variation of the instantaneous age for an MDS code $\mathcal{C}^{\text{MDS}}$ and a non-MDS code $\mathcal{C}$. We assume $n=10$, $k=3$, $t_0^s=0$, $i_{\text{dec}}=n=10$ and $i_{\text{dec}}^{\text{MDS}}=7$.}
	\label{fig:fig_ch8_instantaneous_age_MDS_opt}
\end{figure}

\subsection{The Optimal Transmission Policy}
\label{subsec:subsec_ch8_diff_alph_opt_transmission_policy}
\begin{definition}
	\label{def:def_ch8_MDS_code}
	An $(n,k)$-linear code %$\mathcal{C}(n,k)$ of a coding scheme $(\mathcal{V}^k,n,f,g)$, where $\mathcal{V}$ is the channel-input alphabet,
	is called \emph{maximum distance separable} (MDS) if it achieves the Singleton bound:
	\[ d=n-k+1,\]
	with $d$ denoting the minimum distance\footnote{See \cite{RichardsonUrbanke2008} for more details.} between the
codewords of the code.% $\mathcal{C}$.
\end{definition}
\begin{proposition}
	\label{prop:prop_ch8_MDS_code_property}
	If %$(\mathcal{V}^k,n,f,g)$ is a coding scheme such that 
	the encoder $F$ generates an MDS $(n,k)$-linear code, %$\mathcal{C}(n,k)$, 
	then
	\begin{itemize}
		\item any $k$ columns of the generator matrix $\mathbf{G}$ {of the encoder $F$} are linearly independent,
		\item any subset of size $k$ taken from a length-$n$ codeword is sufficient to recover, with probability 1, the transmitted
	message.
	\end{itemize}
	This means that if the channel is {$q$EC($\epsilon$)}, the decoder needs to observe only $k$ unerased
	channel-input symbols in order to perfectly decode the transmitted source symbol.
\end{proposition}
\Cref{prop:prop_ch8_MDS_code_property} is well known. We refer the reader to
\cite{RichardsonUrbanke2008} for more details. The following theorem presents the optimal channel codes from an age point of
view when the channel does not have any feedback.
\begin{thm}
	\label{thm:thm_ch8_MDS_optimal}
	Consider {a $q$EC($\epsilon$) channel without} feedback and let $n,k$ be such that MDS $(n,k)$-linear codes exist. Among all $(n,k)$-linear codes, MDS codes are age optimal. This
means that, to achieve age optimality, all codes used in the scheme $\mathcal{C}(n,k)$ should be MDS.
\end{thm}
\begin{proof}
	Fix two positive integers $k$ and $n$, {and let $\mathcal{C}$ and $\mathcal{C}^{\text{MDS}}$ be two $\mathcal{C}(n,k)$ coding schemes such that $\mathcal{C}$ is induced by arbitrary $(n,k)$-linear encoders $(F_l)_{l\geq1}$ and $\mathcal{C}^{\text{MDS}}$ is induced by $(n,k)$-linear encoders $(F_l^{\text{MDS}})_{l\geq1}$ which are MDS.}
	
	Now consider a source with alphabet {$\mathcal{V}^k=\mathbb{F}_q^k$}
	generating messages and sending them through two parallel {$q$EC($\epsilon$) channels} but with the same erasure pattern
	$\mathcal{E}$. For the first
	channel we use the coding scheme {$\mathcal{C}$}, while for the second channel we use the coding scheme {$\mathcal{C}^{\text{MDS}}$. We will show that for every $t\geq 1$, we have $\Delta_{\epsilon,\mathcal{C}^{\text{MDS}}}(t)\leq \Delta_{\epsilon,\mathcal{C}}(t)$. We assume that the initial ages before transmission are equal, i.e., $\Delta_{\epsilon,\mathcal{C}^{\text{MDS}}}(t)= \Delta_{\epsilon,\mathcal{C}}(t)$ for $t< 1$.
	
	For every integer $i\geq 1$ and every $t\in (i,i+1)$, we have $\Delta_{\epsilon,\mathcal{C}}(t)=\Delta_{\epsilon,\mathcal{C}}(i)+t-i$ and $\Delta_{\epsilon,\mathcal{C}^{\text{MDS}}}(t)=\Delta_{\epsilon,\mathcal{C}^{\text{MDS}}}(i)+t-i$. This is because the receiver does not receive any information during the interval $(i,i+1)$. Therefore, it is sufficient to show that $\Delta_{\epsilon,\mathcal{C}^{\text{MDS}}}(i)\leq \Delta_{\epsilon,\mathcal{C}}(i)$ for every integer $i\geq 1$.
	
	Let $l\geq 1$ and assume that $\Delta_{\epsilon,\mathcal{C}^{\text{MDS}}}(i)\leq \Delta_{\epsilon,\mathcal{C}}(i)$ for every $i<(l-1)n+1$. As described in the first paragraph of this section, for every $l\geq 1$, the $n$ channel uses between time $t_{(l-1)n+1}^c=(l-1)n+1$ and time $t_{(l-1)n}^c=ln$ are used to transmit the message $U_{m(l)}$ that was generated at time $t_{m(l)}^s=m(l)+t_0^s$, where $m(l)=\lfloor (l-1)n+1-t_0^s\rfloor$. Let $Z_{(l-1)n+1},Z_{(l-1)n+2},\ldots,Z_{ln}$ be the respective outputs of the $((l-1)n+1)^{th}$, the $((l-1)n+2)^{th}$, \ldots, and the $(ln)^{th}$ channel uses when the code $\mathcal{C}$ is used. Similarly, let $Z_{(l-1)n+1}^{\text{MDS}},Z_{(l-1)n+2}^{\text{MDS}},\ldots,Z_{ln}^{\text{MDS}}$ be the respective outputs of the $((l-1)n+1)^{th}$, the $((l-1)n+2)^{th}$, \ldots, and the $(ln)^{th}$ channel uses when the code $\mathcal{C}^{\text{MDS}}$ is used. For every $i\in\{(l-1)n+1,\ldots,ln\}$, we have
	\begin{equation}
	\label{eq:eq_age_recursive_nonMDS}
	\begin{aligned}
	\Delta_{\epsilon,\mathcal{C}}(i)=\begin{cases}
	\Delta_{\epsilon,\mathcal{C}}((l-1)n)+i-(l-1)n\quad&\text{if }i<i_{\text{dec}}(l),\\
	i-t_{m(l)}^s\quad&\text{if }i\geq i_{\text{dec}}(l),
	 \end{cases}
	\end{aligned}
	\end{equation}
where $i_{\text{dec}}(l)$ is the minimum $i\in\{(l-1)n+1,\ldots,ln\}$ such that $U_{m(l)}$ can be uniquely decoded from $(Z_{(l-1)n+1},\ldots,Z_i)$. If $U_{m(l)}$ cannot be uniquely decoded from $(Z_{(l-1)n+1},\ldots,Z_{ln})$,  we define $i_{\text{dec}}(l)=\infty$. Similarly, for every $i\in\{(l-1)n+1,\ldots,ln\}$, we have
	\begin{align}
		\label{eq:eq_age_recursive_MDS}
	\Delta_{\epsilon,\mathcal{C}}^{\text{MDS}}(i)=\begin{cases}
	\Delta_{\epsilon,\mathcal{C}}^{\text{MDS}}((l-1)n)+i-(l-1)n\quad&\text{if }i<i^{\text{MDS}}_{\text{dec}}(l),\\
	i-t_{m(l)}^s\quad&\text{if }i\geq i^{\text{MDS}}_{\text{dec}}(l),
	 \end{cases}
	\end{align}
where $i^{\text{MDS}}_{\text{dec}}(l)$ is the minimum $i\in\{(l-1)n+1,\ldots,ln\}$ such that $U_{m(l)}$ can be uniquely decoded from $(Z^{\text{MDS}}_{(l-1)n+1},\ldots,Z^{\text{MDS}}_i)$. If $U_{m(l)}$ cannot be uniquely decoded from $(Z_{(l-1)n+1}^{\text{MDS}},\ldots,Z_{ln}^{\text{MDS}})$, we define $i_{\text{dec}}^{\text{MDS}}(l)=\infty$.

Now observe that if $U_{m(l)}$ can be uniquely decoded from $(Z_{(l-1)n+1},\ldots,Z_i)$, then $Z_{(l-1)n+1},\ldots,Z_i$ contain at least $k$ non-erased symbols. Therefore, $(Z_{(l-1)n+1}^{\text{MDS}},\ldots,Z_i^{\text{MDS}})$ contain at least $k$ non-erased symbols and $U_{m(l)}$ can be uniquely decoded from $(Z^{\text{MDS}}_{(l-1)n+1},\ldots,Z^{\text{MDS}}_i)$ because $\mathcal{C}^{\text{MDS}}$ uses MDS codes. Therefore, $i^{\text{MDS}}_{\text{dec}}(l)\leq i_{\text{dec}}(l)$. We have:
\begin{itemize}
\item If $(l-1)n+1\leq i<\min\{ln,i^{\text{MDS}}_{\text{dec}}(l)\}$, we have $\Delta_{\epsilon,\mathcal{C}}^{\text{MDS}}(i)=\Delta_{\epsilon,\mathcal{C}}^{\text{MDS}}((l-1)n)+i-(l-1)n$ and $\Delta_{\epsilon,\mathcal{C}}(i)=\Delta_{\epsilon,\mathcal{C}}((l-1)n)+i-(l-1)n$. From the induction hypothesis we know that $\Delta_{\epsilon,\mathcal{C}}^{\text{MDS}}((l-1)n)\leq \Delta_{\epsilon,\mathcal{C}}((l-1)n)$. Therefore, $\Delta_{\epsilon,\mathcal{C}}^{\text{MDS}}(i)\leq \Delta_{\epsilon,\mathcal{C}}(i)$ for every $(l-1)n+1\leq i<i^{\text{MDS}}_{\text{dec}}(l)$.
\item If $i^{\text{MDS}}_{\text{dec}}(l)\leq i<\min\{ln, i_{\text{dec}}(l)\}$, we have $\Delta_{\epsilon,\mathcal{C}}^{\text{MDS}}(i)=i-t_{m(l)}^{s}$ and $\Delta_{\epsilon,\mathcal{C}}(i)=\Delta_{\epsilon,\mathcal{C}}((l-1)n)+i-(l-1)n$. Since the last decoded message by $\mathcal{C}$ before $t_{(l-1)n+1}^c$ has a timestamp that is earlier than $t_{m(l)}^s$, we have $(l-1)n-\Delta_{\epsilon,\mathcal{C}}((l-1)n)<t_{m(l)}^{s}$. Therefore, $\Delta_{\epsilon,\mathcal{C}}^{\text{MDS}}(i)< \Delta_{\epsilon,\mathcal{C}}(i)$ for every $i^{\text{MDS}}_{\text{dec}}(l)\leq i<\min\{ln, i_{\text{dec}}(l)\}$.
\item For every $i_{\text{dec}}(l)\leq i<ln$, we have $\Delta_{\epsilon,\mathcal{C}}^{\text{MDS}}(i)= \Delta_{\epsilon,\mathcal{C}}(i)$.
\end{itemize}

This implies that $\Delta_{\epsilon,\mathcal{C}}^{\text{MDS}}(i)\leq \Delta_{\epsilon,\mathcal{C}}(i)$ for every $(l-1)n+1\leq i\leq ln$. It follows by induction that $\Delta_{\epsilon,\mathcal{C}}^{\text{MDS}}(i)\leq \Delta_{\epsilon,\mathcal{C}}(i)$ for every integer $i\geq 1$.
}

\end{proof}

%Now we are ready to present the optimal transmission policy for a coding scheme $\mathcal{C}(n,k)$ with $\lambda=\mu=1$.
%\begin{lemma}
%	\label{lemma:lemma_ch8_opt_transmission}
%	For $\lambda=\mu$, at every channel use, a new source symbol is generated. The optimal transmission scheme from an age
%	point of view is to use an MDS codes, and, whenever a message finishes its transmission after $n$ channel uses, we
%	begin transmitting, at the $(n+1)^{th}$ channel use, the source symbol generated at that same time instant. This means
%	that when the packet (or message) generated at time $t_0$ finishes its transmission at time $t'=t_0+n-1$, the next
%	packet to be transmitted at $t_0+n$ is the one generated at this same time instant. All messages generated between
%	$t_0+1$ and $t_0+n-1$ are dropped.
%\end{lemma}
%\begin{proof}
%	Storing any number of packets and sending them will lead to a non-zero waiting time incurred by the stored messages
%	since $n\geq k>1$.
%	Whereas the policy presented in \Cref{lemma:lemma_ch8_opt_transmission} guarantees zero waiting time and whenever the
%	system is idle, the freshest message is transmitted. Moreover, since we are assuming a constant service-time, then
%	using \cite[Lemma 3]{UpdateorWait-2016arxiv} we deduce that the just-in-time policy is optimal. In our case, the
%	just-in-time policy translates into the scheme described in \Cref{lemma:lemma_ch8_opt_transmission}.
%\end{proof}
\Cref{thm:thm_ch8_MDS_optimal} shows that for a given couple $(n,k)$, the optimal coding scheme is the one that uses
only MDS codes. However,
an explicit construction of such codes is not available for all values of $(n,k)$. In the rest of this paper, we use random codes to
give an upper bound on the optimal average age. The use of random coding to construct fountain-like
codes was used by Shamai et al. in \cite{ShamaiTelatarVerdu2007}. The authors of \cite{ShamaiTelatarVerdu2007} showed that without any randomness
we cannot properly define the notion of fountain capacity because there is always a case where the deterministic fountain codes cannot
achieve any positive rate with an error probability tending to $0$. Nevertheless, we use the rateless (or fountain) codes, previously adopted
in \cite{NajmYatesSoljanin-ISIT2017}, to give a lower bound on the optimal achievable average-age $\Delta_{\epsilon}$. As shown in \cite{ShamaiTelatarVerdu2007},
these codes cannot be implemented in practice, {this is} why we do not consider them {as} part of the possible coding
schemes. 

\begin{figure*}[t]
	\centering
	\begin{tikzpicture}[scale=0.7,font=\scriptsize]
		\path[pattern=north west lines,pattern color=red!15]
		(2.6,0)--(2.6,1.6)--(6.2,5.2)--(6.2,1.2)--(6.6,1.6)--(6.6,0) -- cycle;
		\path[pattern=horizontal lines, pattern color=blue!15]
		(6.6,0)--(6.6,1.6)--(7.8,2.8)--(7.8,0.8)--(8.6,1.6)--(8.6,0)--cycle;
		\path[pattern=north west lines,pattern color=red!15]
		(8.6,0)--(8.6,1.6)--(12.6,5.6)--(12.6,0) -- cycle;
		\path[pattern=horizontal lines, pattern color=blue!15]
		(-1.4,0)--(-1.4,1.6)--(2.6,5.6)--(2.6,0)--cycle;
		
		% horizontal axis
		\draw[->] (-1.4,0) -- (13.5,0) node[anchor=north] {$t$};
		
		% labels
		\draw[color=red,thick]	(-1,3pt) -- (-1,-3pt);
		\draw[color=red,thick]	(-0.6,3pt) -- (-0.6,-3pt);
		\draw[color=blue] 	(-0.6,0) node {x};
		\draw[color=red,thick]	(-0.2,3pt) -- (-0.2,-3pt);
		\draw[color=red,thick]	(0.2,3pt) -- (0.2,-3pt);
		\draw[color=blue] 	(0.2,0) node {x};
		\draw[color=red,thick]	(0.6,3pt) -- (0.6,-3pt);
		\node at (0.6,0) {o};
		\draw[dashed]		(0.6,-3pt) -- (0.6,-1.2);

		\draw[color=red,thick]	(1,3pt) -- (1,-3pt);
		\draw[color=blue] 	(1,0) node {x};
		\draw[color=blue] 	(1,0) node [anchor=north] {$t_1$};
		\draw[color=red,thick]	(1.4,3pt) -- (1.4,-3pt);
		\draw[color=red,thick]	(1.8,3pt) -- (1.8,-3pt);	
		\draw[color=blue] 	(1.8,0) node {x};
		\draw[color=red,thick]	(2.2,3pt) -- (2.2,-3pt);
		\draw[color=red,thick]	(2.6,3pt) -- (2.6,-3pt);
		\draw[dashed]		(2.6,-0.55) -- (2.6,-1.2);
		\draw[color=red] 	(2.6,0) node [anchor=north] {$t'_1$};

		\draw[color=red,thick]	(3,3pt) -- (3,-3pt);
		\draw[color=blue] 	(3,0) node {x};
		\draw[color=red,thick]	(3.4,3pt) -- (3.4,-3pt);
		\draw[color=blue] 	(3.4,0) node {x};
		\draw[color=red,thick]	(3.8,3pt) -- (3.8,-3pt);			
		\draw[color=blue] 	(3.8,0) node {x};
		\draw[color=red,thick]	(4.2,3pt) -- (4.2,-3pt);
		\draw[color=blue] 	(4.2,0) node {x};
		\draw[color=red,thick]	(4.6,3pt) -- (4.6,-3pt);
		\draw[dashed]		(4.6,-3pt) -- (4.6,-1.2);
		
		\draw[color=red,thick]	(5,3pt) -- (5,-3pt);
		\draw[color=blue] 	(5,0) node [anchor=north] {$t_2$};
		\draw[color=red,thick]	(5.4,3pt) -- (5.4,-3pt);
		\draw[color=red,thick]	(5.8,3pt) -- (5.8,-3pt); 
		\node at (5.8,0) {o};
		\draw[color=red,thick]	(6.2,3pt) -- (6.2,-3pt);
		\draw[color=red,thick]	(6.6,3pt) -- (6.6,-3pt);
		\node at (6.6,0) {o};
		\draw[dashed]		(6.6,-0.55) -- (6.6,-1.2);
		\draw[color=red] 	(6.6,0) node [anchor=north] {$t'_2$};

		\draw[color=red,thick]	(7,3pt) -- (7,-3pt);
		\draw[color=blue] 	(7,0) node [anchor=north] {$t_3$};
		\draw[color=red,thick]	(7.4,3pt) -- (7.4,-3pt);
		\draw[color=red,thick]	(7.8,3pt) -- (7.8,-3pt);
		\draw[color=red,thick]	(8.2,3pt) -- (8.2,-3pt);
		\node at (8.2,0) {o};
		\draw[color=red,thick]	(8.6,3pt) -- (8.6,-3pt);
		\draw[color=blue] 	(8.6,0) node {x};
		\draw[dashed]		(8.6,-0.55) -- (8.6,-1.2);
		\draw[color=red] 	(8.6,0) node [anchor=north] {$t'_3$};
		
		\draw[color=red,thick]	(9,3pt) -- (9,-3pt);
		\draw[color=blue] 	(9,0) node {x};
		\draw[color=red,thick]	(9.4,3pt) -- (9.4,-3pt);
		\draw[color=red,thick]	(9.8,3pt) -- (9.8,-3pt);
		\draw[color=red,thick]	(10.2,3pt) -- (10.2,-3pt);
		\draw[color=blue] 	(10.2,0) node {x};
		\draw[color=red,thick]	(10.6,3pt) -- (10.6,-3pt);
		\draw[color=blue] 	(10.6,0) node {x};
		\draw[dashed]		(10.6,-3pt) -- (10.6,-1.2);

		\draw[color=red,thick]	(11,3pt) -- (11,-3pt);
		\draw[color=blue] 	(11,0) node [anchor=north] {$t_4$};
		\draw[color=red,thick]	(11.4,3pt) -- (11.4,-3pt);
		\draw[color=blue] 	(11.4,0) node {x};
		\draw[color=red,thick]	(11.8,3pt) -- (11.8,-3pt);
		\draw[color=red,thick]	(12.2,3pt) -- (12.2,-3pt);
		\draw[color=blue] 	(12.2,0) node {x};
		\draw[color=red,thick]	(12.6,3pt) -- (12.6,-3pt);
		\draw[dashed]		(12.6,-0.55) -- (12.6,-1.2);
		\draw[color=red] 	(12.6,0) node [anchor=north] {$t'_4$};
		
		% vertical axis
		\draw[->] (-1.4,0) node[anchor=north] {$0$}-- (-1.4,8) node[anchor=south] {$\Delta_{\epsilon,\mathcal{C}}(t)$};
		\draw (-1.4cm+3pt,1.6) -- (-1.4cm-3pt,1.6) node[anchor=east] {$n-1$};
		
		\draw[thick] (-1.4,1.6) -- (2.6,5.6) -- (2.6,1.6) -- (6.2,5.2) -- (6.2,1.2) -- (7.8,2.8) -- (7.8,0.8) --
		(12.6,5.6) -- (12.6,1.6) -- (13,2);

		\draw[dotted,thick,color=blue] (2.6,1.6) -- (2.6,0);
		\draw[dotted,thick,color=blue] (6.6,1.6) -- (6.6,0);
		\draw[dotted,thick,color=blue] (8.6,1.6) -- (8.6,0);
		\draw[dotted,thick,color=blue] (12.6,1.6) --(12.6,0);
		
		% Interdeparture time
		\draw[<->,color=red] (-1.4,-1.3) -- node[anchor=north] {$Y_1$} (2.6,-1.3);
		\draw[<->,color=red] (2.6,-1.3) -- node[anchor=north] {$Y_2$} (6.6,-1.3);
		\draw[<->,color=red] (6.6,-1.3) -- node[anchor=north] {$Y_3$} (8.6,-1.3);
		\draw[<->,color=red] (8.6,-1.3) -- node[anchor=north] {$Y_4$} (12.6,-1.3);
		
		% System time 
		\draw[<->] (1,-0.7) -- node[anchor=north] {$T_1$} (2.6,-0.7);
		\draw[<->] (5,-0.7) -- node[anchor=north] {$T_2$} (6.2,-0.7);
		\draw[<->] (7,-0.7) -- node[anchor=north] {$T_3$} (7.8,-0.7);
		\draw[<->] (11,-0.7) -- node[anchor=north] {$T_4$} (12.6,-0.7);
		
		% Q's
		\node at (2,4) {$Q_1$};
		\node at (5.6,3.5) {$Q_2$};
		\node at (7.4,1.5) {$Q_3$};
		\node at (11.5,3.5) {$Q_4$};

		%Legend
		\draw (6,6.2) rectangle (14,9.2);
		\draw[thick,color=red] (6.3cm,-3pt+8.8cm)--(6.3cm,8.8cm+3pt);
		\node[anchor=west,align=left] at (6.6,8.8) {Channel use/\\Reception of a linearly independent symbol};
		\draw[dashed] (6.1,7.8) -- (6.6,7.8) node [anchor=west,align=center] {Time of end of transmission/\\Reception of the $n^{th}$
		symbol};
		\draw[color=red] (6.35,7cm-3pt)--(6.35,7cm+3pt);
		\draw[color=blue] (6.35,7) node {x};
		\node[anchor=west] at (6.6,7) {Channel erasure};
		\draw[color=red] (6.35,6.4cm-3pt)--(6.35,6.4cm+3pt);
		\draw (6.35,6.4) node {o};
		\node[anchor=west] at (6.6,6.4) {Reception of a linearly dependent symbol};
	\end{tikzpicture}
	\caption{Variation of the instantaneous age when using a random code $\mathcal{C}$ with $n=5$, $k=3$. We assume we
	begin observing after a successful reception. Since $\lambda=\mu=1$ then the interval between channel uses is one second. {Note that $t_0^s=0$ in the shown example.}}
	\label{fig:fig_ch8_instantaneous_age}
\end{figure*}
\subsection{The Random Code}
\label{subsec:subsec_ch8_random_codes}
Consider a  $\mathcal{C}(n,k)$ coding scheme. The encoder-decoder pair {$(F_l,G_l)$}, corresponding to
the $l^{th}$ message to be
transmitted, is constructed as follows: Since we are interested in linear codes, we use the generator matrix
in order to create our code. For that, we choose the $n$ columns of the generator matrix $\textbf{G}_l$ independently and
uniformly at random from the set {$\mathcal{V}^k\setminus\{0^k\}=\mathbb{F}_q^k\setminus\{0^k\}$}, where $0^k$ is the sequence of $k$ zeros. We denote
by {$(\textbf{g}_1^{(l)}, \textbf{g}_2^{(l)},\ldots, \textbf{g}_n^{(l)})$} 
	the $n$ columns\footnote{In this paper, we assume all
vectors to be column vectors.}   of $\textbf{G}_l$. Thus

\begin{equation}
	\textbf{G}_l = \begin{bmatrix}
		\textbf{g}_1^{(l)}& \textbf{g}_2^{(l)}& \cdots & \textbf{g}_n^{(l)}
	\end{bmatrix}.
\end{equation}
Once this matrix is generated, it is shared between the encoder and the decoder. For each new message to be
transmitted, we generate a new generator matrix. However, the encoder and decoder work in a similar fashion for all messages:%The {\color{red} coding scheme works} as follows:
\begin{itemize}
	\item  Let  {$\textbf{u}^{(l)}=\begin{bmatrix} u_1^{(l)}&u_2^{(l)}&\cdots&u_k^{(l)}\end{bmatrix}^T\in\mathcal{V}^k=\mathbb{F}_q^k$} be the
$l^{th}$ message to be
sent. Then, {at the $((l-1)n+i)^{th}$ channel use, we transmit the coded symbol %$z_i$
		 $z_{(l-1)n+i}=z^{(l)}_i:=\sum_{j=1}^k u_j^{(l)} g_{ji}^{(l)}$}, %\]
		  where ${g_{ji}^{(l)}}$ is the $j^{th}$ element of ${\textbf{g}_i^{(l)}}$. For each message $\textbf{u}^{(l)}$, we send $n$ coded symbols. Hence, the encoder is given by
		%\begin{align*}
		{$F_l:\mathcal{V}^k\to\mathcal{V}^n$, with
$\textbf{z}^{(l)}=F_l(\textbf{u}^{(l)})=(\textbf{u}^{(l)})^T\textbf{G}_l$}.
		%\end{align* 
		\item The decoder decodes on the fly. Whenever it receives $k$ linearly independent {non-erased} coded symbols, it decodes
			the message. Otherwise, it declares the packet to be erased.
\end{itemize}

We emphasize the fact that the matrices $(\textbf{G}_l)_{l\geq 1}$ are generated in a i.i.d. fashion, which means that the linear codes corresponding to different messages can (and are likely to) be different.

\subsection{Average Age of Random Codes}
Fix the couple $(n,k)$ and let {$\mathcal{C}$} be a {random $\mathcal{C}(n,k)$} coding scheme generated as described in
\Cref{subsec:subsec_ch8_random_codes}. We define
$\Delta_{\epsilon,(n,k)}$ to be the expected average age of the coding scheme induced by a random linear $(n,k)$-scheme generated
as above.
\begin{definition}
	\label{def:def_ch8_diff_alph_expect_avg_age}
{	For every $n\geq k$, and every $t\geq 0$, define
	\begin{equation}
		\label{eq:eq_ch8_diff_alph_expect_avg_age}
		\Delta_{\epsilon,(n,k)} = \E\lp \Delta_{\epsilon,\mathcal{C}}\rp,
	\end{equation}
	where the expectation in \eqref{eq:eq_ch8_diff_alph_expect_avg_age} is taken over the random $\mathcal{C}(n,k)$ coding scheme $\mathcal{C}$, and over the randomness of the erasure patterns of the $q$EC($\epsilon$) channels.}
\end{definition}

Due to the ergodicity of the system, almost surely {(over the randomly generated $\mathcal{C}$ and over the random erasure patterns)}, we have
\begin{equation}
	\label{eq:eq_ch8_diff_alph_expect_avg_age_vs_avg_age}
	{
	\Delta_{\epsilon,\mathcal{C}}=\Delta_{\epsilon,(n,k)}.
	}
\end{equation}
{	We will formally prove \eqref{eq:eq_ch8_diff_alph_expect_avg_age_vs_avg_age} in \Cref{lemma:lemma_ch8_diff_alph_avg_age_ergodic}.}

%Equation \eqref{eq:eq_ch8_diff_alph_expect_avg_age_vs_avg_age} is also due to the fact that the process
%$\Delta_{\epsilon,\mathcal{C}}(t)$ is mean-ergodic (see \Cref{def:def_ch2_mean_ergodic}) as we will prove later. 
The contribution of the random coding argument in this context is the following: {If we show that, for a given $n\geq k$, we have $\Delta_{\epsilon,(n,k)}<\infty$, then
there must exist a linear $(n,k)$-scheme $\mathcal{C}^{(n)}$ such that almost surely (over the random erasure patterns), $\Delta_{\epsilon,\mathcal{C}^{(n)}}=\Delta_{\epsilon,(n,k)}<\infty$. In fact, as we mentioned above, for almost all $(n,k)$-schemes $\mathcal{C}$ and almost all erasure patterns, we have $\Delta_{\epsilon,\mathcal{C}}=\Delta_{\epsilon,(n,k)}<\infty$. {Thus,
the optimal average age $\Delta_\epsilon$, and the optimal average age among linear block codes $\Delta_\epsilon^{lin}$, satisfy 
\begin{equation}
\Delta_\epsilon\leq \Delta_\epsilon^{lin}\leq \Delta_{\epsilon,(n,k)},\quad\forall n\geq k.
\end{equation}
Therefore,
\begin{equation}
	\label{eq:eq_ch8_diff_alph_upper_bound_opt_age}
	\Delta_{\epsilon}\leq \Delta_\epsilon^{lin}\leq\min_{n\geq k} \Delta_{\epsilon,(n,k)}.
\end{equation}
}
}
Equation \eqref{eq:eq_ch8_diff_alph_upper_bound_opt_age} gives an upper bound on the optimal average age. In the rest of this
paper we will focus on characterizing this bound.

\subsection{Exact Upper Bound on the Optimal Average Age}
\label{subsec:subsec_ch8_bounds_opt_age}
\subsubsection{Preliminaries}
Let $\mathcal{C}$ be a randomly generated $(n,k)$-scheme. Fig.~\ref{fig:fig_ch8_instantaneous_age} illustrates the variation of
the instantaneous age $\Delta_{\epsilon,\mathcal{C}}(t)$ when $n=5$ and $k=3$. Without loss of generality, we assume that we begin
observing right after the reception of a successful packet. We denote by $t_j$ the generation time of the $j^{th}$ successful
packet and by $t'_j$ the end of transmission time of this packet. {Assume that the $j^{th}$ successful message is the $l^{th}$ transmitted message. We have:
\begin{itemize}
\item $t_j=t_{m(l)}^s=m(l)T_s+t_0^s=m(l)+t_0^s$, where $m(l)=\lfloor (l-1)n+1-t_0^s\rfloor$.
\item $t_j'=t_{nl}^c=nlT_c=nl$.
\end{itemize}
Therefore, the instantaneous age at the end of transmission of the $j^{th}$ successful package is 
\begin{equation}
\label{eq:eqAgeAtTheEndOfSuccess}
\Delta_{\epsilon,\mathcal{C}}(t'_j)=t_j'-t_j=nl-m(l)-t_0^s=nl-\lfloor (l-1)n+1-t_0^s\rfloor-t_0^s=n-1-t_0^s-\lfloor -t_0^s\rfloor=n-1+[-t_0^s].
\end{equation}
}

In the scenario depicted in Fig.~\ref{fig:fig_ch8_instantaneous_age}, {we assume that $t_0^s=0$. The} first packet $\textbf{u}^{(1)}=({u}_1^{(1)},\ldots,{u}_1^{(k)})$ is generated and encoded into a
codeword $\textbf{z}^{(1)}=\lp{z}_1^{(1)},\ldots,{z}_n^{(1)}\rp=\lp(\textbf{u}^{(1)})^T\textbf{g}_1^{(1)},\ldots,(\textbf{u}^{(1)})^T\textbf{g}_n^{(1)}\rp$ of length $n=5$ at time $t=1$. At that same instant, $z_1^{(1)}$, the first symbol of
$\textbf{z}^{(1)}$, is sent and received at the monitor. Since it is the first symbol, $z_1^{(1)}$ is linearly independent.\footnote{By ``$z_1^{(1)}$ is linearly independent'', we just mean that the corresponding column $\mathbf{g}_1^{(1)}$ of the generator matrix $\mathbf{G}_1$ forms a linearly independent family of vectors. This is true simply because $\mathbf{g}_1^{(1)}\neq 0^k$.} At time
$t=2$, the coded symbol $z_2^{(1)}$ is erased but the coded symbol $z_1^{(3)}$, which is linearly independent from\footnote{Here, we just mean that in the particular example that is illustrated in Fig.~\ref{fig:fig_ch8_instantaneous_age}, the random matrix $\mathbf{G}_1$ was such that $\mathbf{g}_3^{(1)}$ is linearly independent from $\mathbf{g}_1^{(1)}$.} $z_1^{(1)}$, is received at
time $t=3$. The fourth coded symbol is also erased and the last coded symbol $z_5^{(1)}$ is received. However, as Fig.~\ref{fig:fig_ch8_instantaneous_age} shows, the received symbol $z_5^{(1)}$ is linearly dependent
on the previously received symbols, namely $z_1^{(1)}$ and $z_1^{(3)}$, i.e., $\mathbf{g}_5^{(1)}$ is linearly dependent on $\{\mathbf{g}_1^{(1)},\mathbf{g}_3^{(1)}\}$. The first packet $\textbf{u}^{(1)}$ is declared erased by the
decoder because it did not receive $3$ linearly independent symbols, and $\Delta_{\epsilon,\mathcal{C}}(t)$ increases linearly in the interval $t\in[1,6)$. The packet generated at $t=t_1=6$
is a successful update since the monitor receives $k=3$ linearly independent symbols at times $t=7$, $t=9$ and $t=10$. Therefore, 
$\Delta_{\epsilon,\mathcal{C}}(t)$ drops to $10-6=4$ at time $t=t'_1=10$. Note that, for a given successful packet,
once $k$ linearly independent coded symbols are received, any additional coded symbol must be linearly dependent on them.

In this section we use the following notation:% slightly different than the one introduced in \Cref{ch2}:
\begin{itemize}
	\item $Y_j=t'_j-t'_{j-1}$ is the interdeparture time between the ${(j-1)}^{th}$ and $j^{th}$ successfully received updates.
	\item $T_j$ is the number of channel uses between the decoding instant of the $j^{th}$ successful packet and its generation
		time $t_j$.
	\item $R(\tau)=\max\left\{j:t'_{j}\leq\tau\right\}$ is the number of successfully received updates in the interval $[0,\tau]$. 
	\item {Let $\mathbf{u}^{(l)}\in\mathbb{F}_q^k$ be the $l^{th}$ transmitted packet (not necessarily successful). Imagine that we generate infinitely many vectors $(\mathbf{g}_{i}^{(l)})_{i\geq 1}$ independently and uniformly in $\mathbb{F}_q^k\setminus\{0\}$. Imagine also that we transmit the coded symbol $z_i^{(l)}=(\mathbf{u}^{(l)})^T\mathbf{g}_{i}^{(l)}$ over a $q$EC($\epsilon$) channel to a virtual monitor for every $i\geq 1$. In reality, we only transmit $z_1^{(l)},\ldots,z_n^{(l)}$ to the real monitor. In other words, the first $n$ symbols are really transmitted and the rest are virtually transmitted. Let $B_l$ be the number of channel uses (or sent coded symbols)
		in order for the virtual monitor to receive exactly $k$ linearly independent equations (coded symbols). The $l^{th}$ packet is correctly
		decoded at the real monitor if and only if $B_l\leq n$.}
\end{itemize}
Since the channel is memoryless and the different codes used in the scheme $\mathcal{C}$ are generated independently and in the
same fashion, then the process {$(B_l)_{l\geq 1}$} is i.i.d with a distribution identical to the random variable $B$ {that we describe in the following subsection}.

\subsubsection{The Distribution of $B$}

\begin{figure*}[ht]
	\centering
\begin{tikzpicture}[>=stealth',shorten >=1pt,auto,node distance=2cm, semithick]
	\tikzstyle{every state}=[fill=red,draw=none,text=white]
	\node[state] (A)                     {$0$};
	\node[state,node distance=3.5cm] (B) [right of=A]   {$1$};
	%\node[state] (C) [right of=B]  	     {$2$};
	\node        (D) [right of=B] 	     {$\cdots$};
	\node[state] (E) [right of=D]	     {$s$};
	\node[node distance=3.5cm]       (G) [right of=E]   {$\cdots$};
	%\node[state] (F) [right of=E]	     {$k-1$};
	\node[state] (F) [right of=G]   {$k$};
	
	\path [->] (A) edge 		     node[anchor=north] {$p_0=\bar{\epsilon}=1-\epsilon$} (B);
	\path [->] (A) edge [loop above]      node {$\epsilon$} (A);
	
	\draw [->] (B) edge [loop above]     node[anchor=south] {$1-p_1$} (B);
	\path [->] (B) edge 		     node[anchor=north] {$p_1$} (D);% = \frac{\bar{\epsilon}q(q^{k-1}-1)}{q^k-1}$} (C);
	
	%\draw [->] (C) edge [loop above]     node[anchor=south] {$1-p_2$} (C);
	%\path [->] (C) edge 		     node[anchor=north] {$p_2$} (D); %=\frac{\bar{\epsilon}q^2(q^{k-2}-1)}{q^k-1}$} (D);
	
	\draw [->] (E) edge [loop above]     node[anchor=south] {$1-p_s$} (E);
	\path [->] (D) edge       	     node[anchor=north] {$p_{s-1}$} (E); %=\frac{\bar{\epsilon}q^3(q^{k-3}-1)}{q^k-1}$} (E);
	\path [->] (E) edge       	     node[anchor=north] {$p_{s}=\frac{\bar{\epsilon}q^s(q^{k-s}-1)}{q^k-1}$} (G);
	
	\draw [->] (F) edge [loop above]     node[anchor=south] {$1$} (F);
	\path [->] (G) edge		     node[anchor=north] {$p_{k-1}$} (F);
\end{tikzpicture}
\caption{Markov chain representing the dimension of a codeword at the (virtual) receiver.}
\label{fig:fig_ch8_diif_alph_codeword_dim}
\end{figure*}

Fig.~\ref{fig:fig_ch8_diif_alph_codeword_dim} shows the Markov chain that represents the dimension, at the {(virtual)} receiver, of the codeword relative to a certain update.

The monitor receives the first coded symbol of a new codeword with probability $p_0=\bar{\epsilon}=1-\epsilon$ and
hence the dimension of this codeword at the receiver jumps to $1$. If the first coded symbol is erased then the dimension of
the codeword remains at $0$. If the monitor has already received $s$ linearly independent coded symbols, then it will receive the
$(s+1)^{th}$ linearly independent coded symbol if:
\begin{itemize}
\item[(i)] the next transmitted coded symbol is not erased, and,
\item[(ii)] the next
transmitted coded symbol is linearly independent of all previously received symbols.
\end{itemize}
Event $(i)$ occurs with probability
$\bar{\epsilon}=1-\epsilon$. For event $(ii)$, notice that the symbols that are linearly dependent with the received symbols form a subspace of dimension\footnote{Recall that $s$ linearly independent coded symbols have been received.} $s$, hence there are $q^s$ such symbols. Therefore, the number of nonzero symbols that are linearly dependent with the received symbols is $q^s -1$. Now since coded symbols are generated uniformly at random from the set of nonzero symbols, we can see that event $(ii)$ happens with probability $\frac{q^k-1 - (q^s - 1)}{q^k-1}=\frac{q^s(q^{k-s}-1)}{q^k-1}$. Hence, for a given
message, the dimension of its codeword at the receiver jumps from $s$ to $s+1$ with probability 
\begin{equation}
	p_s=\frac{\bar{\epsilon}q^s(q^{k-s}-1)}{q^k-1},
\end{equation}
where $0\leq s\leq k-1$. If the next transmitted coded symbol is erased or linearly dependent on the previously received coded
symbols, then the dimension of the codeword at the monitor remains at $s$. As previously discussed, once the monitor receives $k$ linearly independent coded symbols, the dimension
of the codeword remains at $k$ and all subsequent coded symbols are linearly dependent on the previously non-erased coded
symbols.

From the above description, we can deduce that $B$ is the number of steps before reaching state $k$ for the first time.
\begin{remark}
	\label{rmk:rmk_ch8_diff_alph_lin_ind}
	Since $p_s=\frac{\bar{\epsilon}(q^{k}-q^s)}{q^k-1}$, then $p_s$ is a decreasing function of $s$. This means that
	whenever the decoder receives a non-erased coded symbol that is linearly independent from all previously received
	coded symbols, and the system jumps to state $s$, then it becomes harder to receive a new linearly independent coded
	symbol. {This is why, on average, the system spends more time} in state $s$ than in previous states.
\end{remark}
\begin{definition}
	\label{def:def_ch8_diff_alph_L}
	Let $L_s$ be the number of trials {needed} to pass from state $s$ to state $s+1$ in
	Fig.~\ref{fig:fig_ch8_diif_alph_codeword_dim}, where $0\leq s\leq k-1$. {It is easy to see that} $L_s$ has a geometric distribution with success
	probability $p_s=\frac{\bar{\epsilon}q^s(q^{k-s}-1)}{q^k-1}$. Thus,
	\begin{equation}
\Prob(L_s=i) = (1-p_s)^{i-1}p_s,\quad i=1,2,3,\ldots	
	\end{equation}
\end{definition}
\begin{corollary}
	\label{cor:cor_ch8_diff_alph_B}
	From \Cref{def:def_ch8_diff_alph_L}, we can write
	\begin{equation}
		\label{eq:eq_ch8_diff_alph_B_cor}
		B = \sum_{s=0}^{k-1} L_s,
	\end{equation}
	where {$(L_s)_{0\leq s<k}$} are independent.
\end{corollary}
\begin{lemma}
	\label{lemma:lemma_ch8_diff_alph_distribution_B}
	The moment generating function of the random variable $B$ is
	\begin{equation}
		\label{eq:eq_ch8_diff_alph_mgf_B}
		\resizebox{0.49\textwidth}{!}{$\displaystyle\phi_B(t) = \E\lp e^{tB}\rp = \lp\prod_{s=0}^{k-1} \lp q^k-q^s\rp\rp\lp\prod_{s=0}^{k-1}
		\frac{\bar{\epsilon}e^t}{q^k-1+e^t(1-\epsilon q^k-\bar{\epsilon}q^s)}\rp.$}
		%\lp\bar{\epsilon}e^t\rp^k \frac{q^{k^2}}{\lp (q^k-1)(1-\epsilon
		%	e^t)+\bar{\epsilon}e^t\rp^k}\prod_{s=0}^{k-1}
		%	\frac{1-q^{s-k}}{1-\frac{\bar{\epsilon}e^tq^s}{(q^k-1)(1-\epsilon e^t)+\bar{\epsilon}e^t}}.
	\end{equation}
\end{lemma}
\begin{proof}
	\begin{align}
		\E\lp e^{tB}\rp = \E\lp e^{t\sum_{s=0}^{k-1} L_s}\rp &= \prod_{s=0}^{k-1} \E\lp e^{tL_s}\rp\nn
								     &= {\prod_{s=0}^{k-1} \sum_{i=1}^\infty e^{ti}
								     (1-p_s)^{i-1}p_s}\nn
								     &= \prod_{s=0}^{k-1} \frac{p_se^t}{1-(1-p_s)e^t},
	\end{align}
	where the {second} equality {follows from the fact that	 $(L_s)_{0\leq s<k}$} are mutually independent. Replacing $p_s$ by its expression
	$p_s=\frac{\bar{\epsilon}q^s(q^{k-s}-1)}{q^k-1}$, we obtain \eqref{eq:eq_ch8_diff_alph_mgf_B}.
\end{proof}
\begin{corollary}
	\label{cor:cor_ch8_diff_alph_expected_B}
	The expected value of $B$ is 
	\begin{equation}
		\label{eq:eq_ch8_diff_alph_EB}
		\E(B) = \frac{q^k-1}{1-\epsilon}\sum_{s=0}^{k-1}\frac{1}{q^k-q^s}.
	\end{equation}
\end{corollary}
\begin{proof}
	Using \eqref{eq:eq_ch8_diff_alph_B_cor}, we get
	\begin{equation}
\E(B) = \sum_{s=0}^{k-1} \E(L_s)=\sum_{s=0}^{k-1} \frac{1}{p_s}=\frac{q^k-1}{1-\epsilon}\sum_{s=0}^{k-1}\frac{1}{q^k-q^s}.	
	\end{equation}
We can also get \eqref{eq:eq_ch8_diff_alph_EB} by using \eqref{eq:eq_ch8_diff_alph_mgf_B} and the fact that
$\displaystyle\E(B) = \left.\frac{\mathrm{d}\phi_B(t)}{\mathrm{d}t}\right|_{t=0}$.
\end{proof}

\subsubsection{Packet Erasure Probability}
{The $l^{th}$ packet is correctly received if $B_l\leq n$.} Otherwise, we declare
the packet to be lost.  {Therefore, the packet erasure probability $\epsilon_p$ is equal to}
\begin{equation}
	\label{eq:eq_ch8_diff_alph_ep}
	{
	\epsilon_p = \Prob(B>n) = \sum_{i=n+1}^\infty \Prob(B=i),
	}
\end{equation}
where the distribution of $B$ is given by \Cref{lemma:lemma_ch8_diff_alph_distribution_B}. We call $1-\epsilon_p=\Prob(B\leq
n)$ to be the packet success probability.

\subsubsection{The Age Analysis}
\begin{definition}
	\label{def:def_ch8_diff_alph_H}
	In every interdeparture interval $Y_j$, we call $H_j$ the number of erased packets before the reception of a successful
	update. $H_j$ is geometric with success probability $\epsilon_p$, so
	\begin{equation}
	\Prob(H_j=l) = \epsilon_p^l(1-\epsilon_p),\quad l=0,1,2,\ldots
	\end{equation}
\end{definition}
We use \Cref{def:def_ch8_diff_alph_H} to characterize the interdeparture interval. Indeed,  any interdeparture interval is the
sum of two components: The time {sending} unsuccessful packets followed by the service time of the successful update. Since each
transmitted packet takes $n$ channel uses and $\mu=1$, then the $j^{th}$ interdeparture time can be written as 
\begin{equation}
	\label{eq:eq_ch8_diff_alph_Y_expr}
	Y_j = \frac{n}{\mu}H_j+\frac{n}{\mu} = n(H_j+1),\quad j\geq 1.
\end{equation}
Given that we assume a memoryless erasure channel and independently generated packets, then {$(H_j)_{j\geq1}$ are independent and identically distributed.} Since the interdeparture interval $Y_j$ is {a function of $H_j$, then $(Y_j)_{j\geq1}$ are also independent and identically distributed.} Hence the following lemma:
\begin{lemma}
	\label{lemma:lemma_ch8_diff_alph_renewal}
	The process $R(\tau)=\max\left\{n:t'_{n}\leq\tau\right\}$ is a renewal process with the interdeparture times
	$(Y_j)_{j\geq1}$ being the renewal intervals.
\end{lemma}
The importance of \Cref{lemma:lemma_ch8_diff_alph_renewal} stems from the fact that it shows that %the instantaneous age process
$\Delta_{\epsilon,\mathcal{C}}$ exists and the system is ergodic.%is mean-ergodic (see \Cref{def:def_ch2_mean_ergodic}). 
\begin{lemma}
	\label{lemma:lemma_ch8_diff_alph_avg_age_ergodic}
	{	Almost surely (over the random choice of the $(n,k)$-scheme $\mathcal{C}$, and over the random erasure patterns of the $q$EC($\epsilon$) channels), we have}
	\begin{equation}
		\label{eq:eq_ch8_diff_alph_avg_age_ergodic}
		\Delta_{\epsilon,\mathcal{C}} = \lim_{\tau\to\infty}\frac{1}{\tau}\int_0^\tau
		\Delta_{\epsilon,\mathcal{C}}(t)\mathrm{d}t = \frac{\E(Q)}{\E(Y)},
	\end{equation}
	where $Q$ is {a generic random variable that has the same distribution as} $\displaystyle Q_j=\int_{t'_{j-1}}^{t'_j}\Delta_{\epsilon,\mathcal{C}}(t)\mathrm{d}t$ which is represented by the shaded areas in Fig.~\ref{fig:fig_ch8_instantaneous_age},
	and $Y$ is {a generic random variable that has the same distribution as} the interdeparture interval $Y_j$.
%	{	More precisely,
%		\begin{equation}
%		\Prob\left(\left\{\lim_{\tau\to\infty}\frac{1}{\tau}\int_0^\tau
%		\Delta_{\epsilon,\mathcal{C}}(t)\mathrm{d}t = \frac{\E(Q)}{\E(Y)}\right\}\right)=1.
%	\end{equation}
%	}
\end{lemma}
\begin{proof}
	%We will use the DTA introduced in \Cref{sec:sec_ch2_DTA} to compute the average age.
	By \Cref{lemma:lemma_ch8_diff_alph_renewal}, $R(\tau)$ forms a renewal process and hence by \cite{ross} we know that
	$\displaystyle\lim_{\tau\to\infty}\frac{R(\tau)}{\tau}=\frac{1}{\E(Y)}$. By defining $\displaystyle Q_j=\int_{t'_{j-1}}^{t'_j}\Delta_{\epsilon,\mathcal{C}}(t)\mathrm{d}t$ to be the reward function over
	the renewal period $Y_j$, we get (using renewal reward theory \cite{Gallager1996,ross}) that {	 almost surely}
	%\[\resizebox{0.48\textwidth}{!}{$\displaystyle\Delta_{\epsilon,\mathcal{C}} = \lim_{\tau\to\infty} \frac{1}{\tau}\int_0^\tau \Delta_{\epsilon,\mathcal{C}}(t)\mathrm{d}t =
	%\lim_{\tau\to\infty} \frac{R(\tau)}{\tau}\frac{1}{R(\tau)}\sum_{j=1}^{R(\tau)}Q_j=\frac{\E(Q_j)}{\E(Y_j)}<\infty.$}\]
	\begin{align}
		\Delta_{\epsilon,\mathcal{C}}	&= \lim_{\tau\to\infty} \frac{1}{\tau}\int_0^\tau
		\Delta_{\epsilon,\mathcal{C}}(t)\mathrm{d}t\nn
						&=\lim_{\tau\to\infty}
						\frac{R(\tau)}{\tau}\frac{1}{R(\tau)}\sum_{j=1}^{R(\tau)}Q_j=\frac{\E(Q_j)}{\E(Y_j)}<\infty.
	\end{align}
\end{proof}
Before computing the average age, we still need one more lemma that gives the distribution of the random variables $(T_j)_{j\geq 1}$.
\begin{lemma}
	\label{lemma:lemma_ch8_diff_alph_dist_T}
	Let $T$ be {a generic random variable that has the same distribution as} the number of channel uses $T_j$ between the decoding instant of the
	$j^{th}$ successful packet and its generation time $t_j$. Then,
	\begin{equation}
		\label{eq:eq_ch8_diff_alph_dist_T}
		{\Prob(T=i) = \frac{\Prob(B=i)\mathbbm{1}_{\{i\leq n\}}}{\Prob(B\leq n)},}
	\end{equation}
	where $\mathbbm{1}_{\{.\}}$ is the indicator function.
\end{lemma}
\begin{proof}
	A packet is successfully decoded if the decoder receives exactly $k$ linearly independent coded symbols {after at most}
	$n$ channel uses. Thus, for the $j^{th}$ successful packet we have that
	\begin{equation}
	\Prob(T_j=i) =\Prob(B=i|B\leq n).
	\end{equation}
\end{proof}
We are now ready to give the main theorem of this section.
\begin{thm}
	\label{thm:thm_ch8_diff_alph_avg_age_expr}
	Assume {a $q$EC($\epsilon$)} and an $(n,k)$-coding scheme $\mathcal{C}$ as defined in
	\Cref{subsec:subsec_ch8_random_codes}. {Almost surely, the} average age $\Delta_{\epsilon,\mathcal{C}}$ corresponding to such setup is
	given by
	\begin{equation}
		\label{eq:eq_ch8_diff_alph_avg_age_expr}
{		\Delta_{\epsilon,\mathcal{C}} = \E(T)-1+\frac{n(1+\epsilon_p)}{2(1-\epsilon_p)}+[-t_0^s],}
	\end{equation}
	where $\epsilon_p$ is the packet erasure probability given by \eqref{eq:eq_ch8_diff_alph_ep}.
\end{thm}
\begin{proof}
	From \eqref{eq:eq_ch8_diff_alph_avg_age_ergodic}, we know that we need to compute $\E(Q)$ and $\E(Y)$.
	We start with $\E(Y)$. We have shown that for {every} $j\geq 1$, $Y_j=n(H_j+1)$. Thus, {
	\begin{equation}
		\label{eq:eq_ch8_diff_alph_EY}
		\E(Y)=\E(Y_j)= n(\E(H_j)+1) = n\lp\frac{\epsilon_p}{1-\epsilon_p}+1\rp = \frac{n}{1-\epsilon_p},
	\end{equation}
	}
	where the {third} equality is due to the fact that $H$ has a geometric distribution with success probability
	$\epsilon_p$ as seen in \Cref{def:def_ch8_diff_alph_H}. 

	Now we turn to $\E(Q)$. For {every} $j\geq 1$, the shaded area $Q_j$ shown in Fig.~\ref{fig:fig_ch8_instantaneous_age} is
	the sum of the areas of two trapezoids: a large trapezoid with height $n(H_j+T_j)$ and a smaller one with height
	$n-T_j$. Recall from \eqref{eq:eqAgeAtTheEndOfSuccess} that the instantaneous age at the end of transmission of the $j^{th}$ successful package is $\Delta_{\epsilon,\mathcal{C}}(t'_j)=n-1+[-t_0^s]$.
	 Thus,
	{
	\begin{align}
		Q_j &=\frac{(n-1+[-t_0^s]+n-1+[-t_0^s]+nH_j+T_j)(nH_j+T_j)}{2}\nn
			&\quad+\frac{(T_j-1+[-t_0^s]+n-1+[-t_0^s])(n-T_j)}{2}\nn
		&= \frac{1}{2}\lp 2n(n-1)H_j+2nT_j(1+H_j)+n^2H_j^2+n(n-2)\rp+[-t_0^s]n(H_j+1)\nn
		&= \frac{1}{2}\lp 2n(n-1)H_j+2nT_j(1+H_j)+n^2H_j^2+n(n-2)\rp+[-t_0^s]Y_j.
	\end{align}
	}
	Note that $H_j$ and $T_j$ are independent. Therefore,
	{
	\begin{align}
 \E(Q_j) &= 	\frac{1}{2}\E\Big( 2n(n-1)H_j+2nT_j(1+H_j)+n^2H_j^2+n(n-2) \Big)+[-t_0^s]\E(Y_j) 	\nonumber\nn
	&= n(n-1)\E(H_j)+n\E(T_j(1+H_j))+\frac{n^2\E\lp H_j^2\rp}{2}+ \frac{n(n-2)}{2}+[-t_0^s]\E(Y)\nonumber\nn
	&=n(n-1)\E(H_j)+n\E(T_j)\E(1+H_j)+\frac{n^2\E\lp H_j^2\rp}{2}+\frac{n(n-2)}{2}+[-t_0^s]\E(Y)\nonumber\nn
	&=n(n-1)\E(H_j)+\E(T_j)\E(Y)+\frac{n^2\E\lp H_j^2\rp}{2}+\frac{n(n-2)}{2}+[-t_0^s]\E(Y)\nonumber\nn
	&=n(n-1)\frac{\epsilon_p}{1-\epsilon_p}+\E(T_j)\E(Y)+\frac{n^2\epsilon_p(1+\epsilon_p)}{2(1-\epsilon_p)^2}+\frac{n(n-2)}{2}+[-t_0^s]\E(Y),		\label{eq:eq_ch8_diff_alph_EQ}
	\end{align}
	}
	Replacing $\E(Y)$ and $\E(Q)$ in \eqref{eq:eq_ch8_diff_alph_avg_age_ergodic} by their expressions in
	\eqref{eq:eq_ch8_diff_alph_EY} and \eqref{eq:eq_ch8_diff_alph_EQ}, we obtain \eqref{eq:eq_ch8_diff_alph_avg_age_expr}.
\end{proof}
In the expression of $\Delta_{\epsilon,\mathcal{C}}$ {in \eqref{eq:eq_ch8_diff_alph_avg_age_expr}}, $\E(T)$ and $\epsilon_p$ cannot be easily expressed in {terms} of
$\epsilon$, $k$ and $n$. {This} is why we study $\Delta_{\epsilon,\mathcal{C}}$ in the next two subsections by presenting upper
and lower bounds on the expression in \eqref{eq:eq_ch8_diff_alph_avg_age_expr}. 

\subsection{Bounding $\Delta_{\epsilon,\mathcal{C}}$}
\label{subsec:subsec_ch8_diff_alph_bounding_age}
{As we mentioned in the previous paragraph, the expression of $\Delta_{\epsilon,\mathcal{C}}$ is not easy to calculate. This is mainly because the distribution of the random variable $B$ is complicated. In this section, we provide upper and lower bounds on $\Delta_{\epsilon,\mathcal{C}}$ which are computed using random variables that have simpler distributions compared to $B$.}

\begin{definition}
	\label{def:def_ch8_diff_alph_lb_ub}
	We define  $\tilde{B}$ to be the sum of $k$ i.i.d random variables distributed like $L_0$. We also define $\hat{B}$ to
	be the sum of $k$ i.i.d random variables distributed like $L_{k-1}$. Formally,
	\begin{equation}
		\label{eq:eq_ch8_diff_alph_lb_up}
{		\tilde{B} = \sum_{s=0}^{k-1} L_0^{(s)}\quad\text{and}\quad \hat{B}=\sum_{s=0}^{k-1} L_{k-1}^{(s)},}
	\end{equation}
	where $L_0$ is geometrically distributed with success probability $\bar{\epsilon}=1-\epsilon$ and $L_{k-1}$ is also
	geometrically distributed with success probability $p_{k-1} = \frac{\bar{\epsilon}q^{k-1}(q-1)}{q^k-1}$.
\end{definition}

\begin{lemma}
	\label{lemma:lemma_ch8_diff_alph_dist_B_Bhat}
	The random variables $\tilde{B}$ and $\hat{B}$ defined in \Cref{def:def_ch8_diff_alph_lb_ub} are both negative
	binomials with
	\begin{equation}
	\Prob(\tilde{B}= i) = {i-1\choose k-1}(1-\epsilon)^k\epsilon^{i-k},
	\end{equation}
	and
	\begin{equation}
\Prob(\hat{B}= i) = {i-1 \choose k-1}(p_{k-1})^k(1-p_{k-1})^{i-k},	
	\end{equation}
	where $i=k,k+1,k+2,\ldots$
\end{lemma}
\begin{proof}
	$\tilde{B}$ is the sum of $k$ i.i.d geometric random variables with success probability $1-\epsilon$. Similarly,
	$\hat{B}$ is the sum of $k$ i.i.d geometric random variables with success probability $p_{k-1}$.
\end{proof}

{
We will show that the random variables $\tilde{B}$ and $\hat{B}$ can be coupled with the random variable $B$ in such a way that $\tilde{B}\leq B\leq \hat{B}$ with probability 1.

\begin{lemma}
\label{lemma:lemma_diff_alph_couplingB}
Let $B=\sum_{s=0}^{k-1} L_s$, and let $\tilde{B}$ and $\hat{B}$ be as in \Cref{def:def_ch8_diff_alph_lb_ub}. It is possible to couple $B$, $\tilde{B}$ and $\hat{B}$ in such a way that $\tilde{B}\leq B\leq \hat{B}$ with probability 1. More precisely, we can define three random variables $O$, $\tilde{O}$ and $\hat{O}$ on the same probability space such that:
\begin{itemize}
\item $O$, $\tilde{O}$ and $\hat{O}$ have the same distributions as $B$, $\tilde{B}$ and $\hat{B}$, respectively, i.e., for every $i\geq 1$, we have $\Prob(O=i)=\Prob(B=i)$, $\Prob(\hat{O}=i)=\Prob(\hat{B}=i)$ and $\Prob(\tilde{O}=i)=\Prob(\tilde{B}=i)$.
\item $\tilde{O}\leq O\leq \hat{O}$ with probability 1.
\end{itemize}
\end{lemma}
\begin{proof}
The proof can be found in \Cref{subsec:subsec_lemma_diff_alph_couplingB_proof}.
\end{proof}

\begin{corollary}
	\label{corollary:corollary_ch8_diff_alph_coupling}
	Given $B=\sum_{s=0}^{k-1} L_s$ and $\tilde{B}$ and $\hat{B}$ as defined in \Cref{def:def_ch8_diff_alph_lb_ub}, the
	following relations hold for $i\geq k$:
	\begin{enumerate}
		%\item $\Prob(\tilde{B}\leq B)=1$,
		\item $\Prob(\tilde{B}\leq i)\geq\Prob(B\leq i)$,
		\item $\E(\tilde{B})\leq \E(B)$,
		%\item $\Prob(\hat{B}\geq B)=1$,
		\item $\Prob(\hat{B}\leq i)\leq\Prob(B\leq i)$.
		\item $\E(\hat{B})\geq \E(B)$,
	\end{enumerate}
\end{corollary}
\begin{proof}
Let $\tilde{O}, O$ and $\hat{O}$ be as in \Cref{lemma:lemma_diff_alph_couplingB}. Since $O\geq
	\tilde{O}$ with probability 1, we deduce that the event $\{O\leq i\}$ is a subset of the event $\{\tilde{O}\leq i\}$.
	Hence,
\begin{equation}
\Prob(B\leq i) = \Prob(O\leq i)\leq \Prob(\tilde{O}\leq i)=\Prob(\tilde{B}\leq i).
\end{equation}	
	This inequality also implies that $\Prob(B\geq i)\geq \Prob(\tilde{B}\geq i)$. Furthermore, since $O\geq
	\tilde{O}$ with probability 1, we have
	\begin{equation}
\E(B)=\E(O)\geq\E(\tilde{O})=\E(\tilde{B}).	
	\end{equation}

On the other hand, since $\hat{O}\geq
	O$ with probability 1, we deduce that the event $\{\hat{O}\leq i\}$ is a subset of the event $\{O\leq i\}$.
	Hence,
\begin{equation}
\Prob(\hat{B}\leq i) = \Prob(\hat{O}\leq i)\leq \Prob(O\leq i)=\Prob(B\leq i).
\end{equation}	
	This inequality also implies that $\Prob(\hat{B}\geq i)\geq \Prob(B\geq i)$. Furthermore, since $\hat{O}\geq O$ with probability 1, we have
	\begin{equation}
\E(\hat{B})=\E(\hat{O})\geq\E(O)=\E(B).	
	\end{equation}
\end{proof}
}
\Cref{corollary:corollary_ch8_diff_alph_coupling} can be interpreted as follows: $\tilde{B}$ can be seen as the number of channel uses in order to
receive exactly $k$ linearly independent coded symbols when any $k$ coded symbols are linearly independent. This means that
$\tilde{B}$ corresponds to the number of channel uses needed to decode a packet when the encoders of the $(n,k)$-scheme only
use MDS codes.
Hence, $\tilde{B}$ is equivalent to the number of channel uses needed to receive exactly $k$ non-erased coded symbols.
Intuitively, we would expect to need a number $\tilde{B}$ of channel uses to receive $k$ non-erased coded symbols which is smaller than the
number $B$ needed to receive $k$ linearly independent coded symbols. This
explains the intuition behind items {$(1)$ and $(2)$} in \Cref{corollary:corollary_ch8_diff_alph_coupling}. On the opposite side of the
spectrum, $\hat{B}$ can be seen as a worst case scenario since the jump from state $s$ to state $s+1$ in
Fig.~\ref{fig:fig_ch8_diif_alph_codeword_dim} occurs with the smallest possible probability, namely $p_{k-1}$. This discussion leads us to the
idea that $\Delta_{\epsilon,\mathcal{C}}$ could be upper bounded  by the average age corresponding to a coding system with
$\hat{B}$ as the number of channel uses needed to receive exactly $k$ linearly independent coded symbols. Similarly,
$\Delta_{\epsilon,\mathcal{C}}$ could be lower bounded by the average age achieved using only MDS codes with $\tilde{B}$ as the
number of channel uses needed to receive $k$ linearly independent coded symbols.

{
By applying \Cref{lemma:lemma_diff_alph_couplingB}, we can define a sequence of independent and identically distributed triplets $(\tilde{B}_l,B_l,\hat{B}_l)_{l\geq 1}$ such that for every $l\geq 1$, we have:
\begin{itemize}
\item $\tilde{B}_l\leq B_l\leq \hat{B}_l$ with probability 1.
\item $\tilde{B}_l,B_l$ and $\hat{B}_l$ are distributed as $\tilde{B}, B$ and $\hat{B}$, respectively.
\end{itemize}

We will use $(B_l)_{l\geq 1}$ to describe the age of information of the system as we explained at the beginning of \Cref{subsec:subsec_ch8_bounds_opt_age}. More precisely, for $t\geq 1$, we have
\begin{equation}
\label{eq:eq_age_recursive_random}
\Delta_{\epsilon,\mathcal{C}}(t)=\begin{cases}
\Delta_{\epsilon,\mathcal{C}}((l_t-1)n)+t-(l_t-1)n\quad&\text{if }t-(l_t-1)n<B_{l_t},\\
 t-1-(l_t-1)n+[-t_0^s]\quad&\text{if }t-(l_t-1)n\geq B_{l_t},
 \end{cases}
\end{equation}
 where $l_t=\lfloor\frac{t-1}{n}\rfloor+1$ is the number of the packet that is being transmitted at time $t$. Note that \eqref{eq:eq_age_recursive_random} can be shown exactly as \eqref{eq:eq_age_recursive_nonMDS}.
 
We now define two virtual ages, that we denote as $\tilde{\Delta}_{\epsilon,\mathcal{C}}(t)$ and $\hat{\Delta}_{\epsilon,\mathcal{C}}(t)$. These virtual ages are initially equal to the real age $\Delta_{\epsilon,\mathcal{C}}(t)$, but instead of using $(B_l)_{l\geq 1}$, the evolution of $\tilde{\Delta}_{\epsilon,\mathcal{C}}(t)$ and $\hat{\Delta}_{\epsilon,\mathcal{C}}(t)$ will be governed by $(\tilde{B}_l)_{l\geq 1}$ and $(\hat{B}_l)_{l\geq 1}$, respectively. More precisely,

\begin{equation}
\label{eq:eq_age_recursive_randomtilde}
\tilde{\Delta}_{\epsilon,\mathcal{C}}(t)=\begin{cases} 
\Delta_{\epsilon,\mathcal{C}}(t)\quad&\text{if }t<1,\\
\tilde{\Delta}_{\epsilon,\mathcal{C}}((l_t-1)n)+t-(l_t-1)n\quad&\text{if }t\geq 1\text{ and }t-(l_t-1)n<\tilde{B}_{l_t},\\
 t-1-(l_t-1)n+[-t_0^s]\quad&\text{if }t\geq 1\text{ and }t-(l_t-1)n\geq \tilde{B}_{l_t},
 \end{cases}
\end{equation}
and
\begin{equation}
\label{eq:eq_age_recursive_randomhat}
\hat{\Delta}_{\epsilon,\mathcal{C}}(t)=\begin{cases} 
\Delta_{\epsilon,\mathcal{C}}(t)\quad&\text{if }t<1,\\
\hat{\Delta}_{\epsilon,\mathcal{C}}((l_t-1)n)+t-(l_t-1)n\quad&\text{if }t\geq 1\text{ and }t-(l_t-1)n<\hat{B}_{l_t},\\
 t-1-(l_t-1)n+[-t_0^s]\quad&\text{if }t\geq 1\text{ and }t-(l_t-1)n\geq \hat{B}_{l_t}.
 \end{cases}
\end{equation}

Similarly to the proof of \Cref{thm:thm_ch8_MDS_optimal}, since $\tilde{B}_{l_t}\leq B_{l_t}\leq \hat{B}_{l_t}$ for every $t\geq 1$, we can show by induction on $l_t$ that $\tilde{\Delta}_{\epsilon,\mathcal{C}}(t)\leq \Delta_{\epsilon,\mathcal{C}}(t)\leq \hat{\Delta}_{\epsilon,\mathcal{C}}(t)$ for every $t$. Therefore,
\begin{equation}
\label{eq:eq_BoundAgeTilde}
\tilde{\Delta}_{\epsilon,\mathcal{C}}\leq \Delta_{\epsilon,\mathcal{C}}\leq \hat{\Delta}_{\epsilon,\mathcal{C}},
\end{equation}
where
\begin{equation}
\tilde{\Delta}_{\epsilon,\mathcal{C}}=\lim_{\tau\to\infty}\frac{1}{\tau}\int_0^\tau\tilde{\Delta}_{\epsilon,\mathcal{C}}(t)dt,
\end{equation}
and
\begin{equation}
\hat{\Delta}_{\epsilon,\mathcal{C}}=\lim_{\tau\to\infty}\frac{1}{\tau}\int_0^\tau\hat{\Delta}_{\epsilon,\mathcal{C}}(t)dt.
\end{equation}

\subsubsection{Upper Bound on $\Delta_{\epsilon,\mathcal{C}}$}
From \eqref{eq:eq_BoundAgeTilde} we know that $\Delta_{\epsilon,\mathcal{C}}\leq \hat{\Delta}_{\epsilon,\mathcal{C}}$.

Since $\hat{\Delta}_{\epsilon,\mathcal{C}}$ was defined in a similar way as $\Delta_{\epsilon,\mathcal{C}}$ but using $(\hat{B}_l)_{l\geq 1}$ instead of $(B_l)_{l\geq 1}$, $\hat{\Delta}_{\epsilon,\mathcal{C}}$ will satisfy a similar equation as \eqref{eq:eq_ch8_diff_alph_avg_age_expr} but the terms will be defined using $\hat{B}$  instead of $B$. More precisely, by using the same techniques that were used to prove \Cref{thm:thm_ch8_diff_alph_avg_age_expr}, we can show that almost surely, we have
	\begin{equation}
		\label{eq:eq_ch8_diff_alph_avg_age_expr_hat}
\hat{\Delta}_{\epsilon,\mathcal{C}} = \E(\hat{T})-1+\frac{n(1+\hat{\epsilon}_p)}{2(1-\hat{\epsilon}_p)}+[-t_0^s],
	\end{equation}
	where the distribution of $\hat{T}$ is given by
	\begin{equation}
		\label{eq:eq_ch8_diff_alph_dist_T_hat}
	\Prob(\hat{T}=i) = \frac{\Prob(\hat{B}=i)\mathbbm{1}_{\{i\leq n\}}}{\Prob(\hat{B}\leq n)},
	\end{equation}
	and
\begin{equation}
	\label{eq:eq_ch8_diff_alph_ep_hat}
	\hat{\epsilon}_p = \Prob(\hat{B}>n) = \Prob(\hat{B}\geq n+1) = \sum_{i=n+1}^\infty \Prob(\hat{B}=i).
\end{equation}

From \Cref{lemma:lemma_ch8_diff_alph_dist_B_Bhat}, we know that $\hat{B}$ is a negative binomial random variable. Hence,
\begin{align}
	\label{eq:eq_ch8_diff_alph_ub_T_2}
	\E(\hat{T})&= \sum_{i=k}^n i \frac{\Prob(\hat{B}=i)}{\Prob(\hat{B}\leq n)}\nn
	&= \sum_{i=k}^n i\frac{ {i-1 \choose k-1}(1-p_{k-1})^{i-k}{p_{k-1}}^k}{\Prob(\hat{B}\leq n)}\nn
	&= \frac{k}{\Prob(\hat{B}\leq n)}\sum_{i=k}^n {i \choose k}(1-p_{k-1})^{i-k}(p_{k-1})^k. 
\end{align}
Let $\hat{\hat{B}}= \sum_{s=0}^k \hat{L}_{s}$, where $(\hat{L}_{s})_{0\leq s\leq k}$ are i.i.d with a marginal distribution
identical to $L_{k-1}$. Hence $\hat{\hat{B}}$ is also a negative binomial and
\begin{equation}
 \Prob(\hat{\hat{B}}=i) = {i-1 \choose k}(1-p_{k-1})^{i-k-1}(p_{k-1})^{k+1},\quad \forall i\geq k+1.
\end{equation}
We use the same trick as in \cite{RoyNajmSoljaninZhong-ISIT2017} and set $i'=i+1$ in
\eqref{eq:eq_ch8_diff_alph_ub_T_2}. This leads to
\begin{align}
	\label{eq:eq_ch8_diff_alph_ub_T_3}
	\E(\hat{T})&= \frac{k}{\Prob(\hat{B}\leq n)}\sum_{i'=k+1}^{n+1} {i'-1 \choose k}(1-p_{k-1})^{i'-k-1}(p_{k-1})^{k}\nn
	&=\frac{k\Prob(\hat{\hat{B}}\leq n+1)}{p_{k-1}\Prob(\hat{B}\leq n)}, 
\end{align}
where \resizebox{0.43\textwidth}{!}{$\displaystyle\Prob(\hat{\hat{B}}\leq n+1)= \sum_{i=k+1}^{n+1}{i-1 \choose
k}(1-p_{k-1})^{i-k-1}(p_{k-1})^{k+1}$}.

Using this
result, together with \eqref{eq:eq_ch8_diff_alph_avg_age_expr_hat} and \eqref{eq:eq_ch8_diff_alph_ep_hat}, we get
\begin{align}
	\hat{\Delta}_{\epsilon,\mathcal{C}}	&= \frac{k\Prob(\hat{\hat{B}}\leq n+1)}{p_{k-1}\Prob(\hat{B}\leq n)}-1
	+\frac{n\lp1+\Prob(\hat{B}\geq n+1)\rp}{2\lp1-\Prob(\hat{B}\geq n+1)\rp}  +[-t_0^s]\nn
					&= \frac{2np_{k-1}-p_{k-1}\Prob(\hat{B}\leq n)(n+2)+2k\Prob(\hat{\hat{B}}\leq
				n+1)}{2p_{k-1}\Prob(\hat{B}\leq n)}+[-t_0^s],
\end{align}
where the second equality is obtained by using
\begin{equation}
\Prob(\hat{B}\geq n+1)=1-\Prob(\hat{B}\leq n).
\end{equation}
We denote by $\Delta_{\epsilon,(n,k)}^{ub}$ the upper bound we just found. Thus,
\begin{equation}
	\label{eq:eq_ch8_diff_alph_ub_1}
	\Delta_{\epsilon,(n,k)}^{ub} = \frac{2np_{k-1}-p_{k-1}\Prob(\hat{B}\leq n)(n+2)+2k\Prob(\hat{\hat{B}}\leq
n+1)}{2p_{k-1}\Prob(\hat{B}\leq n)}+[-t_0^s].
\end{equation}
	
}

\subsubsection{Lower Bound on $\Delta_{\epsilon,\mathcal{C}}$}
{
Let $\displaystyle\tilde{\tilde{B}}= \sum_{s=0}^k \tilde{L}_s$, where $(\tilde{L}_s)_{0\leq s\leq k}$} are i.i.d with a marginal distribution
identical to $L_0$. Hence $\tilde{\tilde{B}}$ is also a negative binomial and 
{
\begin{equation}
 \Prob(\tilde{\tilde{B}}=i) = {i-1 \choose k}\epsilon^{i-k-1}(1-\epsilon)^{k+1},\quad \forall i\geq k+1.
\end{equation}

From \eqref{eq:eq_BoundAgeTilde}, we know that $\Delta_{\epsilon,\mathcal{C}}\geq \tilde{\Delta}_{\epsilon,\mathcal{C}}$. Using an argument} identical to that used for the computation of the upper bound
$\Delta_{\epsilon,\mathcal{C}}^{ub}$ we show that $\Delta_{\epsilon,\mathcal{C}}\geq \Delta_{\epsilon,(n,k)}^{lb}$, where
\begin{equation}
	\label{eq:eq_ch8_diff_alph_lb}
	\Delta_{\epsilon,(n,k)}^{lb} = \frac{2n(1-\epsilon)-(1-\epsilon)\Prob(\tilde{B}\leq n)(n+2)+2k\Prob(\tilde{\tilde{B}}\leq
n+1)}{2(1-\epsilon)\Prob(\tilde{B}\leq n)}.
\end{equation}
\begin{remark}
	\label{rmk:rmk_ch8_diff_alph_yates}
	The lower bound found here is similar to the average age derived in \cite{RoyNajmSoljaninZhong-ISIT2017} for the finite
	redundancy (FR) case. However, the time scale is different since Yates et al. in \cite{RoyNajmSoljaninZhong-ISIT2017}
	assume {that} the source generates a new update at the same instant it finishes transmitting the previous one. Whereas in our
	case, when $t_0^s=0$, we assume we generate and begin transmitting a new packet $\frac{1}{\mu}$ seconds after the last update finishes
	transmission. 
\end{remark}

\subsection{Age-Optimal Codes}
We have already discussed that the lower bound on $\Delta_{\epsilon,\mathcal{C}}$, $\Delta_{\epsilon,(n,k)}^{lb}$,
corresponds to the average age when the $(n,k)$-scheme uses only MDS codes with $\tilde{B}$ as the number of channel uses needed to receive $k$ linearly independent coded
symbols. Recall from {\Cref{thm:thm_ch8_MDS_optimal}} that, for a given couple $(n,k)$, using an MDS code is optimal.
This observation gives a different explanation on why the expression found in \eqref{eq:eq_ch8_diff_alph_lb} is indeed a lower bound
on the average age corresponding to a scheme using any other type of codes than MDS, in particular a code generated randomly. This means
that the lower bound is universal over all codes and the optimal achievable age 
	{
\begin{equation}
	\label{eq:eq_ch8_diff_alph_opt_age_lb}
	\Delta_{\epsilon,(n,k)}^{ub}\geq\Delta_{\epsilon,\mathcal{C}}\geq\Delta_\epsilon^{lin} \geq \min_{n\geq k}
	\Delta_{\epsilon,(n,k)}^{lb},
\end{equation}

where $\mathcal{C}$ is a random $(n,k)$-scheme, and $\Delta_\epsilon^{lin}$ is the optimal average age over coding schemes that are induced by linear block codes.}
However, for a given $(n,k)$, an explicit construction of an MDS code is not always available. In this section, we show that if
the channel-input alphabet is large enough, then random codes are (almost) age-optimal {among linear block codes}.
\begin{thm}
	\label{thm:thm_ch8_diff_alph_random_code_opt}
	Fix a couple $(n,k)$. We have that  $\forall \delta>0$, $\exists q_0>0$ such that
	$\forall q\geq q_0$, {a random $(n,k)$-coding scheme $\mathcal{C}$ almost surely satisfies}
	\begin{equation}
		|\Delta_{\epsilon,\mathcal{C}}-\Delta_{\epsilon,(n,k)}^{lb}|<\delta.
	\end{equation}
	This means that for a channel-input alphabet large enough ($q$ large), random codes are (almost) age-optimal {among linear block codes} and
	\begin{equation}
		\label{eq:eq_ch8_diff_alph_opt_achievable_age}
		{
		\Delta_\epsilon^{lin}\doteq \min_{n\geq k}\Delta_{\epsilon,\mathcal{C}},
		}
	\end{equation}
	where $\mathcal{C}$ is {a random $(n,k)$-coding scheme,} and the dot above the equal sign refers to the fact that the difference
	between the two sides approaches zero as $q$ gets large. %equality in the asymptotic regime.
\end{thm}
\begin{proof}
	For a given random code $\mathcal{C}$, recall that 
	
	\begin{equation}
\Delta_{\epsilon,\mathcal{C}} = \E(T)-1+\frac{n(1+\epsilon_p)}{2(1-\epsilon_p)}+[-t_0^s].	
	\end{equation}
From \eqref{eq:eq_ch8_diff_alph_dist_T} and \eqref{eq:eq_ch8_diff_alph_ep}, we notice that $\E(T)$ and $\epsilon_p$ both depend only on the distribution of $B=\sum_{s=0}^{k-1} L_s$. However, for {every}
$s\in\{0,1,\ldots,k-1\}$, 
\begin{equation}
 \lim_{q\to\infty} p_s=\lim_{q\to\infty} (1-\epsilon)\frac{q^k-q^s}{q^k-1} = 1-\epsilon=p_0.
\end{equation}
This means that, for {every} $s$, $L_s$ converges in distribution to $L_0$ as $q\to\infty$. Therefore, $B$ converges in
distribution to $\tilde{B}=\sum_{s=0}^{k-1} L_0^{(s)}$, as $q\to\infty$. Hence, as $q\to\infty$,
$\Delta_{\epsilon,\mathcal{C}}$ converges to $\Delta_{\epsilon,(n,k)}^{lb}$. So, for $q$ large enough, we can write
\begin{equation}
 \Delta_{\epsilon,\mathcal{C}}=\Delta_{\epsilon,(n,k)}\doteq \Delta_{\epsilon,(n,k)}^{lb}.
\end{equation}
From \eqref{eq:eq_ch8_diff_alph_upper_bound_opt_age}, we know that the optimal age {among linear block codes}, for a given $q$, is {$\displaystyle\Delta_\epsilon^{lin}\leq
\min_{n\geq k} \Delta_{\epsilon,(n,k)}$}. For large enough $q$, we have
$\Delta_{\epsilon,(n,k)}\doteq\Delta_{\epsilon,(n,k)}^{lb}$. This means that asymptotically,
{$\displaystyle\Delta_\epsilon^{lin}\dot\leq\min_{n\geq k}\Delta_{\epsilon,(n,k)}^{lb}$}. However, from \eqref{eq:eq_ch8_diff_alph_opt_age_lb}, we
have that {$\displaystyle\Delta_\epsilon^{lin} \geq \min_{n\geq k}\Delta_{\epsilon,(n,k)}^{lb}$} for {every} $q$. Therefore, asymptotically
\begin{equation}
 \Delta_\epsilon^{lin}\doteq\min_{n\geq k}\Delta_{\epsilon,(n,k)}^{lb}.
\end{equation}
\end{proof}

Notice that for very large $k$, it is extremely unlikely that a (randomly generated) coded symbol is linearly dependent with any subset of size $k-1$ of the $n-1$ remaining coded symbols. This means that as $k$ becomes large, the behavior of random codes approaches that of MDS codes. This is essentially the main reason why \Cref{thm:thm_ch8_diff_alph_random_code_opt} is true

%\begin{figure}[t]
%	\centerline{
%	%\hspace{-1cm}
%		\includegraphics[scale=0.35,trim={1.6cm 1.6cm 2.2cm 2.2cm},clip]{Age_optimal_bounds_k_3_q_5}}
%	\caption{Bounds on $\Delta_{\epsilon,\mathcal{C}}$ with respect to the blocklength $n$, with $k=3$ and a channel-input alphabet
%	of size $q=5$. The age is in log scale.}
%	\label{fig:fig_ch8_diff_alph_age_opt_bounds}
%\end{figure}

\subsection{Other Bounds and Approximations}
\label{subsec:subsec_ch8_diff_alph_other_bounds}
\subsubsection{Upper Bounding the Lower  Bound}
In \Cref{rmk:rmk_ch8_diff_alph_yates}, we discussed how the lower bound found in \eqref{eq:eq_ch8_diff_alph_lb} is similar, up
to a time scale difference, to the average age computed by Yates et al. in \cite[Section 3]{RoyNajmSoljaninZhong-ISIT2017}. In
this paper, the authors present a tight upper bound on the computed average age. We borrow the same techniques as in
\cite[Section 3.A]{RoyNajmSoljaninZhong-ISIT2017} to upper bound $\Delta_{\epsilon,(n,k)}^{lb}$. Interestingly, simulations will
show that the upper bound to $\Delta_{\epsilon,(n,k)}^{lb}$ is a tight approximation to $\Delta_{\epsilon,\mathcal{C}}$, the
average age achieved when using {a random $(n,k)$-scheme $\mathcal{C}$}.

Recall that 
{
\begin{align}
	\label{eq:eq_ch8_diff_alph_lb_app_1}
	\Delta_{\epsilon,(n,k)}^{lb}	&= \frac{2n(1-\epsilon)-(1-\epsilon)\Prob(\tilde{B}\leq n)(n+2)+2k\Prob(\tilde{\tilde{B}}\leq
n+1)}{2(1-\epsilon)\Prob(\tilde{B}\leq n)}+[-t_0^s]\nn
					&= \frac{k\Prob(\tilde{\tilde{B}}\leq n+1)}{(1-\epsilon)\Prob(\tilde{B}\leq
				n)}-1+\frac{n(2-\Prob(\tilde{B}\leq n))}{2\Prob(\tilde{B}\leq n)}+[-t_0^s].
\end{align}
}
Denote by $\tilde{\mu}_n=\frac{k\Prob(\tilde{\tilde{B}}\leq  n+1)}{(1-\epsilon)\Prob(\tilde{B}\leq n)}$. From \cite[Lemma
1]{RoyNajmSoljaninZhong-ISIT2017}, we know that $\tilde{\mu}_n\leq\min\lp n,\frac{k}{1-\epsilon}\rp$. Hence,
\begin{equation}
	\Delta_{\epsilon,(n,k)}^{lb}\leq \frac{k}{1-\epsilon}-1+\frac{n(2-\Prob(\tilde{B}\leq n))}{2\Prob(\tilde{B}\leq n)}+[-t_0^s].
\end{equation}
We denote by $\Delta_{\epsilon,(n,k)}^{*}$ this approximation. Thus,

\begin{equation}
	\label{eq:eq_ch8_diff_alph_lb_app_2}
	\Delta_{\epsilon,(n,k)}^{*} = \frac{k}{1-\epsilon}-1+\frac{n(2-\Prob(\tilde{B}\leq n))}{2\Prob(\tilde{B}\leq n)}+[-t_0^s].
\end{equation}

\begin{remark}
	We can apply the techniques discussed in \cite[Section 3.A]{RoyNajmSoljaninZhong-ISIT2017} in order to approximate the
	optimal codeword length $n$ for $\Delta_{\epsilon,(n,k)}^{lb}$ and write $\Delta_{\epsilon,(n,k)}^{*}$ solely in function of
	$\epsilon$, $k$, $n$ and the size $q$ of the channel-input alphabet.   
\end{remark}
\begin{figure}[t]
	\centerline{
	%\hspace{-1cm}
		\subfloat[$k=3$ and a channel-input alphabet of size $q=5$. The age is in log scale.\label{fig:fig_ch8_diff_alph_age_opt_bounds}]{\includegraphics[scale=0.35,trim={1.6cm 1.6cm 2.2cm 2.2cm},clip]{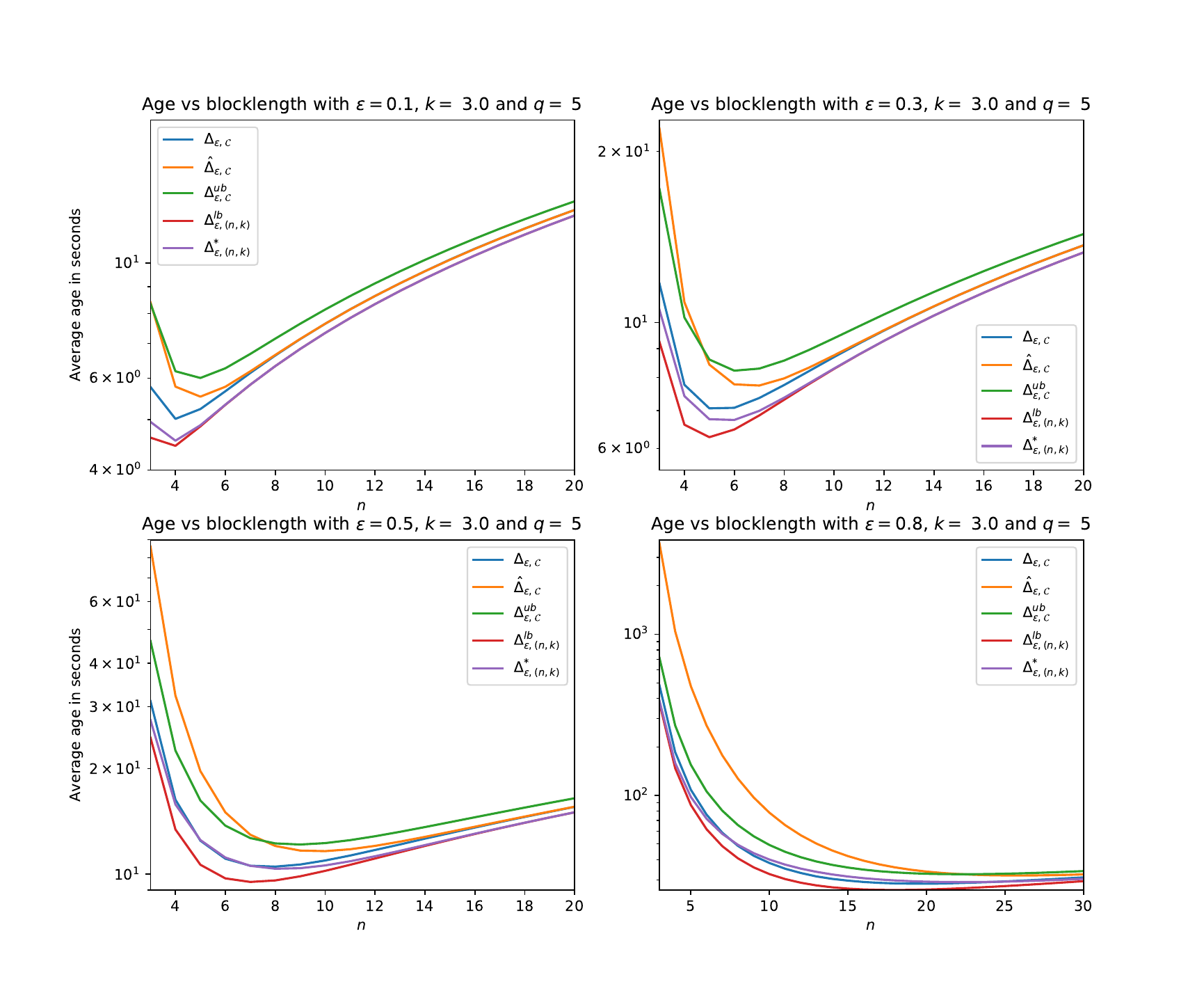}}
		\subfloat[$k=3$ and a channel-input alphabet of size
		$q=25$.\label{fig:fig_ch8_diff_alph_age_opt_bounds_k_10}]{\includegraphics[scale=0.35,trim={1.6cm 1.6cm 2.2cm 2.2cm},clip]{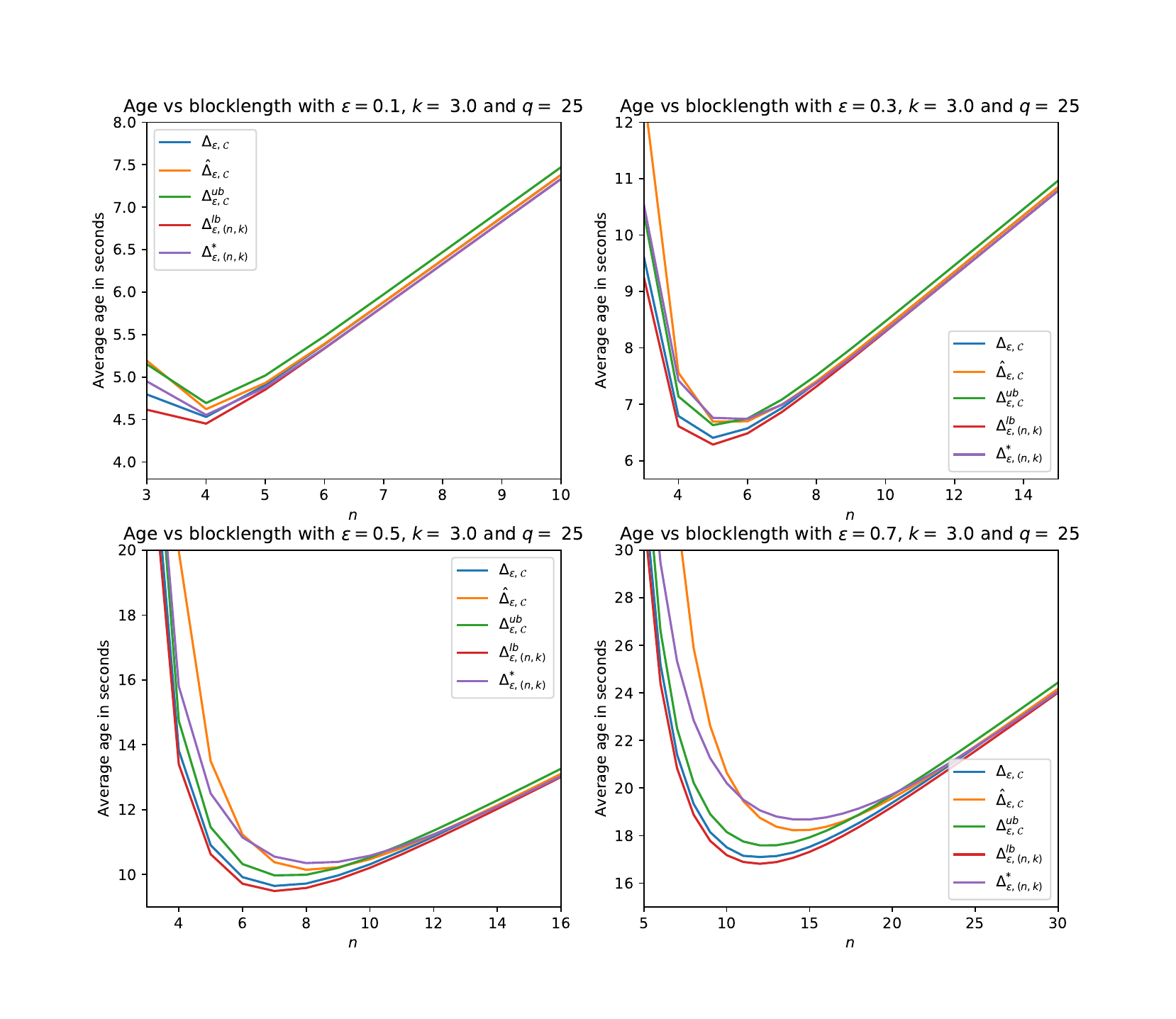}}
	}
	\caption{Bounds on $\Delta_{\epsilon,\mathcal{C}}$ with respect to the blocklength $n$.}
	%\label{fig:fig_ch8_diff_alph_age_opt_bounds_k_10}
\end{figure}

\subsubsection{Another Upper Bound on $\Delta_{\epsilon,\mathcal{C}}$}
We derive here a second upper bound on $\Delta_{\epsilon,\mathcal{C}}$ which is easier to compute than

$\Delta_{\epsilon,(n,k)}^{ub}$. First recall from \Cref{thm:thm_ch8_diff_alph_avg_age_expr} that 
\begin{equation}
\Delta_{\epsilon,\mathcal{C}} = \E(T)-1+\frac{n(1+\epsilon_p)}{2(1-\epsilon_p)}+[-t_0^s].
\end{equation}

However,

\begin{align}
	\E(T)	&= \sum_{i=k}^n i\frac{\Prob(B=i)}{\Prob(B\leq n)}\nn
		&= \frac{1}{\Prob(B\leq n)} \lp \sum_{i=k}^\infty i\Prob(B=i) - \sum_{i=n+1}^\infty i\Prob(B=i)\rp\nn
		&= \frac{1}{\Prob(B\leq n)} \lp \E(B) - \sum_{i=n+1}^\infty i\Prob(B=i)\rp\nn
		&\leq \frac{1}{\Prob(B\leq n)} \lp \E(B) - (n+1)(1-\Prob(B\leq n))\rp.
\end{align}

Hence,
\begin{align}
	\label{eq:eq_ch8_diff_alph_ub_2_1}
	\Delta_{\epsilon,\mathcal{C}}	&\leq \frac{1}{\Prob(B\leq n)} \lp \E(B) - (n+1)(1-\Prob(B\leq n))\rp -1
	+\frac{n(1+\epsilon_p)}{2(1-\epsilon_p)}+[-t_0^s]\nn
					&= \frac{1}{1-\epsilon_p} \lp \E(B) - (n+1)\epsilon_p\rp -1
	+\frac{n(1+\epsilon_p)}{2(1-\epsilon_p)}+[-t_0^s]\nn
					&= \frac{\E(B)-1}{1-\epsilon_p}+\frac{n}{2}+[-t_0^s].
\end{align}

Whereas $\E(B)$ (given in \eqref{eq:eq_ch8_diff_alph_EB}) is easy to compute, 
\begin{equation}
\epsilon_p=\Prob(B\geq n+1)
\end{equation}
is hard to compute
due to the complex nature of the distribution of $B$ (given in \Cref{lemma:lemma_ch8_diff_alph_distribution_B}). To solve this
problem, we use $\hat{B}$ as defined in \Cref{def:def_ch8_diff_alph_lb_ub} to upper bound $\epsilon_p$. Indeed, from {\Cref{corollary:corollary_ch8_diff_alph_coupling}} we know that
\begin{equation}
 \epsilon_p=\Prob(B\geq n+1)\leq \Prob(\hat{B}\geq n+1).
\end{equation}
Hence, 

\begin{equation}
	\Delta_{\epsilon,\mathcal{C}}\leq  \frac{\E(B)-1}{\Prob(\hat{B}\leq n)}+\frac{n}{2}+[-t_0^s].
\end{equation}
Therefore, using  \eqref{eq:eq_ch8_diff_alph_EB}, the new upper bound $\hat{\hat{\Delta}}_{\epsilon,(n,k)}$ is 
\begin{align}
	\label{eq:eq_ch8_diff_alph_ub_2_expr}
	\hat{\hat{\Delta}}_{\epsilon,(n,k)}	&= \frac{\E(B)-1}{\Prob(\hat{B}\leq n)}+\frac{n}{2}+[-t_0^s] \nn
						&= \frac{-1 + \frac{q^k-1}{1-\epsilon}\sum_{s=0}^{k-1}
						\frac{1}{q^k-q^s}}{\Prob(\hat{B}\leq n)} + \frac{n}{2}+[-t_0^s]\nn
						&= \frac{\epsilon + \lp q^k-1\rp\sum_{s=1}^{k-1}
						\lp q^k-q^s\rp^{-1}}{(1-\epsilon)\Prob(\hat{B}\leq n)} + \frac{n}{2}+[-t_0^s].
\end{align}

\subsection{Numerical Results}
\begin{figure*}[t]
	\centerline{
	\subfloat[$q=5$.\label{fig:fig_ch8_diff_alph_bounds_opt_region}]{
		\includegraphics[scale=0.5,trim={0.8cm 0.3cm 1.5cm 0.8cm},clip]{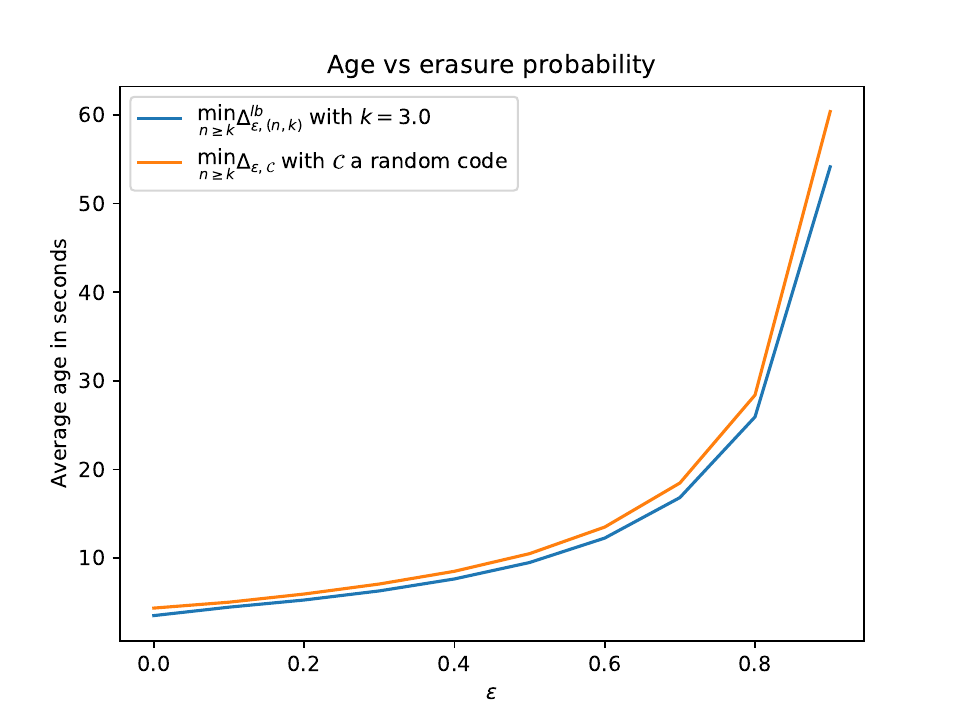}
	}
	\hfill
	\subfloat[$q=25$.\label{fig:fig_ch8_diff_alph_bounds_opt_region_k_10}]{
		\includegraphics[scale=0.5,trim={0.8cm 0.3cm 1.5cm 0.8cm},clip]{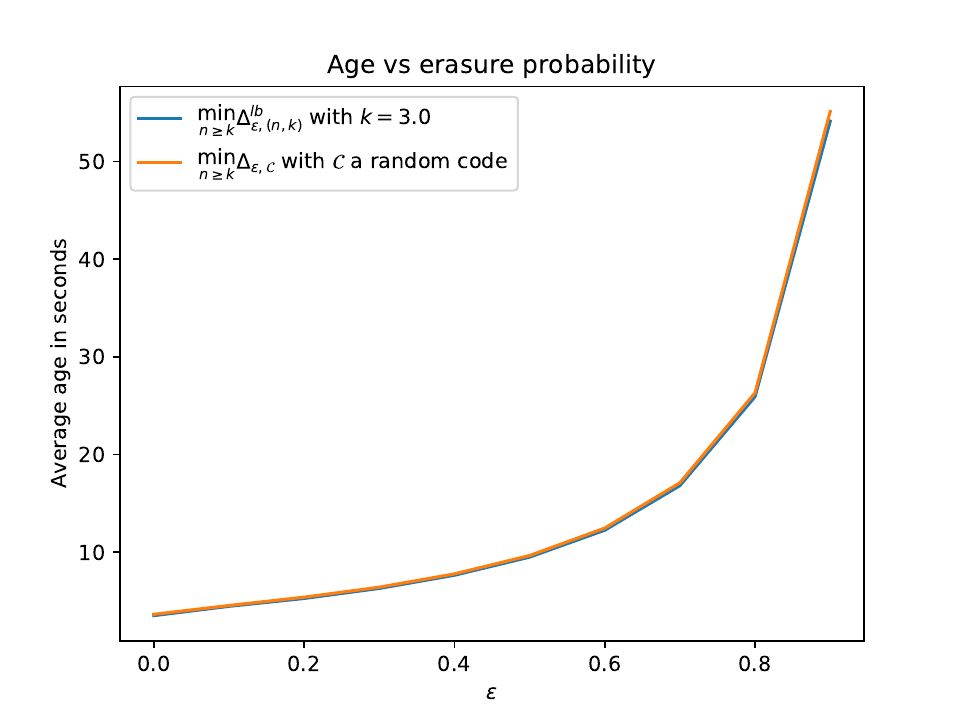}
	}}
	\caption{Bounds on the optimal achievable age $\Delta_\epsilon$	with $k=3$.}
\end{figure*}
\label{subsec:subsec_ch8_numerical_results}
Fig.~\ref{fig:fig_ch8_diff_alph_age_opt_bounds} and Fig.~\ref{fig:fig_ch8_diff_alph_bounds_opt_region} correspond to a system
with {$k=3$, $q=|\mathcal{V}|=5$, $t_0^s=0$, and using a random $(n,k)$-coding scheme} $\mathcal{C}$.
Fig.~\ref{fig:fig_ch8_diff_alph_age_opt_bounds} plots $\Delta_{\epsilon,\mathcal{C}}$ as well as the bounds and the approximation
derived in \Cref{subsec:subsec_ch8_diff_alph_bounding_age,subsec:subsec_ch8_diff_alph_other_bounds} with respect to the
blocklength $n$, for four erasure channels with erasure probabilities $0.1, 0.3, 0.5, 0.8$. The tightness of the bounds with
respect to $\Delta_{\epsilon,\mathcal{C}}$ differs according to the erasure probability: 
\begin{itemize}
	\item For all error probabilities, we notice that the upper bound {$\hat{\hat{\Delta}}_{\epsilon,(n,k)}$} (the orange
		curve) is very tight (almost equal to $\Delta_{\epsilon,\mathcal{C}}$) at large enough $n$. However, the value
		$n^*$ of the blocklength $n$ starting which {$\hat{\hat{\Delta}}_{\epsilon,(n,k)}$} becomes tight depends on $\epsilon$: The larger the
		erasure probability, the larger the blocklength $n$. For instance, for $\epsilon=0.1$ we have $n^*=7$. But for
		$\epsilon=0.5$, $n^*=12$ and for $\epsilon=0.8$ we have $n^*=30$. For $n> n^*$, the upper bound
		{$\hat{\hat{\Delta}}_{\epsilon,(n,k)}$} is tighter than all other bounds. Notice that for {every} $n$ and {every}
		$\epsilon$, $\Delta_{\epsilon,\mathcal{C}}\leq \hat{\Delta}_{\epsilon,\mathcal{C}}$.
	\item For the approximation $\Delta_{\epsilon,(n,k)}^{*}$, we notice that it becomes tighter as the erasure probability
		becomes larger. This is true especially at low values of $n$, more particularly for $n<n^*$. For this range of
		blocklength values the approximation $\Delta_{\epsilon,(n,k)}^{*}$ is the extremely close to
		$\Delta_{\epsilon,\mathcal{C}}$.
	\item For {every $n$ and every $\epsilon$, we have} $\Delta_{\epsilon,(n,k)}^{lb}\leq
		\Delta_{\epsilon,\mathcal{C}}$ and $\Delta_{\epsilon,(n,k)}^{lb}\leq \Delta_{\epsilon,(n,k)}^{*}$. We notice
		that, for all values of $\epsilon$, $\Delta_{\epsilon,(n,k)}^{lb}$ is close to $\Delta_{\epsilon,\mathcal{C}}$
		at large $n$. Whereas, for small values of $n$, this lower bound does not show any noticeable behavioral
		modification as $\epsilon$ increases.
	\item The upper bound {$\Delta_{\epsilon,(n,k)}^{ub}$} is always larger than
		$\Delta_{\epsilon,\mathcal{C}}$. Even though at $n>n^*$ we observe that {$\hat{\hat{\Delta}}_{\epsilon,(n,k)}\leq
		\Delta_{\epsilon,(n,k)}^{ub}$}, for $n\leq n^*$ the upper bound {$\Delta_{\epsilon,(n,k)}^{ub}$} is
		closer to $\Delta_{\epsilon,\mathcal{C}}$ than {$\hat{\hat{\Delta}}_{\epsilon,(n,k)}$. In fact, as $\epsilon$
		increases, the gap between the two upper bounds also increases.}
\end{itemize}
Fig.~\ref{fig:fig_ch8_diff_alph_age_opt_bounds} also suggests that there exists, for each erasure probability, an optimal
blocklength that minimizes $\Delta_{\epsilon,\mathcal{C}}$. This echoes the observations presented in \cite{NajmYatesSoljanin-ISIT2017} and in
\cite{RoyNajmSoljaninZhong-ISIT2017}. Moreover, each bound also has its optimal blocklength. Although the
channel-input alphabet chosen is small ($k=3$ and $q=5$), we remark that the gap between $\Delta_{\epsilon,\mathcal{C}}$
and the lower bound $\Delta_{\epsilon,(n,k)}^{lb}$ is not too great irrespective of the value of $\epsilon$. This means that
even for small channel-input alphabets, the performance of the optimal linear code is not too far from the performance achieved by
random coding. This idea is illustrated in Fig.~\ref{fig:fig_ch8_diff_alph_bounds_opt_region}. In this last figure, we find and
plot, at each value of $\epsilon$, the minimum (with respect to $n$) of $\Delta_{\epsilon,\mathcal{C}}$ and
$\Delta_{\epsilon,(n,k)}^{lb}$. We observe that these two minimums are close to each other. Since
{$\min_{n\geq k}\Delta_{\epsilon,(n,k)}^{lb} \leq \Delta_\epsilon^{lin}\leq \min_{n\geq k}\Delta_{\epsilon,\mathcal{C}}$}, then
Fig.~\ref{fig:fig_ch8_diff_alph_bounds_opt_region} suggests that, for {every} $\epsilon$, if we use the optimal blocklength, then
random codes achieve an age-performance close to the optimal linear code.

Fig.~\ref{fig:fig_ch8_diff_alph_age_opt_bounds_k_10} and Fig.~\ref{fig:fig_ch8_diff_alph_bounds_opt_region_k_10} mirror
Fig.~\ref{fig:fig_ch8_diff_alph_age_opt_bounds} and Fig.~\ref{fig:fig_ch8_diff_alph_bounds_opt_region} respectively, but for a
larger channel-input alphabet with $q=25$. We can apply the same analysis as the one we just presented for the case $q=5$.
%%%%%%
% Add comment about how small q is still good
%%%%
%%with one addition: 
In this case we can notice the effect of increasing the size of the channel-input alphabet, while keeping $k$
constant. In fact, comparing Fig.~\ref{fig:fig_ch8_diff_alph_age_opt_bounds}
and Fig.~\ref{fig:fig_ch8_diff_alph_age_opt_bounds_k_10}, we observe a clear
convergence of $\Delta_{\epsilon,\mathcal{C}}$ toward the lower bound $\Delta_{\epsilon,(n,k)}^{lb}$. In
Fig.~\ref{fig:fig_ch8_diff_alph_age_opt_bounds_k_10}, the
approximation $\Delta_{\epsilon,(n,k)}^{*}$ is not as tight as for the case of $q=5$, for all $\epsilon$ and $n$. Indeed, we
can notice that, for $\epsilon=0.9$,  $\Delta_{\epsilon,(n,k)}^{*}$ is worse than {$\Delta_{\epsilon,(n,k)}^{ub}$} for
$n\leq 20$. For large
$n$, all bounds are tight except for the upper bound {$\Delta_{\epsilon,(n,k)}^{ub}$}. In fact, in
Fig.~\ref{fig:fig_ch8_diff_alph_age_opt_bounds_k_10}, the lower bound $\Delta_{\epsilon,(n,k)}^{lb}$ is the tightest bound on
$\Delta_{\epsilon,\mathcal{C}}$. However, the convergence of $\Delta_{\epsilon,\mathcal{C}}$ toward the lower bound
$\Delta_{\epsilon,(n,k)}^{lb}$ is best observed in Fig.~\ref{fig:fig_ch8_diff_alph_bounds_opt_region_k_10}. In this figure, we
remark that the performance of the random code with the optimal blocklength is almost optimal.
These simulations support our claim that random codes are age-optimal as $q$ grows and the channel-input alphabet becomes large.
%%%%%%%%%%%%%%%%%%%%%%%%%%%%%%%%%%%%%%%%%%%%%%%%%%%%%%%%%%%%%%%%%%%%%
\section{Conclusion}
\label{sec:sec_ch8_conclusion}
In this paper, we have studied the optimal achievable average age over an erasure channel in two scenarios: $(i)$ %In the first scenario
%we have considered 
When the source alphabet and channel-input alphabet are be the same, and $(ii)$ when %we have
they are different. We have demonstrated that in the first case, we do not need
any type of channel coding to achieve the minimal average age, for which  we have computed the exact expression. As for the second
case, we have used random coding technique to compute bounds on the optimal achievable age. We have also
shown that {for a large enough source alphabet, random codes are (almost) age-optimal among linear block codes}. Finally, the numerical results have pointed out an interesting
observation: Even for a small source alphabet, the performance of random codes is not too far from optimal from an age point of
view.
%In this paper, we have studied the optimal achievable average age over an erasure channel in two scenarios: In the first scenario
%we have considered the source alphabet and channel-input alphabet to be the same. Whereas, in the second scenario, we have
%assumed the source alphabet to be different than the channel-input alphabet. We have demonstrated that in the first case, we do not need
%any type of channel coding to achieve the minimal average age, for which  we have computed the exact expression. As for the second
%case, we have first made use of the random coding technique to compute bounds on the optimal achievable age. However, we have
%then shown that random codes are age-optimal for large enough source alphabet. Finally, the numerical results have pointed out an interesting
%observation: Even for a small source alphabet, the performance of random codes is not too far from optimal from an age point of
%view.
%%%%%%%%%%%%%%%%%%%%%%%%%%%%%%%%%%%%%%%%%%%%%%%%%%%%%%%%%%%%%%%%%%%%%
\section*{Acknowledgements}

We would like to thank Roy Yates and an anonymous reviewer for helpful comments. This research was supported in part by grant No.  200021\_166106/1 of the Swiss National Science Foundation.


\appendix

%\section{On the Equidistribution Theory}
%\label{sec:sec_ch8_appendix}
\subsection{Equidistribution and Weyl's Equidistribution Theorem}
\label{subsec:subsec_ch8_weyl_thm}

In this section\footnote{The material in this section is based on \cite{Weyl1916,chandrasekharan1968}.}, for {every} real number
$x$, we use $[x]$ to denote its fractional part, i.e., $[x]=x-\lfloor x\rfloor$.
\begin{definition}
	\label{def:def_ch8_equidistribution}
	A sequence $(u_i)_{i\geq1} \in[0,1)$ is said to be equidistributed on $[0,1)$ if for {every} interval
	$(a,b)\subset [0,1]$ we have 
	\begin{equation}
		\label{eq:eq_ch8_equidistribution_def}
		\lim_{N\to\infty}\frac{1}{N} \left|\Big\{1\leq i\leq N:\; u_i\in (a,b)\Big\}\right|=b-a,
	\end{equation}
	where $|A|$ denotes the cardinality of the set $A$.
\end{definition}
\begin{remark}
\label{rmk:rmk_ch8_equidistribution_def}
In \Cref{def:def_ch8_equidistribution}, we can replace $(a,b)$ with $[a,b)$, $(a,b]$ or $[a,b]$ in
	\eqref{eq:eq_ch8_equidistribution_def} and the limit still holds.
	\end{remark}

\begin{thm}
	\label{thm:thm_ch8_app_weyl_equidistribution_criterion}
	Let $(u_i)_{i\geq1}$ be a sequence of real numbers and denote by $[u_i]=u_i-\lfloor u_i\rfloor$ the fractional part of
	$u_i$. Then the following are equivalent:
	\begin{enumerate}
		\item The sequence $([u_i])_{i\geq1}$ is equidistributed on $[0,1{[}$.
		\item For {every} $k\in\mathbb{N}^*$,
			\begin{equation}
				\label{eq:eq_ch8_app_exp_criterion}
				\lim_{N\to\infty} \frac{1}{N} \sum_{i=1}^N e^{j2\pi ku_i}=0,
			\end{equation}
			where $j^2 = -1$.
		\item For {every} Riemann-integrable function $f:[0,1]\rightarrow \mathbb{C}$, we have
			\begin{equation}
				\label{eq:eq_ch8_app_integral_criterion}
				\lim_{N\to\infty} \frac{1}{N}\sum_{i=1}^N f([u_i]) = \int_0^1 f(x)\mathrm{d}x.
			\end{equation}
		\end{enumerate}
\end{thm}
The proof of \Cref{thm:thm_ch8_app_weyl_equidistribution_criterion} is outside the scope of this paper but we encourage the
reader to check \cite{chandrasekharan1968} for the full proof. An important application of this theorem is given next.

\begin{corollary}
	\label{corr:corr_ch8_app_equidistribution_average}
 If $(u_i)_{i\geq 1}$ is a sequence that is equidistributed over $[0,1)$, then we have
\begin{equation}
\lim_{N\to\infty}\frac{1}{N}\sum_{i=1}^N u_i=\frac{1}{2}.
\end{equation}

\end{corollary}
\begin{proof}
From the third condition of \Cref{thm:thm_ch8_app_weyl_equidistribution_criterion} we have
\begin{equation}
\lim_{N\to\infty}\frac{1}{N}\sum_{i=1}^N u_i=\int_0^1 x\mathrm{d}x=\frac{1}{2}.
\end{equation}
\end{proof}

\subsection{A variation of the strong law of large numbers}
In this section, we prove a well known variation of the strong law of large numbers.

\begin{lemma}
\label{lemma:lemma_as_convergence_to_zero}
If $(X_i)_{i\geq 1}$ is a sequence of complex-valued random variables satisfying $\displaystyle\sum_{i=1}^\infty\E(|X_i|^2)<\infty$, then almost surely, we have $\displaystyle\lim_{i\to\infty} X_i=0$.
\end{lemma}
\begin{proof}
Observe that $\displaystyle \E\left(\sum_{i=1}^\infty |X_i|^2\right)=\sum_{i=1}^\infty\E\left( |X_i|^2\right)<\infty$. Therefore, $\displaystyle \Prob\left(\sum_{i=1}^\infty |X_i|^2=\infty\right)=0$ which can be rewritten as $\displaystyle \Prob\left(\sum_{i=1}^\infty |X_i|^2<\infty\right)=1$. This implies that $\displaystyle \Prob\left(\lim_{i\to\infty} |X_i|^2=0\right)=1$. We conclude that almost surely, we have $\displaystyle\lim_{i\to\infty} X_i=0$.
\end{proof}

\begin{proposition}
\label{prop:prop_slln}
Let $(X_i)_{i\geq 1}$ be a sequence of complex-valued random variables. If there exists $0<C<\infty$ such that
\begin{equation}
 \forall l\geq 1,\quad \sum_{i=1}^{\infty} |\E(X_iX_l^\ast)|\leq C,
\end{equation}
then almost surely
\begin{equation}
\lim_{N\to\infty}\frac{1}{N}\sum_{i=1}^N X_i=0.
\end{equation}
\end{proposition}
\begin{proof}
Let $S_0=0$ and for every $N\geq 1$, let 
\begin{equation}
S_N=\sum_{i=1}^NX_i.
\end{equation}

For every $N_2>N_1\geq 0$, we have
\begin{align}
\E\left(\left|S_{N_2}-S_{N_1}\right|^2\right)
&=\E\left(\left(\sum_{i_1=N_1+1}^{N_2} X_{i_1}\right)\left(\sum_{i_2=N_1+1}^{N_2} X_{i_2}\right)^{\ast}\right)= \sum_{i_1=N_1+1}^{N_2}\sum_{i_2=N_1+1}^{N_2}\E\left(X_{i_1}X_{i_2}^\ast\right)\nn
&\leq \sum_{i_1=N_1+1}^{N_2}\sum_{i_2=N_1+1}^{N_2}\left|\E\left(X_{i_1}X_{i_2}^\ast\right)\right|\leq \sum_{i_1=N_1+1}^{N_2}\sum_{i_2=1}^{\infty}\left|\E\left(X_{i_1}X_{i_2}^\ast\right)\right|\leq \sum_{i_1=N_1+1}^{N_2}C= C\cdot(N_2-N_1).\label{eq:eq_expectation_diff_squared_bound}
\end{align}

In particular, for every $N\geq 1$, we have
\begin{equation}
\E\left(\left|S_N\right|^2\right)=\E\left(\left|S_N-S_{0}\right|^2\right)\leq C N.
\end{equation}

Therefore,
\begin{equation}
\sum_{N=1}^{\infty}\E\left(\left|\frac{S_{N^2}}{N^2}\right|^2\right)\leq\sum_{N=1}^\infty\frac{C N^2}{N^4}=\sum_{N=1}^\infty\frac{C }{N^2}<\infty.
\end{equation}

Lemma \ref{lemma:lemma_as_convergence_to_zero} now implies that 
\begin{equation}
\label{eq:eq_subsequence_converges}
\text{a.s.},\quad\lim_{N\to\infty}\frac{S_{N^2}}{N^2}=0.
\end{equation}

Now, for every $N\geq 1$, define
\begin{equation}
D_N=\sup_{N^2\leq i<(N+1)^2}|S_i-S_{N^2}|.
\end{equation}

We have
\begin{align}
\E(D_N^2)
&=\E\left(\sup_{N^2\leq i<(N+1)^2}|S_i-S_{N^2}|^2\right)\leq\E\left(\sum_{i=N^2+1}^{(N+1)^2-1}|S_i-S_{N^2}|^2\right)\nn
&\stackrel{(\ast)}{\leq} \sum_{i=N^2+1}^{(N+1)^2-1}C(i-N^2)\leq C \sum_{i=N^2+1}^{(N+1)^2-1}((N+1)^2-1-N^2)= 4C N^2,
\end{align}
where $(\ast)$ follows from \eqref{eq:eq_expectation_diff_squared_bound}. Thus,
\begin{align}
\sum_{N=1}^\infty\E\left(\left|\frac{D_N}{N^2}\right|^2\right)\leq \sum_{N=1}^\infty\frac{4C N^2}{N^4}<\infty.
\end{align}
It follows from Lemma \ref{lemma:lemma_as_convergence_to_zero} that
\begin{equation}
\label{eq:eq_subsequence_converges_2}
\text{a.s.},\quad\lim_{N\to\infty}\frac{D_N}{N^2}=0.
\end{equation}

Now observe that for every $N\geq 1$, we have
\begin{align}
\frac{|S_N|}{N}
&\leq \frac{|S_N|}{\lfloor\sqrt{N}\rfloor^2}\leq \frac{\left|S_{\lfloor\sqrt{N}\rfloor^2}\right|}{\lfloor\sqrt{N}\rfloor^2}+\frac{\left|S_N-S_{\lfloor\sqrt{N}\rfloor^2}\right|}{\lfloor\sqrt{N}\rfloor^2}\leq \frac{\left|S_{\lfloor\sqrt{N}\rfloor^2}\right|}{\lfloor\sqrt{N}\rfloor^2}+\frac{D_{\lfloor\sqrt{N}\rfloor}}{\lfloor\sqrt{N}\rfloor^2}.
\end{align}

Equations \eqref{eq:eq_subsequence_converges} and \eqref{eq:eq_subsequence_converges_2} now imply that
\begin{equation}
\text{a.s.},\quad\lim_{n\to\infty}\frac{S_N}{N}=0.
\end{equation}
\end{proof}

\begin{corollary}
\label{corr:corr_slln_exp}
Let $(X_i)_{i\geq 1}$ be a sequence of complex-valued random variables. If there exists $0<C<\infty$ and $0<\beta<1$ such that for every $i,l\geq 1$ we have
\begin{equation}
|\E(X_iX_l^\ast)|\leq C\cdot \beta^{|i-l|},
\end{equation}
then almost surely
\begin{equation}
\lim_{N\to\infty}\frac{1}{N}\sum_{i=1}^N X_i=0.
\end{equation}
\end{corollary}
\begin{proof}
This is a direct corollary of \Cref{prop:prop_slln}.
\end{proof}

%\begin{proposition}
%\label{prop:prop_slln_general_average}
%Let $(X_n)_{n\geq 1}$ be a sequence of complex-valued random variables. Assume that $\E(|X_n|^2)<\infty$ for every $n\geq 1$, and assume that $\displaystyle\lim_{n\to\infty}\frac{1}{n}\sum_{i=1}^n\E(X_i)$ exists. If there exists $0<C<\infty$ such that for every $m\geq 1$ we have
%$$\sum_{n=1}^{\infty} |\E(X_nX_m^\ast)-\E(X_n)\E(X_m)^\ast|\leq C,$$
%then almost surely
%$$\lim_{n\to\infty}\frac{1}{n}\sum_{i=1}^n X_i=\lim_{n\to\infty}\frac{1}{n}\sum_{i=1}^n\E(X_i).$$
%\end{proposition}
%\begin{proof}
%Apply \Cref{prop:prop_slln} to $\tilde{X}_i=X_i-\E(X_i)$.
%\end{proof}
%
%\begin{corollary}
%\label{corr:corr_slln_exp}
%Let $(X_n)_{n\geq 1}$ be a sequence of complex-valued random variables. Assume that $\E(|X_n|^2)<\infty$ for every $n\geq 1$, and assume that $\displaystyle\lim_{n\to\infty}\frac{1}{n}\sum_{i=1}^n\E(X_i)$ exists.  If there exists $0<C<\infty$ and $0<r<1$ such that for every $n,m\geq 1$ we have
%$$\sum_{n=1}^{\infty} |\E(X_nX_m^\ast)-\E(X_n)\E(X_m)^\ast|\leq C\cdot r^{|n-m|},$$
%then almost surely
%$$\lim_{n\to\infty}\frac{1}{n}\sum_{i=1}^n X_i=\lim_{n\to\infty}\frac{1}{n}\sum_{i=1}^n\E(X_i).$$
%\end{corollary}

\subsection{Proof of \Cref{lemma:lemma_ch8_weyl_equidistribution}}
\label{subsec:subsec_ch8_app_weyl_equi_proof}
{
Let $(X_l)_{l\geq 1}$ be a sequence of independent and identically distributed random variables which take values in the set of strictly positive natural numbers $\mathbb{N}^\ast$ and which satisfy $\E(X_l^2)=\E(X^2)<\infty$. Let $S_0=0$ and $\displaystyle S_l=\sum_{r=1}^lX_r$ for $l\geq 1$. For every $i\geq 0$, let
\begin{equation}
L_i=\max\left\{l\geq 0: S_l\leq i\right\},
\end{equation}
and
\begin{equation}
Y_i=\max\left\{S_l:l\geq 0\text{ and }S_l\leq i\right\}.
\end{equation}
Clearly, we have $Y_i=S_{L_i}$. Furthermore, since $X_r\geq 1$ for every $r\geq 1$, we have $L_i\leq i$ with probability 1.

Let $\rho\in\mathbb{R}\setminus\mathbb{Q}$ be an irrational number, and let $\alpha\in\mathbb{R}$ be an arbitrary real number. From \Cref{corr:corr_ch8_app_equidistribution_average} and the second criterion of \Cref{thm:thm_ch8_app_weyl_equidistribution_criterion}, we know that in order to show that
		\begin{equation}
%		\label{eq:eq_ch8_weyl_equidistribution}
\text{a.s.,}\quad		\lim_{N\to\infty} \frac{1}{N}\sum_{i=1}^{N} [\rho Y_i+\alpha] = \frac{1}{2},
	\end{equation}
it is sufficient to show that
\begin{equation}
\text{a.s.,}\quad\lim_{N\to\infty}\frac{1}{N}\sum_{i=1}^N e^{j2\pi k [\rho Y_i +\alpha]}=0,\quad\forall k\in\mathbb{N}^\ast.
\end{equation}

Now fix $k\in\mathbb{N}^{\ast}$. For every $N\geq 1$, we have
\begin{align}
\frac{1}{N}\sum_{i=1}^N e^{j2\pi k [\rho Y_i +\alpha]}&=\frac{1}{N}\sum_{i=1}^N e^{j2\pi k (\rho Y_i +\alpha)}=\frac{e^{j2\pi k\alpha}}{N}\sum_{i=1}^N e^{j2\pi k \rho Y_i}\nn
&=e^{j2\pi k\alpha}\left(\frac{Y_N}{N}\frac{1}{Y_N}\sum_{i=0}^{Y_N-1} e^{j2\pi k \rho Y_i}-\frac{1}{N}+\frac{1}{N}\sum_{i=Y_N}^{N} e^{j2\pi k \rho Y_i}\right).\label{eq:eq_equidist_rewrite_condition}
\end{align}

We would like to show that almost surely $\displaystyle\lim_{N\to\infty}\frac{Y_N}{N}=1$. First, observe that
\begin{equation}
\sum_{l=1}^{\infty}\E\left(\left|\frac{X_l}{l}\right|^2\right)=\sum_{l=1}^{\infty}\frac{\E\left(X_l^2\right)}{l^2}=\E\left(X^2\right)\sum_{l=1}^{\infty}\frac{1}{l^2}<\infty.
\end{equation}
It follows from \Cref{lemma:lemma_as_convergence_to_zero} that almost surely $\displaystyle\lim_{l\to\infty}\frac{X_l}{l}=0$.

It is easy to see that as $N\to\infty$, we have $L_N\to\infty$ and $Y_N\to\infty$. Now since $L_N\leq N$ with probability 1, we have
\begin{equation}
\text{a.s.},\quad0\leq\lim_{N\to\infty}\frac{X_{L_N+1}}{N}\leq \lim_{N\to\infty}\frac{X_{L_N+1}}{L_N}=\lim_{l\to\infty}\frac{X_l}{l-1}=0.
\end{equation}

Furthermore, since $Y_N\leq N< Y_N+X_{L_N+1}$, and since we have just showed that $\displaystyle\lim_{N\to\infty}\frac{X_{L_N+1}}{N}=0$, it follows that

\begin{equation}
\label{eq:eq_equidist_rewrite_condition_1}
\text{a.s.,}\quad\lim_{N\to\infty}\frac{Y_N}{N}=1.
\end{equation}

Now observe that

\begin{align}
\left|\frac{1}{N}\sum_{i=Y_N}^{N} e^{j2\pi k \rho Y_i}\right|
&\leq\frac{1}{N}\sum_{i=Y_N}^{N} \left|e^{j2\pi k \rho Y_i}\right|=\frac{N-Y_N+1}{N}=1-\frac{Y_N}{N}+\frac{1}{N},
\end{align}
which implies that
\begin{equation}
\label{eq:eq_equidist_rewrite_condition_2}
\text{a.s.,}\quad\lim_{N\to\infty}\frac{1}{N}\sum_{i=Y_N+1}^{N} e^{j2\pi k \rho Y_i}=0.
\end{equation}

From \eqref{eq:eq_equidist_rewrite_condition}, \eqref{eq:eq_equidist_rewrite_condition_1} and \eqref{eq:eq_equidist_rewrite_condition_2}, we conclude that it is sufficient to show that

\begin{equation}
\text{a.s.,}\quad\lim_{N\to\infty}\frac{1}{Y_N}\sum_{i=0}^{Y_N-1} e^{j2\pi k \rho Y_i}=0.
\end{equation}

Notice that 
\begin{equation}
\frac{1}{Y_N}\sum_{i=0}^{Y_N-1} e^{j2\pi k \rho Y_i}=\frac{L_N}{S_{L_N}}\frac{1}{L_N}\sum_{l=1}^{L_N} X_l  e^{j2\pi k \rho S_{l-1}}.
\end{equation}

Now since $L_N\to\infty$ as $N\to\infty$, the strong law of large numbers implies that
\begin{equation}
\text{a.s.,}\quad\lim_{N\to\infty}\frac{L_N}{S_{L_N}}=\displaystyle\lim_{M\to\infty}\frac{M}{S_M}=\frac{1}{\E(X)}<\infty.
\end{equation}
Therefore, it is sufficient to show that
\begin{equation}
\label{eq:eq_slln}
\text{a.s.,}\quad\lim_{M\to\infty}\frac{\tilde{S}_M}{M}=0,
\end{equation}
where
\begin{equation}
\tilde{S}_M=\sum_{l=1}^{M} Z_l\quad\text{and}\quad Z_l=X_l  e^{j2\pi k \rho S_{l-1}}.
\end{equation}

For every $l\geq 1$, we have
\begin{equation}
\E(|Z_l|^2)=\E(X_l^2)=\E(X^2)<\infty.
\end{equation}

Furthermore, for every $l_1>l_2\geq 1$, we have
\begin{align}
\E(Z_{l_1}Z_{l_2}^{\ast})
&=\E\left(X_{l_1}  e^{j2\pi k \rho S_{l_1-1}} X_{l_2}  e^{-j2\pi k \rho S_{l_2-1}}\right)=\E\left(X_{l_1}X_{l_2}   e^{j2\pi k \rho (S_{l_1-1}-S_{l_2-1})}\right)\nn
&=\E\left(X_{l_1}X_{l_2} \prod_{l=l_2}^{l_1-1}   e^{j2\pi k \rho X_i}\right)=\E\left(X_{l_1}X_{l_2} e^{j2\pi k \rho X_{l_2}} \prod_{l=l_2+1}^{l_1-1}   e^{j2\pi k \rho X_l}\right)\nn
&=\E\left(X_{l_1}\right)\E\left(X_{l_2} e^{j2\pi k \rho X_{l_2}}\right) \prod_{l=l_2+1}^{l_1-1}   \E\left(e^{j2\pi k \rho X_l}\right)=\E(X)\E\left(Xe^{j2\pi k \rho X}\right) \E\left(e^{j2\pi k \rho X}\right)^{l_1-l_2-1}.\label{eq:eq_correlation_Z}
\end{align}
Thus,
\begin{align}
|\E(Z_{l_1}Z_{l_2}^{\ast})|
&=\E(X)\cdot\left|\E\left(Xe^{j2\pi k \rho X}\right)\right|\cdot \left|\E\left(e^{j2\pi k \rho X}\right)\right|^{l_1-l_2-1}\nn
&\leq\E(X)\cdot\E\left(\left|Xe^{j2\pi k \rho X}\right|\right)\cdot \left|\E\left(e^{j2\pi k \rho X}\right)\right|^{l_1-l_2-1}\nn
&=\E(X)\cdot\E\left(X\right)\cdot \left|\E\left(e^{j2\pi k \rho X}\right)\right|^{l_1-l_2-1}\nn
&=\E(X)^2\cdot \left|\E\left(e^{j2\pi k \rho X}\right)\right|^{l_1-l_2-1}.\label{eq:eq_correlation_ineq_Z}
\end{align}

Now since $X$ is nondeterministic and takes values in $\mathbb{N}^\ast$, there are two different integers $x_1,x_2\in \mathbb{N}^\ast$ such that $\mathbb{P}(X=x_1)>0$ and $\mathbb{P}(X=x_2)>0$. We have
\begin{align}
\E\left(e^{j2\pi k \rho X}\right) &= \mathbb{P}(X=x_1)e^{j2\pi k \rho x_1} + \mathbb{P}(X=x_2)e^{j2\pi k \rho x_2} + \sum_{x\in\mathbb{N}^*\setminus\{x_1,x_2\}}\mathbb{P}(X=x)e^{j2\pi k \rho x}\nn
&= e^{j2\pi k \rho x_1}\left(\mathbb{P}(X=x_1) + \mathbb{P}(X=x_2)e^{j2\pi k \rho (x_2-x_1)}\right) + \sum_{x\in\mathbb{N}^*\setminus\{x_1,x_2\}}\mathbb{P}(X=x)e^{j2\pi k \rho x},
\end{align}
which implies that
\begin{align}
\left|\E\left(e^{j2\pi k \rho X}\right)\right| &\leq \left|\mathbb{P}(X=x_1) + \mathbb{P}(X=x_2)e^{j2\pi k \rho (x_2-x_1)}\right| + \sum_{x\in\mathbb{N}^*\setminus\{x_1,x_2\}}\mathbb{P}(X=x_2)\nn
 &= \left|\mathbb{P}(X=x_1) + \mathbb{P}(X=x_2)e^{j2\pi k \rho (x_2-x_1)}\right| + \mathbb{P}(X\notin\{x_1,x_2\})\nn
 &= \mathbb{P}(X\in\{x_1,x_2\})\cdot\left|\frac{\mathbb{P}(X=x_1)}{\mathbb{P}(X\in\{x_1,x_2\})} + \frac{\mathbb{P}(X=x_2)}{\mathbb{P}(X\in\{x_1,x_2\})}e^{j2\pi k \rho (x_2-x_1)}\right| + \mathbb{P}(X\notin\{x_1,x_2\}).
 \label{eq:eq_correlation_ineq_EphoX}
\end{align}

Now since $\rho$ is irrational and $x_2-x_1$ is a nonzero integer, we have $e^{j2\pi k \rho (x_2-x_1)}\neq 1$, which means that
\begin{equation}
\frac{\mathbb{P}(X=x_1)}{\mathbb{P}(X\in\{x_1,x_2\})} + \frac{\mathbb{P}(X=x_2)}{\mathbb{P}(X\in\{x_1,x_2\})}e^{j2\pi k \rho (x_2-x_1)}
\end{equation}
is a convex combination between 1 and $e^{j2\pi k \rho (x_2-x_1)}\neq 1$. This implies that
\begin{equation}
\left|\frac{\mathbb{P}(X=x_1)}{\mathbb{P}(X\in\{x_1,x_2\})} + \frac{\mathbb{P}(X=x_2)}{\mathbb{P}(X\in\{x_1,x_2\})}e^{j2\pi k \rho (x_2-x_1)}\right|<1.
\end{equation}
By combining this with \eqref{eq:eq_correlation_ineq_EphoX}, we get
\begin{equation}
\label{eq:eq_correlation_ineq_EphoX2}
\left|\E\left(e^{j2\pi k \rho X}\right)\right| < 1.
\end{equation}

Now \eqref{eq:eq_correlation_ineq_EphoX2}, \eqref{eq:eq_correlation_ineq_Z} and \Cref{corr:corr_slln_exp} imply that
\begin{equation}
\text{a.s.},\quad\lim_{M\to\infty}\frac{\tilde{S}_M}{M}=0.
\end{equation}
}

\subsection{Proof of \Cref{lemma:lemma_diff_alph_couplingB}}
\label{subsec:subsec_lemma_diff_alph_couplingB_proof}

We need the following lemma:

\begin{lemma}
\label{lemma:lemma_coupling_geometric}
Let $0<\delta<\gamma\leq 1$. We can define three random variables $X,Y$ and $Z$ taking values in the set of natural numbers $\mathbb{N}$ such that:
\begin{itemize}
\item $X$ is geometrically distributed with success probability $\gamma$, i.e., $\Prob(X=i)=\gamma(1-\gamma)^{i-1}, \forall i\geq 1$.
\item $Z$ is independent of $X$.
\item $Y=X+Z$.
\item $Y$ is geometrically distributed with success probability $\delta$, i.e., $\Prob(Y=i)=\delta(1-\delta)^{i-1}, \forall i\geq 1$.
\end{itemize}
\end{lemma}
\begin{proof}
Let $X$ and $Z$ be two independent random variables such that
\begin{equation}
\Prob(X=i)=\gamma(1-\gamma)^{i-1}, \forall i\geq 1,
\end{equation}

and
\begin{equation}
		\Prob(Z=i)=\begin{cases}
			\frac{\delta}{\gamma}&\text{if }i=0,\\
			\frac{\gamma-\delta}{\gamma}(1-\delta)^{i-1}\delta&\text{if }i=1,2,3,\ldots\\
			0 & \text{otherwise}.
		\end{cases}
\end{equation}

The distribution of $Y$ is given by:
	\begin{align}
		\Prob(Y=i) &= \Prob(X+Z=i)\nn
			     &= \Prob(Z=i-X)\nn
			     &= \sum_{i'=1}^i \Prob(X=i')\Prob(Z=i-i'|X=i')\nn
			     &= \sum_{i'=1}^i (1-\gamma)^{i'-1}\gamma\lp\frac{\delta}{\gamma}\mathbbm{1}_{\{i-i'=0\}}+\frac{\gamma-\delta}{\gamma}(1-\delta)^{i-i'-1}\delta\mathbbm{1}_{\{i-i'\neq 0\}}\rp\nn
			     &= \delta(\gamma-\delta)(1-\delta)^{i-2}\sum_{i'=1}^{i-1}\lp\frac{1-\gamma}{1-\delta}\rp^{i'-1}+ \delta(1-\gamma)^{i-1}\nn
			     &= \delta(\gamma-\delta)(1-\delta)^{i-1}\lp\frac{1-\lp\frac{1-\gamma}{1-\delta}\rp^{i-1}}{\gamma-\delta}\rp+ \delta(1-\gamma)^{i-1}\nn
			     &= \delta(1-\delta)^{i-1}.
	\end{align}
\end{proof}

Now we are ready to prove \Cref{lemma:lemma_diff_alph_couplingB}

\begin{proof}[Proof of \Cref{lemma:lemma_diff_alph_couplingB}]
First notice that the probabilities $p_s=\frac{\bar{\epsilon}q^s(q^{k-s}-1)}{q^k-1}$ are decreasing in $s$, where $0\leq s\leq k-1$.
	This means that 
	\[ p_0\geq p_1\geq \cdots\geq p_{k-1}.\]

It follows from \Cref{lemma:lemma_coupling_geometric} that for each $0\leq s< k$, we can define five random variables: $\tilde{A}_s, J_s, A_s, \hat{J}_s$ and $\hat{A}_s$, such that:
\begin{itemize}
\item $\tilde{A}_s$ is geometrically distributed with success probability $p_0$, i.e., $\tilde{A}_s$ is distributed as $L_0$ and so $\Prob(\tilde{A}_s=i)=p_0(1-p_0)^{i-1}=\Prob(L_0=i), \forall i\geq 1$.
\item $J_s$ is independent of $\tilde{A}_s$.
\item $A_s=\tilde{A}_s+J_s$.
\item $A_s$ is geometrically distributed with success probability $p_s$, i.e., $A_s$ is distributed as $L_s$ and so $\Prob(A_s=i)=p_s(1-p_s)^{i-1}=\Prob(L_s=i), \forall i\geq 1$.
\item $\hat{J}_s$ is independent of $(\tilde{A}_s,J_s,A_s)$.
\item $\hat{A}_s=A_s+\hat{J}_s$.
\item $\hat{A}_s$ is geometrically distributed with success probability $p_{k-1}$, i.e., $\hat{A}_s$ is distributed as $L_{k-1}$ and so $\Prob(\hat{A}_s=i)=p_{k-1}(1-p_{k-1})^{i-1}=\Prob(L_{k-1}=i), \forall i\geq 1$.
\end{itemize}
Assume that $(\tilde{A}_s, J_s, A_s, \hat{J}_s,\hat{A}_s)$ is independent of $(\tilde{A}_{s'}, J_{s'
}, A_{s'}, \hat{J}_{s'},\hat{A}_{s'})$ if $s\neq s'$. Now define
\begin{equation}
\tilde{O}=\sum_{s=0}^{k-1} \tilde{A}_s,
\end{equation}
\begin{equation}
O=\sum_{s=0}^{k-1} A_s=\tilde{O}+\sum_{s=0}^{k-1} J_s,
\end{equation}
and
\begin{equation}
\hat{O}=\sum_{s=0}^{k-1} \hat{A}_s=O+\sum_{s=0}^{k-1} \hat{J}_s.
\end{equation}
Clearly, the distribution of $\tilde{O}, O$ and $\hat{O}$ is the same as that of $\tilde{B}, B$ and $\hat{B}$, respectively. Furthermore, we have $\tilde{O}\leq O\leq \hat{O}$ with probability 1.
\end{proof}



\bibliographystyle{IEEEtran}
%\bibliography{IEEEabrv,bibliography}
\bibliography{bibliography}

\end{document}